\documentclass[aps,prd,preprint,nofootinbib,superscriptaddress,11pt]{revtex4-2}

\usepackage[T1]{fontenc}
\usepackage[utf8]{inputenc}
\usepackage{lmodern}

\usepackage{amsmath}
\usepackage{amssymb}
\usepackage{amsfonts}
\usepackage{amsthm}

\usepackage{mathtools}
\usepackage{mathrsfs}
\usepackage{bm}
\usepackage{tikz}

\usepackage{graphicx}
\usepackage{xcolor}

\usepackage{booktabs}
\usepackage{multirow}

\usepackage[
colorlinks=true,
linkcolor=blue,
citecolor=blue,
urlcolor=blue
]{hyperref}

\newcommand{\dd}{\mathrm d}
\newcommand{\ii}{\mathrm i}
\newcommand{\ee}{\mathrm e}

\newcommand{\PV}{\mathrm{PV}}

\providecommand{\eth}{\text{\dh}}
\renewcommand{\eth}{\text{\dh}}

\begin{document}

\title{Quantum graviton scattering with definite helicities\\
in the null surface formulation. II}
\author{C.~N.~Kozameh}
	\email{carlos.kozameh@unc.edu.ar}
	\affiliation{FaMAF, Universidad Nacional de C\'ordoba,
		5000 C\'ordoba, Argentina}
	
	\author{G.~O.~Depaola}
	\affiliation{FaMAF, Universidad Nacional de C\'ordoba,
		5000 C\'ordoba, Argentina}
	
	\date{\today}

\begin{abstract}
We construct the third-order outgoing operator in the null-surface
formulation, retaining separately the nonlinear free-data contribution,
the perturbative metric contribution, and the deformation of the null
cones.  Only after the complete third-order structure has been identified,
and exclusively for the explicit \(2\to2\) calculation, we introduce a
Minkowskian approximation and define the corresponding operator
\(\delta a_{3,M}^{\rm out}\) by applying three simultaneous approximations
to the complete result.  It provides a direct benchmark against
conventional perturbative quantum gravity, whose asymptotic states,
propagators, contractions, and scattering kinematics are defined with
respect to a fixed Minkowski background.  We refer to this prescription as
the Minkowskian approximation.

Within this approximation, the calculation identifies a four-leg
tree-level contact term, to be added to the \(s,t,u\) exchange
contributions of Part~I, and a one-loop contribution generated by two
cubic outgoing corrections.  After subtraction of the lower-order BCH
iteration, the contact term is determined by the quartic generator.  For
the irreducible cubic--cubic loop sector, we give the Wick multiplicities
and mode normalization explicitly.  The ultraviolet tail is bounded by
\(\int^\infty dK/K^2\), without an ultraviolet counterterm, whereas the
bare radial infrared power count for a nonzero leading coefficient gives
\(\int_0 d\rho/\rho^3\).  Both statements refer to the specified
Minkowskian \((33)\) kernel, not to the complete one-loop amplitude.

The complete canonical outgoing operator also identifies an additional
contact-type \(2\to2\) graph: a two-leg insertion attached to a four-leg
vertex. This graph is connected but one-particle reducible and is not
evaluated in the present irreducible calculation. Its implications for
radiative transitions are the subject of the immediate continuation of
this work, using the Bondi total four-momentum balance, including the news
flux, to account for energy and momentum.
\end{abstract}

\maketitle
\section{Introduction}
\label{sec:introduction}

Perturbative gravitational scattering is conventionally described in terms of
fields propagating on a fixed Minkowski background.  The asymptotic particles
are the free radiative modes of the linearized gravitational field, and the
scattering amplitudes are extracted from the perturbative $\mathcal S$-matrix.
This description reproduces the standard tree-level graviton amplitudes and
organizes quantum general relativity as an effective field theory.  At
sufficiently high loop order, the bulk expansion requires
higher-derivative counterterms, reflecting the ultraviolet limitations of the
Einstein theory as a fundamental perturbative quantum field theory
\cite{GS86,Donoghue2022}.

Modern on-shell methods reorganize the same scattering information directly
in terms of physical momenta, helicities, factorization channels, and analytic
properties \cite{BCFW2005,Bedford2005,BrandhuberLoop2007,DunbarEttlePerkins2010}. These
methods provide a useful benchmark for a construction based entirely on
asymptotic radiative data.  The comparison made here is structural; no
identification with BCFW recursion is implied. We focus on the
requirements of its one-loop extensions and on their distinction from
the direct NSF loop calculation.

In general relativity the radiative degrees of freedom live naturally at null
infinity \cite{Bondi1962,Sachs1962,AshtekarStreubel1981,Ashtekar1987,Ashtekar2018}.  Their relation
to the interior geometry is encoded by null hypersurfaces.  The null-surface
formulation (NSF) makes this relation explicit: the spacetime geometry is
reconstructed from a family of cuts of null infinity,
\begin{equation}
Z(x^a,\zeta,\bar\zeta),
\label{eq:intro_cut_function}
\end{equation}
whose angular derivatives determine the associated optical variables and the
geometry of the null cones \cite{BKR2023,KZ2025}.  The Bondi shear at
$\mathscr I^\pm$ supplies the free radiative data entering the NSF field
equations.

The perturbative sector considered here consists of regular deformations of
the Minkowski cuts. We restrict the admissible classical geometries to be
regular throughout, with smooth and complete past and future null
infinities. In NSF the variables \(x^a\) in
\(Z(x^a,\zeta,\bar\zeta)\) label the four-dimensional solution space of the
cut equation; they are not coordinates on a spacetime background supplied
independently of the null-surface construction.  Spacetime geometry is
reconstructed on this solution space.  Classically, the radiative data are
restricted so that the reconstructed cuts remain smooth closed two-surfaces
with topology $S^2$, and the Bondi news has the decay required for finite
radiated energy \cite{KZ2025}.  After quantization, the shear modes are
operator-valued distributions.  Physical incoming and outgoing states are therefore
constructed by smearing the modes with smooth angular profiles and energy wave packets of rapid decay,
for example Schwartz-class profiles.  This state-space condition is imposed
from the outset and is not an ultraviolet regulator to be removed after the
calculation.

This work extends the perturbative quantum-scattering construction of
Ref.~\cite{KDI2026}, which evaluates the tree scattering sector constructed
from the quadratic outgoing correction and compares its helicity kernels
with conventional graviton scattering. Our main objectives are to identify
the four-leg contact contribution that completes the leading tree-level
organization in the perturbative-gravity comparison, and to construct
the irreducible cubic--cubic one-loop contribution and establish its
radial ultraviolet finiteness at the order considered. The contact term
must be added to the three exchange channels of Part~I; the loop is
generated by two third-order outgoing corrections and is ultraviolet
finite by the radial power counting developed below, without adding an
ultraviolet counterterm.

The operator construction also identifies a second, contact-type graph
in $2\to2$ scattering, consisting of a two-leg insertion attached to a
four-leg vertex. Although connected, it is one-particle reducible and
therefore distinct from the irreducible four-leg contact term. We identify
its topology here without evaluating its matrix element. Investigating
its physical implications is an immediate next step of the program.
The appropriate energy and momentum accounting for that investigation
is the Bondi total four-momentum balance
\eqref{eq:Bondi_four_momentum_balance_II}, with the BKR convention for
$l^a$: the four-momentum on an earlier cut equals that on a later cut
plus the radiated news four-momentum. The energy equation is its temporal
component, rather than a separate assumption about the energies of a
selected particle channel.

Sections~\ref{sec:perturbative_hierarchy}--\ref{sec:matching_and_delta_a}
first retain the general construction, including the nonlinear free data,
reconstructed metric, and perturbed null cones.  Only after these three
sectors have been identified at third order do we define the Minkowskian
approximation used for the explicit \(2\to2\) comparison with perturbative
quantum gravity and, after the Poincar\'e reduction, with Weinberg.  The
definition is given immediately before the scattering calculation in
Sec.~\ref{subsec:minkowskian_approximation_II}; no such restriction is
imposed on the preceding NSF construction.

This Minkowski kinematic specialization should not be confused with the
later Poincar\'e, or Weinberg, reduction.  The NSF boundary construction
produces conservation of the total spatial three-momentum, but it does not
impose conservation of the total Minkowski energy at the level of the
elementary outgoing operators.  Only after one restricts the external data
to the ordinary Poincar\'e scattering sector and imposes full
four-momentum conservation can the result be rewritten in terms of the
Mandelstam variables and compared directly with the standard \(s\)-,
\(t\)-, and \(u\)-channel amplitudes~\cite{Weinberg1965,KDI2026}.  Thus the
Minkowskian calculation performed here is deliberately broader than
the final Weinberg reduction.

The exact NSF equations are rational in the variables that determine the null surfaces. This property allows us to construct a recursive perturbative hierarchy for the cut function, the conformal factor, and the geometric sources integrated along the null cones. The bookkeeping parameter $\epsilon$ labels the perturbative order and is set to one after the hierarchy has been identified. The scattering map is obtained by comparing the future reconstruction with its antipodally related past counterpart. We write
\begin{equation}
a^{\mathrm{out}}_{\lambda}
=
a^{\mathrm{in}}_{\lambda}
+
\delta a_{\lambda},
\qquad
\delta a_{\lambda}
=
\sum_{n\geq2}
\epsilon^n
\delta a_{n,\lambda},
\label{eq:intro_operator_expansion}
\end{equation}
where the absence of a first-order correction follows from the vanishing of the linear cone source and the resulting trivial identification of the free incoming and outgoing radiative modes. At each nonlinear order, the outgoing correction is obtained by applying the radiative helicity projection to the complete relation between the future and antipodal-past reconstructions.

At third order, the cubic component of $\delta a^{\mathrm{out}}_{3,\lambda}$ contains three incoming radiative operators; the complete canonical correction also includes the linear Wick term described in Sec.~\ref{sec:unitarity}. Its connected contribution to four-graviton scattering has the form
\begin{equation}
\left\langle 0 \left|
\delta a^{\mathrm{out}}_{3,\lambda'_1}(p'_1)
\delta a^{\mathrm{out}}_{3,\lambda'_2}(p'_2)
a^{\mathrm{in}\dagger}_{\lambda_1}(p_1)
a^{\mathrm{in}\dagger}_{\lambda_2}(p_2)
\right| 0 \right\rangle_{\mathrm{conn}}.
\label{eq:intro_deltaa3_matrix_element}
\end{equation}
Each third-order insertion contains one outgoing external mode and three incoming radiative operators, and therefore acts as an effective quartic vertex. Two Wick contractions connect the two insertions through internal on-shell graviton modes. The spatial delta functions generated by the NSF vertices and by the Wick contractions combine to yield conservation of the total three-momentum of the scattering process. They also reduce the remaining integrations to one independent internal three-momentum. Thus, with two effective vertices and two internal contractions,
\begin{equation}
L
=
I-V+1
=
2-2+1
=
1,
\label{eq:intro_loop_counting}
\end{equation}
and this connected sector has one-loop topology.

At this stage the spatial delta functions enforce conservation of the total
three-momentum, but there is no independent delta function imposing
conservation of the total Minkowski energy.  This is a characteristic
feature of the asymptotic NSF operator map and should not be confused with
a failure of the calculation.  The ordinary four-dimensional conservation
law belongs to the subsequent Poincar\'e reduction of the external
scattering states.  As shown explicitly at tree level in
Ref.~\cite{KDI2026}, only after this reduction can the celestial kernels be
rewritten in terms of the usual Mandelstam invariants \(s\), \(t\), and
\(u\).

Throughout this paper we use the equivalent Cartesian-oscillator normalization specified below for the shear of
Part~III~\cite{KDIII2026}. The radial power counting needed here is
derived explicitly in Sec.~\ref{subsec:uv_power_counting_deltaa3}.
The same momentum-space operator is used in Cartesian and
frequency--angular coordinates, without a frequency-dependent rescaling.
Rewriting each radiative-mode integral in terms of the Lorentz-invariant
momentum measure introduces the weight
\begin{equation}
\nu(\omega)
=
\frac{\sqrt{2\pi G}}{\pi}\, \omega^{-1}.
\label{eq:intro_mode_weight}
\end{equation}
In the one-loop routing, the two internal modes therefore contribute a factor $K^{-2}$ to each third-order vertex. Since the corresponding reduced geometric kernel grows at most linearly,
\begin{equation}
\mathcal V^{(3),\mathrm{geom}}
=
O(K),
\end{equation}
each fully normalized vertex scales as $O(K^{-1})$. The radial part of the on-shell phase-space measure is $O(dK)$, and the ultraviolet tail of the product of the two vertices is consequently bounded by
\begin{equation}
\int^\infty \frac{dK}{K^2},
\label{eq:intro_uv_integral}
\end{equation}
which is convergent. Thus, the plane-wave kernel of this connected one-loop sector is radially ultraviolet finite in the uniform hard region. Smooth wave packets are nevertheless required to define the physical domain of the operator-valued radiative distributions. The generalization to Minkowskian two-vertex sectors at arbitrary order is developed in Part~III~\cite{KDIII2026}; the explicit calculation here concerns the one-loop sector.

The complete canonically ordered map used in the scattering calculation preserves, through the
order considered here,
the free-field canonical commutation relations, together with the adjoint
operation,
\begin{equation}
(\delta a)^\dagger=\delta(a^\dagger),
\label{eq:intro_adjoint_preservation}
\end{equation}
as discussed in Sec.~\ref{sec:unitarity}. The implication of these
properties for a unitary formulation and the explicit reconstruction of
its Hermitian generator are established in a separate BCH study, cited
there. That existence argument is not part of the present paper. The ultraviolet
estimate is independent of it and follows from
the explicit normalization and large-$K$ behavior described above.

We use the Minkowski null cone to extract flat-space scattering kernels.  A
more intrinsic nonlinear NSF
treatment would involve a generalized wave operator containing perturbative
corrections to the flat d'Alembertian.  Thus the higher-order cuts
$Z_2,Z_3,\ldots$ should not be interpreted as free waves of $\Box_0$; only the
linear cut $Z_1$ has that direct flat-space interpretation.

The paper is organized as follows.
Section~\ref{sec:perturbative_hierarchy} presents the exact NSF field
equations in rational form and develops the perturbative hierarchy for the
cut function, the angular field, the conformal factor, and the nonlinear
cone sources.

Section~\ref{sec:asymptotic_quantization} introduces the asymptotic
quantization of the Bondi shear, fixes the normalization and canonical
commutation relations of the radiative modes, and constructs the incoming
and outgoing quantum cuts.

Section~\ref{sec:matching_and_delta_a} imposes the geometric matching
condition between the advanced and retarded NSF reconstructions. It
develops the general cut-phase recursion, scalar flux, and third-order
matching, including the metric and null-generator perturbations.
Appendix~\ref{app:geodesics_from_cut_II} derives the generator displacement
from the cut function.

At the end of Sec.~\ref{sec:matching_and_delta_a},
Sec.~\ref{subsec:minkowskian_approximation_II} identifies the three relevant
sectors and defines \(\delta a_{3,M}^{\rm out}\) through the corresponding
Minkowskian approximations.  Section~\ref{sec:two_to_two_one_loop_uv}
then applies it to the perturbative-gravity comparison. Its
first subsection completes the tree-level organization with the four-leg
contact term reconstructed from
$\delta a_{3,M}^{\rm out}$. Its second subsection evaluates the cubic kernels and constructs the
connected irreducible \(2\to2\) contribution generated by
the corresponding pair of third-order operators, identifies its one-loop topology, and
establishes the radial ultraviolet convergence and infrared power count for
the fully normalized Minkowskian plane-wave kernel.

Section~\ref{sec:unitarity} states explicitly the outgoing canonical
commutation relations through $\delta a_3$ and the compatibility with
Hermitian conjugation. It then identifies the separate BCH result that
provides a unitary implementation and reconstructs its Hermitian generator.

Section~\ref{sec:on_shell_asymptotic_connection} compares the direct loop
calculation with published one-loop extensions of BCFW recursion,
identifying which complex-deformation requirements are absent from the
NSF calculation.

Section~\ref{sec:discussion_and_conclusions} summarizes the results, states
their domain of validity, and identifies the additional reducible
contact-type graph and its future analysis using Bondi four-momentum
balance.

The appendices contain the technical details of the third-order construction.
Appendix~\ref{app:sign_hermitian_tables} fixes the branch-sign convention,
Hermitian pairing, Bose reduction, and radiative on-shell map.
Appendix~\ref{app:explicit_C3_kernels} lists the explicit paired-cone
third-order kernels used in the Minkowskian approximation, with the two cone
contributions combined before the radiative and helicity projections.
Appendix~\ref{app:geodesics_from_cut_II} derives the perturbation of the
null generators retained in the general construction.

\section{Perturbative hierarchy of the NSF equations}
\label{sec:perturbative_hierarchy}

This section summarizes the perturbative NSF field equations used below
\cite{BKR2023,KZ2025}.  We work in the future sector; the past-sector equations
follow from the advanced conformal construction and antipodal identification.

The cut function is expanded as
\begin{equation}
Z(x^a,\zeta,\bar{\zeta})
=
Z_0(x^a,\zeta,\bar{\zeta})
+
\sum_{n\geq 1} \epsilon^n Z_n(x^a,\zeta,\bar{\zeta}),
\label{eq:Z_expansion}
\end{equation}
where
\begin{equation}
Z_0(x^a,\zeta,\bar{\zeta})
=
x^a l_a(\zeta,\bar{\zeta})
\label{eq:flat_cut}
\end{equation}
represents the flat-space cut. The Bondi shear, the angular variable 
$\Lambda = \eth^2 Z$, and the conformal factor $\Omega$ are expanded 
analogously:
\begin{align}
\sigma
&=
\sum_{n\geq 1} \epsilon^n \sigma_n,
&
\bar{\sigma}
&=
\sum_{n\geq 1} \epsilon^n \bar{\sigma}_n,
\label{eq:sigma_expansion}
\\
\Lambda
&=
\sum_{n\geq 1} \epsilon^n \Lambda_n,
&
\Omega
&=
1 + \sum_{n\geq 2} \epsilon^n \Omega_n .
\label{eq:Lambda_Omega_expansion}
\end{align}
Note that there is no first-order correction to $\Omega$, because the 
governing equation for the conformal factor is quadratic in $\Lambda$ at 
leading order.

\subsection{Exact equations}
\label{subsec:exact_nsf_equations}

The exact NSF equation for the cut function has the schematic form
\begin{equation}
\bar{\eth}^2 \eth^2 Z
=
\eth^2 \bar{\sigma}
+
\bar{\eth}^2 \sigma
-
\Sigma^+
-
\int_r^\infty
\left[
\Omega^{-2} \bar{\eth}(\eth\Omega^2)
+
h^{ab} \partial_a\Lambda \partial_b\bar{\Lambda}
\right] dr' .
\label{eq:exact_Z_equation}
\end{equation}
Here the quantity
\begin{equation}
\Sigma^+(x,\zeta,\bar{\zeta})
=
\int_{-\infty}^{Z(x,\zeta,\bar{\zeta})}
\dot{\sigma}(u,\zeta,\bar{\zeta})
\dot{\bar{\sigma}}(u,\zeta,\bar{\zeta}) \, du
\label{eq:Sigma_definition}
\end{equation}
is the standard quadratic energy-flux term evaluated on the cut.
The free boundary datum for the cut is the Tamburino--Winicour mass
aspect defined below. The cut equation includes the accumulated News flux
with the minus sign dictated by its evolution. In a Bondi frame at $\mathscr I^+$, introduce the real
supermomentum aspect in the Moreschi convention
\cite{MoreschiDain1998},
\begin{equation}
\Psi\equiv\psi_2^0+\sigma\dot{\bar\sigma}
+\eth^2\bar\sigma,
\qquad \Psi=\bar\Psi,
\label{eq:supermomentum_aspect_II}
\end{equation}
Here $\psi_2^0$ is the leading asymptotic Weyl scalar. In the
conventions of Part~I, the Bondi and Tamburino--Winicour aspects are
\[
m_B=-\operatorname{Re}(\psi_2^0+\sigma\dot{\bar\sigma}),
\qquad
m_{TW}=m_B+\operatorname{Re}(\eth^2\bar\sigma).
\]
Thus $\Psi=\operatorname{Re}(\eth^2\bar\sigma)-m_B
=2\operatorname{Re}(\eth^2\bar\sigma)-m_{TW}$.
The vacuum Bianchi identity gives
\begin{equation}
\begin{aligned}
\partial_u m_B&=\operatorname{Re}(\eth^2\dot{\bar\sigma})-|\dot\sigma|^2,\\
\partial_u m_{TW}&=2\operatorname{Re}(\eth^2\dot{\bar\sigma})-|\dot\sigma|^2,\\
\partial_u\Psi&=|\dot\sigma|^2 .
\end{aligned}
\label{eq:supermomentum_balance_II}
\end{equation}
Writing $[F]_{-\infty}^{Z}=F(Z(x,z),z)-F(-\infty,z)$,
the integrated boundary balance and its cut source are
\begin{equation}
\begin{aligned}
[m_{TW}]_{-\infty}^{Z}
&=2[\operatorname{Re}(\eth^2\bar\sigma)]_{-\infty}^{Z}-\Sigma^+(x,z),\\
\Sigma^+(x,z)&=[\Psi]_{-\infty}^{Z}.
\end{aligned}
\label{eq:Sigma_supermomentum_II}
\end{equation}
The angular derivative is taken at fixed $u$ before evaluation on the cut.
We use the BKR null-vector convention \cite{BKR2023},
\begin{equation}
l^a(z)=\frac{1}{\sqrt2}(1,\hat{\boldsymbol r}(z)),
\qquad |\hat{\boldsymbol r}|=1.
\label{eq:BKR_null_vector_normalization_II}
\end{equation}
For a real, $u$-independent angular weight $f(z)$, normalize the
supermomentum components consistently with the four translation
weights $f=l^a$:
\begin{equation}
P_f(u)=-\frac{1}{8\pi G}\oint_{S^2}\dd\Omega\,f(z)\Psi(u,z).
\label{eq:supermomentum_charge_II}
\end{equation}
Its balance law follows directly:
\begin{equation}
P_f(u_2)-P_f(u_1)
=-\frac{1}{8\pi G}\int_{u_1}^{u_2}\dd u
\oint_{S^2}\dd\Omega\,f(z)\dot\sigma\dot{\bar\sigma}.
\label{eq:supermomentum_charge_balance_II}
\end{equation}
In particular, these translation weights give $P^a=(M,\vec P)$ and
\begin{align}
P^a(u)&=-\frac{1}{8\pi G}\oint_{S^2}\dd\Omega\,l^a\Psi(u,z),
\nonumber\\
P^a(u_2)-P^a(u_1)
&=-\frac{1}{8\pi G}\int_{u_1}^{u_2}\dd u
\oint_{S^2}\dd\Omega\,l^a\dot\sigma\dot{\bar\sigma}
\equiv-P_{\rm rad}^a[u_1,u_2].
\label{eq:Bondi_four_momentum_balance_II}
\end{align}
Here $M=P^0$ denotes the Bondi energy in that frame, and the radiated
energy $P_{\rm rad}^0$ is nonnegative classically. With the unit-time-component
vector $L^a=\sqrt2\,l^a$, the same charge and flux have the
Newman--Penrose prefactor $1/(8\pi\sqrt2\,G)$; the change of prefactor
here only absorbs this $\sqrt2$ into the BKR vector and does not rescale
the shear or retarded time. On a constant-$u$ cut,
$\oint l^a\eth^2\bar\sigma\,\dd\Omega=0$, since $l^a$ has only
$\ell=0,1$ harmonics. Thus the four-momentum definition is equivalently
\begin{equation}
P^a(u)=-\frac{1}{8\pi G}\oint_{S^2}\dd\Omega\,
l^a\bigl(\psi_2^0+\sigma\dot{\bar\sigma}\bigr).
\label{eq:Bondi_four_momentum_NP_II}
\end{equation}
This is the four-momentum
projection of the supermomentum flux discussed in Part~I
\cite{KDI2026}; higher spherical harmonics resolve its remaining angular
components. These are balance laws between cuts, including the radiative flux,
not an assertion that the total charge including that flux fails to be
conserved. Their quantum implementation and application to additional
scattering channels will be studied separately; the operator flux
retained in the general construction is displayed in
Sec.~\ref{subsec:general_scalar_flux_II}.

We rewrite the radial integral in terms of the affine parameter $s$ along the
null geodesic generators of the cone. Using the 
relation
\begin{equation}
dr = \Omega^2 ds ,
\label{eq:dr_Omega_ds}
\end{equation}
the radial part of Eq.~\eqref{eq:exact_Z_equation} becomes
\begin{equation}
- \int_s^\infty
\left[
\bar{\eth}(\eth\Omega^2)
+
g^{ab} \partial_a\Lambda \partial_b\bar{\Lambda}
\right] ds' .
\label{eq:cone_source_affine}
\end{equation}
At order $n$, we define the generator integral as
\begin{equation}
I_n(x,z')
\equiv
- \left[
\int_s^\infty
\left(
\bar{\eth}(\eth\Omega^2)
+
g^{ab} \partial_a\Lambda \partial_b\bar{\Lambda}
\right) ds'
\right]_n .
\label{eq:In_definition}
\end{equation}
The quantity $I_n(x,z')$ integrates the $n$-th order source along a 
single null geodesic generator labeled by the celestial direction $z'$. 
The complete null-cone contribution to the cut function is obtained by 
performing an angular inversion on the sphere:
\begin{equation}
Z_{n,\text{cone}}(x,z)
=
\oint d^2z'\,
G_{0,0'}(z,z') I_n(x,z') .
\label{eq:Zn_cone_from_In}
\end{equation}
Thus, $I_n$ represents a generator-specific integral, whereas $Z_{n,\text{cone}}$ 
denotes the corresponding full null-cone contribution to the cut function.

The coefficients \(\sigma_n(u,z)\) in
Eq.~\eqref{eq:sigma_expansion} are defined at fixed Bondi time. The cut
hierarchy instead contains their pullback to \(u=Z(x,z)\). We denote its
coefficients by
\begin{equation}
\sigma_{[Z],n}(x,z)\equiv[\sigma(Z(x,z),z)]_{\epsilon^n}.
\label{eq:shear_pullback_definition_II}
\end{equation}
For example, at the classical level,
\begin{align}
\sigma_{[Z],1}&=\sigma_1(Z_0,z),\nonumber\\
\sigma_{[Z],2}&=\sigma_2(Z_0,z)+Z_1\dot\sigma_1(Z_0,z),\nonumber\\
\sigma_{[Z],3}&=\sigma_3(Z_0,z)+Z_1\dot\sigma_2(Z_0,z)
+Z_2\dot\sigma_1(Z_0,z)
+\tfrac12 Z_1^2\ddot\sigma_1(Z_0,z).
\label{eq:shear_pullback_third_order_II}
\end{align}
The conjugate pullback is defined analogously. Angular derivatives of
these composed fields are taken at fixed \(x\). After quantization, the
ordered cut-phase expansion in Secs.~\ref{sec:asymptotic_quantization}
and~\ref{sec:matching_and_delta_a} specifies the corresponding operator products.
These pullback terms and that phase expansion describe the same
composition and must not be counted twice.

In this notation, the order-$n$ equation for $Z_n$ reads
\begin{equation}
\bar{\eth}^2 \eth^2 Z_n
=
\eth^2 \overline{\sigma_{[Z],n}}
+
\bar{\eth}^2 \sigma_{[Z],n}
-
\Sigma_n^+
+
I_n ,
\label{eq:Zn_hierarchy}
\end{equation}
where $\Sigma_n^+$ denotes the $n$-th order component of the flux $\Sigma^+$.

We decompose the cut function as
\begin{equation}
Z_n
=
Z_{n,\text{cut}}
+
Z_{n,\text{cone}},
\label{eq:Zn_cut_cone_split}
\end{equation}
with the individual components given by
\begin{equation}
Z_{n,\text{cut}}(x,z)
=
\oint d^2z'\,
G_{0,0'}(z,z')
\left[
\eth'^2 \overline{\sigma_{[Z],n}}
+
\bar{\eth}'^2 \sigma_{[Z],n}
-
\Sigma_n^+
\right](x,z')
\label{eq:Zn_cut_part}
\end{equation}
and
\begin{equation}
Z_{n,\text{cone}}(x,z)
=
\oint d^2z'\,
G_{0,0'}(z,z') I_n(x,z') .
\label{eq:Zn_cone_part}
\end{equation}
The cut part contains the radiative data and the local flux term on the cut, 
whereas the cone part encodes the purely radial contributions integrated 
along the null generators, followed by the angular inversion.

\subsection{The conformal factor equation}
\label{subsec:Omega_hierarchy}

The conformal factor $\Omega$ satisfies the radial differential equation
\begin{equation}
2\partial_r^2\Omega
=
R_{rr}[h] \Omega ,
\label{eq:Omega_exact_equation}
\end{equation}
where the radial Ricci component reads
\begin{equation}
R_{rr}[h]
=
\frac{1}{4q}
\partial_r^2\Lambda \partial_r^2\bar{\Lambda}
+
\frac{3}{8q^2} (\partial_r q)^2
-
\frac{1}{4q} \partial_r^2 q ,
\label{eq:Rrr_exact}
\end{equation}
and we define the auxiliary metric function
\begin{equation}
q
=
1 - \partial_r\Lambda \partial_r\bar{\Lambda} .
\label{eq:q_definition}
\end{equation}
Since $R_{rr}$ is a rational function of $q$, its perturbative expansion 
is defined around the flat-space limit $q=1$.

We introduce the compact notation
\begin{align}
P_n
&\equiv
\sum_{i+j=n} \partial_r\Lambda_i \partial_r\bar{\Lambda}_j,
\label{eq:Pn_definition}
\\
A_n
&\equiv
\sum_{i+j=n} \partial_r^2\Lambda_i \partial_r^2\bar{\Lambda}_j ,
\label{eq:An_definition}
\end{align}
which allows us to expand $q$ as
\begin{equation}
q
=
1 - \sum_{n\geq 2} \epsilon^n P_n .
\label{eq:q_expansion}
\end{equation}

At second order, only the quadratic terms in the first-order fields contribute. 
The differential equation for $\Omega_2$ is
\begin{equation}
2\partial_r^2\Omega_2
=
\frac{1}{4}
\left(
A_2 + \partial_r^2 P_2
\right),
\label{eq:Omega2_differential}
\end{equation}
where
\begin{equation}
P_2
=
\partial_r\Lambda_1 \partial_r\bar{\Lambda}_1,
\qquad
A_2
=
\partial_r^2\Lambda_1 \partial_r^2\bar{\Lambda}_1 .
\label{eq:P2_A2_definition}
\end{equation}
Imposing standard boundary conditions at null infinity, the integrated form is
\begin{equation}
\Omega_2(r)
=
\frac{1}{8} P_2(r)
+
\frac{1}{8}
\int_r^\infty dr'
\int_{r'}^\infty dr''\,
A_2(r'') .
\label{eq:Omega2_integrated}
\end{equation}

At third order, the conformal-factor equation is
\begin{equation}
2\partial_r^2\Omega_3
=
\frac{1}{4}
\left(
A_3 + \partial_r^2 P_3
\right),
\label{eq:Omega3_differential}
\end{equation}
where the bilinear sources are
\begin{align}
P_3
&=
\partial_r\Lambda_1 \partial_r\bar{\Lambda}_2
+
\partial_r\Lambda_2 \partial_r\bar{\Lambda}_1,
\label{eq:P3_definition}
\\
A_3
&=
\partial_r^2\Lambda_1 \partial_r^2\bar{\Lambda}_2
+
\partial_r^2\Lambda_2 \partial_r^2\bar{\Lambda}_1 .
\label{eq:A3_definition}
\end{align}
No term proportional to $\Omega_2 (A_2 + \partial_r^2 P_2)$ appears at 
this order; such products are of fourth order and enter at the next level of the hierarchy.

At fourth order, the differential equation reads
\begin{align}
\begin{aligned}
2\partial_r^2\Omega_4
=
\frac{1}{4}
\left(
A_4 + \partial_r^2 P_4
\right)
& +
\frac{1}{4} P_2
\left(
A_2 + \partial_r^2 P_2
\right)
+
\frac{3}{8}
(\partial_r P_2)^2
\\
& +
\frac{1}{4}
\left(
A_2 + \partial_r^2 P_2
\right) \Omega_2 .
\end{aligned}
\label{eq:Omega4_differential}
\end{align}
Here the fourth-order sources are
\begin{align}
P_4
&=
\partial_r\Lambda_1 \partial_r\bar{\Lambda}_3
+
\partial_r\Lambda_2 \partial_r\bar{\Lambda}_2
+
\partial_r\Lambda_3 \partial_r\bar{\Lambda}_1,
\label{eq:P4_definition}
\\
A_4
&=
\partial_r^2\Lambda_1 \partial_r^2\bar{\Lambda}_3
+
\partial_r^2\Lambda_2 \partial_r^2\bar{\Lambda}_2
+
\partial_r^2\Lambda_3 \partial_r^2\bar{\Lambda}_1 .
\label{eq:A4_definition}
\end{align}
The final line in Eq.~\eqref{eq:Omega4_differential} originates from the 
product $R_{rr}^{(2)}\Omega_2$, constituting a genuinely fourth-order backreaction.

The integrated form at any order $n$ may be written compactly as
\begin{equation}
\Omega_n(r)
=
\frac{1}{2}
\int_r^\infty dr'
\int_{r'}^\infty dr''\,
\left[R_{rr} \Omega\right]_n(r'') ,
\label{eq:Omega_n_integrated_general}
\end{equation}
evaluated with the order $n=2,3,4$ sources defined above.

\subsection{Order-by-order cut equation}
\label{subsec:order_by_order_cut_equation}

Combining the cut equation with our cone-source notation, the first few 
perturbative orders take the following explicit forms.

At first order, the cone contribution vanishes identically:
\begin{equation}
I_1 = 0,
\label{eq:I1_zero}
\end{equation}
which yields
\begin{equation}
\bar{\eth}^2 \eth^2 Z_1
=
\eth^2 \overline{\sigma_{[Z],1}}
+
\bar{\eth}^2 \sigma_{[Z],1} .
\label{eq:Z1_equation}
\end{equation}
This is the standard linear radiative relation between the cut and the free shear.

At second order, we have
\begin{equation}
\bar{\eth}^2 \eth^2 Z_2
=
\eth^2 \overline{\sigma_{[Z],2}}
+
\bar{\eth}^2 \sigma_{[Z],2}
-
\Sigma_2^+
+
I_2 ,
\label{eq:Z2_equation}
\end{equation}
where the flux term is
\begin{equation}
\Sigma_2^+
=
\int_{-\infty}^{Z_0}
\dot{\sigma}_1 \dot{\bar{\sigma}}_1 \, du .
\label{eq:Sigma2_definition}
\end{equation}
The term $I_2$ represents the second-order generator integral produced by 
$\Omega_2$ and the quadratic combination $g^{ab} \partial_a\Lambda_1 \partial_b\bar{\Lambda}_1$.

At third order, the equation is
\begin{equation}
\bar{\eth}^2 \eth^2 Z_3
=
\eth^2 \overline{\sigma_{[Z],3}}
+
\bar{\eth}^2 \sigma_{[Z],3}
-
\Sigma_3^+
+
I_3 .
\label{eq:Z3_equation}
\end{equation}
The source $I_3$ contains the connected third-order cone contribution, built 
from $\Omega_3$ and the cross-terms
\begin{equation*}
\partial_a\Lambda_1 \partial_b\bar{\Lambda}_2
+
\partial_a\Lambda_2 \partial_b\bar{\Lambda}_1 .
\end{equation*}
This is the component that subsequently generates the connected cubic correction 
$\delta a^\text{out}_3$ after helicity projection.

At fourth order, the cut equation is
\begin{equation}
\bar{\eth}^2 \eth^2 Z_4
=
\eth^2 \overline{\sigma_{[Z],4}}
+
\bar{\eth}^2 \sigma_{[Z],4}
-
\Sigma_4^+
+
I_4 .
\label{eq:Z4_equation}
\end{equation}
The source $I_4$ contains the fourth-order cone contributions, including backreaction 
products involving $\Omega_4$, the quadratic combinations $\Lambda_1\Lambda_3$ and 
$\Lambda_2\Lambda_2$, and the higher-order corrections inherited from the rational 
expansion of $R_{rr}$.

\subsection{Role in the scattering construction}
\label{subsec:role_in_scattering_construction}

The hierarchy separates the free radiative data on the cut from nonlinear
propagation through the interior of the null cone. 
The asymptotic free data enter through their composed coefficients
$\sigma_{[Z],n}$ and $\overline{\sigma_{[Z],n}}$. 
Conversely, the nonlinear radial propagation along the null generators is 
encoded in the integrals $I_n$. The physical cone contribution to the cut 
function is not $I_n$ itself, but its angularly inverted counterpart:
\begin{equation}
Z_{n,\text{cone}}(x,z)
=
\oint d^2z'\,
G_{0,0'}(z,z') I_n(x,z') .
\end{equation}

This distinction is required in the quantization procedure. The third-order 
connected scattering kernel is extracted from the connected part of $I_3$, 
followed by angular inversion and radiative helicity extraction.  The
subsequent sections use this procedure to obtain the reduced cubic vertices
$\mathcal{V}^{(3),\lambda,\mathrm{geom}}_{\eta_2\eta_3}$.  The
vertices entering the operator and the scattering amplitude are the
mode-normalized quantities $\mathcal W^{(3)}$, defined below.

\section{Asymptotic quantization and quantum cuts}
\label{sec:asymptotic_quantization}
The perturbative NSF hierarchy expresses the spacetime geometry as a
functional of the radiative data at null infinity. We therefore begin by
recalling the asymptotic quantization of the Bondi shear and then apply it
to the cut functions entering the NSF reconstruction.

To simplify the notation, $z$ denotes a point on the celestial sphere,
and $z'$ denotes the angular integration variable. The conjugate
stereographic coordinates will not be displayed explicitly. We also use
the spherical parametrization of a future-directed null momentum,
\begin{equation}
k = (\omega, z),
\qquad
k^a = \sqrt2\,\omega l^a(z),
\qquad
\omega > 0.
\label{eq:spherical_null_momentum_quantization}
\end{equation}
We label each Cartesian oscillator by the energy $\omega=|\vec k|$.
For $u=x\cdot\ell$, $\ell^0=1/\sqrt2$, put $\varpi=\sqrt2\,\omega$
(and $\varpi_j=\sqrt2\,\omega_j$). Then $k^a=\varpi\ell^a$.
Relative to Part~III we use $a_\lambda=\sqrt2\,a_\lambda^{\rm III}$,
$d\mu=d^3k/(2\omega)$ and the spatial distribution $\delta^{(3)}$.
Thus $\int d\mu(k')\,2\omega\delta^{(3)}(\vec k-\vec k')f(k')=f(k)$.
With $N=\sqrt{2\pi G}$ and $\nu(\omega)=N/(\pi\omega)$,
$d\mu\,\nu=N\,d\omega\,d\Omega/(2\pi)$. The two helicity contractions give
\[
[\sigma(u,\Omega),\bar\sigma(u',\Omega')]
=-\frac{2iG}{\pi}\delta_{S^2}^{(2)}(\Omega,\Omega')
\int_0^\infty\frac{d\omega}{\omega}\sin[\sqrt2\,\omega(u-u')]
=-iG\,{\rm sgn}(u-u')\delta_{S^2}^{(2)}(\Omega,\Omega').
\]

We quantize the asymptotic radiative data using the Ashtekar--Streubel
radiative phase space and Ashtekar's asymptotic quantization
\cite{AshtekarStreubel1981,Ashtekar1987}, using throughout the physical shear commutator of Part~III~\cite{KDIII2026}. In these conventions, the
first-order shear at $\mathscr{I}^+$ is promoted to the operator
\begin{align}
\widehat{\sigma}_1^+(u,z)
={}&
\int_0^\infty
\frac{d\omega}{2\pi}
\sqrt{2\pi G}
\left[
a_{1,+}^{\rm in}(\omega,z) e^{-i\varpi u}
+
a_{1,-}^{{\rm in}\dagger}(\omega,z) e^{+i\varpi u}
\right],
\label{eq:quantum_sigma_one} \\
\widehat{\sigma}_1^{+\dagger}(u,z)
={}&
\int_0^\infty
\frac{d\omega}{2\pi}
\sqrt{2\pi G}
\left[
a_{1,-}^{\rm in}(\omega,z) e^{-i\varpi u}
+
a_{1,+}^{{\rm in}\dagger}(\omega,z) e^{+i\varpi u}
\right].
\label{eq:quantum_sigma_one_dagger}
\end{align}
Thus $a_{1,+}^{\rm in}$ and $a_{1,-}^{\rm in}$ annihilate
incoming gravitons of helicity $+2$ and $-2$, respectively.
The frequency--angular notation means evaluation of the same operator,
\begin{equation}
a_{1,\lambda}^{\rm in}(\omega,z)
\equiv a_{1,\lambda}^{\rm in}
\bigl(\vec k=\omega\widehat k(z)\bigr),
\label{eq:operator_coordinate_convention_II}
\end{equation}
and does not introduce a rescaled oscillator. Here $\widehat k(z)$ is the
unit spatial direction and $d^2z$ denotes the solid-angle measure used in
$d^3k=\omega^2d\omega\,d^2z$.
The nonvanishing canonical commutator is
\begin{equation}
\left[
a_{1,\lambda}^{\rm in}(\vec{k}),
a_{1,\lambda'}^{{\rm in}\dagger}(\vec{k}')
\right]
=
2\omega_k
\delta_{\lambda\lambda'}
\delta^{(3)}(\vec{k}-\vec{k}'),
\qquad
\omega_k = |\vec{k}|,
\label{eq:asymptotic_quantization_commutator}
\end{equation}
and all commutators between two annihilation operators or two creation
operators vanish.
In particular, Eq.~\eqref{eq:asymptotic_quantization_commutator} is the
Cartesian-momentum CCR used in every Wick contraction below. A change to
frequency--angular coordinates transforms its delta function with the
Jacobian $\omega^{-2}$; it does not change the operator convention in
Eq.~\eqref{eq:operator_coordinate_convention_II}. No additional oscillator
rescaling is made when the mode integrals are rewritten with $d\mu(k)$.

The perturbative parameter $\epsilon$ is not part of the canonical
normalization. It only keeps track of the amplitude order. The physical
incoming shear is therefore
\begin{equation}
\widehat{\sigma}_{\rm in}^+
=
\epsilon\widehat{\sigma}_1^+,
\qquad
\widehat{\sigma}_{\rm in}^{+\dagger}
=
\epsilon\widehat{\sigma}_1^{+\dagger}.
\label{eq:physical_incoming_shear}
\end{equation}
The same helicity decomposition defines the outgoing shear:
\begin{align}
\widehat{\sigma}_{\rm out}^+(u,z)
={}&
\int_0^\infty
\frac{d\omega}{2\pi}
\sqrt{2\pi G}
\left[
a_+^{\rm out}(\omega,z) e^{-i\varpi u}
+
a_-^{{\rm out}\dagger}(\omega,z) e^{+i\varpi u}
\right],
\label{eq:quantum_outgoing_shear} \\
\widehat{\sigma}_{\rm out}^{+\dagger}(u,z)
={}&
\int_0^\infty
\frac{d\omega}{2\pi}
\sqrt{2\pi G}
\left[
a_-^{\rm out}(\omega,z) e^{-i\varpi u}
+
a_+^{{\rm out}\dagger}(\omega,z) e^{+i\varpi u}
\right].
\label{eq:quantum_outgoing_shear_dagger}
\end{align}
The distinction between the incoming and outgoing operators is dynamical.
At linear order they are identified by the antipodal matching condition,
whereas their nonlinear difference is determined by the higher-order NSF
equations.

\subsection{Quantum cuts}
\label{subsec:quantum_cuts}
We first quantize the shear-dependent cut contribution, denoted by
\(Z_{\rm cut,\sigma}\), without expanding the cut in its argument.
The full asymptotic cut also contains the scalar-flux term \(\Sigma\),
whose perturbative contribution is developed in
Sec.~\ref{subsec:general_scalar_flux_II}. The cut used to evaluate the
shear and its phase is the full cut field.

\paragraph{Exact shear-dependent outgoing contribution.}
At the classical level, this contribution is
\begin{align}
Z_{\rm cut,\sigma}^{\rm out}(x,z)
={}&
\oint d^2z'
\Big[
G_{0,-2'}
\sigma_{\rm out}^+
\bigl(Z^{\rm out}(x,z'),z'\bigr)
\nonumber \\
&\hspace{2.4cm}
+
G_{0,2'}
\sigma_{\rm out}^{+*}
\bigl(Z^{\rm out}(x,z'),z'\bigr)
\Big].
\label{eq:exact_outgoing_cut_shear_form}
\end{align}
The angular arguments $(z,z')$ of the Green functions are understood.
The function $Z^{\rm out}(x,z')$ appearing in the argument of the shear
is the complete outgoing cut.

Using the Fourier expansion of the outgoing shear,
Eq.~\eqref{eq:exact_outgoing_cut_shear_form} becomes
\begin{align}
Z_{\rm cut,\sigma}^{\rm out}(x,z)
={}&
\oint d^2z'
\int_0^\infty
\frac{d\omega}{2\pi}
\sqrt{2\pi G}
\Bigg\{
e^{-i\varpi Z^{\rm out}(x,z')}
\Big[
G_{0,-2'} a_+^{\rm out}(k)
+
G_{0,2'} a_-^{\rm out}(k)
\Big]
\nonumber \\
&\hspace{2.1cm}
+
e^{+i\varpi Z^{\rm out}(x,z')}
\Big[
G_{0,-2'} a_-^{\rm out\dagger}(k)
+
G_{0,2'} a_+^{\rm out\dagger}(k)
\Big]
\Bigg\},
\qquad
k = (\omega,z').
\label{eq:exact_classical_outgoing_cut}
\end{align}
No expansion of $Z^{\rm out}$ or of the outgoing amplitudes has been
made in this equation.

We denote the full outgoing quantum cut by
$\widehat{Z}_{\rm total}^{\rm out}$ and separate its flat Minkowski part
from its radiative contribution according to
\begin{equation}
\widehat{Z}_{\rm total}^{\rm out}(x,z')
=
x^a l_a(z') \mathbf{1}
+
\widehat{Z}_{\rm rad}^{\rm out}(x,z').
\label{eq:outgoing_quantum_cut_split}
\end{equation}
The operator $\widehat{Z}_{\rm rad}^{\rm out}$ contains the complete
radiative deformation of the outgoing cut. It is constructed from the
asymptotic on-shell operators
$a_\lambda^{\rm out}$ and
$a_\lambda^{\rm out\dagger}$, although this functional dependence
will be left implicit. Since the complete cut is real and the flat
Minkowski contribution is real, the radiative cut operator is
self-adjoint,
\begin{equation}
\widehat{Z}_{\rm rad}^{\rm out\dagger}
=
\widehat{Z}_{\rm rad}^{\rm out}.
\label{eq:outgoing_radiative_cut_hermiticity}
\end{equation}
Because $x^a l_a(z')$ is a $c$-number, the Fourier phase factorizes
exactly:
\begin{equation}
e^{-i\varpi\widehat{Z}_{\rm total}^{\rm out}(x,z')}
=
e^{-i\varpi x\cdot l'}
\widehat{U}_\omega^{\rm out}(x,z'),
\label{eq:outgoing_quantum_phase_factorization}
\end{equation}
where
\begin{equation}
\widehat{U}_\omega^{\rm out}(x,z')
\equiv
\exp\left[
-i\varpi\widehat{Z}_{\rm rad}^{\rm out}(x,z')
\right].
\label{eq:outgoing_unitary_definition}
\end{equation}
The self-adjointness of
$\widehat{Z}_{\rm rad}^{\rm out}$ implies
\begin{equation}
\widehat{U}_\omega^{\rm out\dagger}
\widehat{U}_\omega^{\rm out}
=
\widehat{U}_\omega^{\rm out}
\widehat{U}_\omega^{\rm out\dagger}
=
\mathbf{1}.
\label{eq:outgoing_unitary_identity}
\end{equation}

The quantum shear-dependent contribution is defined by promoting each
complete Fourier product together with its Hermitian conjugate:
\begin{align}
\widehat{Z}_{\rm cut,\sigma}^{\rm out}(x,z)
={}&
\oint d^2z'
\int_0^\infty
\frac{d\omega}{2\pi}
\sqrt{2\pi G}
:
\Bigg\{
e^{-i\varpi x\cdot l'}
\Big[
G_{0,-2'} a_+^{\rm out}(k)
+
G_{0,2'} a_-^{\rm out}(k)
\Big]
\widehat{U}_\omega^{\rm out}
\nonumber \\
&\hspace{1.7cm}
+
e^{+i\varpi x\cdot l'}
\widehat{U}_\omega^{{\rm out}\dagger}
\Big[
G_{0,-2'} a_-^{{\rm out}\dagger}(k)
+
G_{0,2'} a_+^{{\rm out}\dagger}(k)
\Big]
\Bigg\}
: .
\label{eq:exact_quantum_outgoing_cut}
\end{align}
The arguments $(x,z')$ of the unitary operators are understood.
The operator ordering in
Eq.~\eqref{eq:exact_quantum_outgoing_cut} follows from
\begin{equation}
\left(
a_\lambda^{\rm out}\widehat{U}_\omega^{\rm out}
\right)^\dagger
=
\widehat{U}_\omega^{{\rm out}\dagger}
a_\lambda^{{\rm out}\dagger}.
\label{eq:Hermitian_ordering_outgoing_cut}
\end{equation}
Normal ordering is applied only after the complete annihilation term and
its Hermitian conjugate have been assembled. It is part of the definition
of the quantum cut.

\paragraph{Incoming cut.}
The incoming cut is quantized by the same prescription. After applying
the antipodal transformation, its flat phase can be written in the same
form as the outgoing one. Its radiative deformation is a nonlinear
functional of the incoming data and has the perturbative expansion
\begin{equation}
\widehat{\mathcal Z}^{-}
=
\sum_{j\geq1}\epsilon^j\widehat Z_j,
\qquad
\widehat{\mathcal Z}^{-\dagger}
=
\widehat{\mathcal Z}^{-}.
\label{eq:incoming_nonlinear_cut_deformation}
\end{equation}
We define the corresponding unitary phase operator by
\begin{equation}
\boxed{
\widehat U_\omega^{-}
\equiv
\exp[-i\varpi\widehat{\mathcal Z}^{-}],
\qquad
\widehat U_\omega^{-\dagger}
=
\exp[+i\varpi\widehat{\mathcal Z}^{-}].
}
\label{eq:incoming_unitary_definition}
\end{equation}
The coefficients $\widehat Z_j$ are determined recursively from the
incoming radiative data; they are not additional asymptotic degrees of
freedom.

The antipodally transformed incoming shear-dependent contribution has the operator form
\begin{align}
\mathcal{A}\widehat{Z}_{\rm cut,\sigma}^{\rm in}(x,z)
={}&
\oint d^2z'
\int_0^\infty
\frac{d\omega}{2\pi}
\sqrt{2\pi G}
:
\Bigg\{
e^{-i\varpi x\cdot l'}
\Big[
G_{0,-2'}\epsilon a_{1,+}^{\rm in}(k)
+
G_{0,2'}\epsilon a_{1,-}^{\rm in}(k)
\Big]
\widehat U_\omega^{-}
\nonumber \\
&\hspace{1.7cm}
+
e^{+i\varpi x\cdot l'}
\widehat U_\omega^{-\dagger}
\Big[
G_{0,-2'}\epsilon a_{1,-}^{{\rm in}\dagger}(k)
+
G_{0,2'}\epsilon a_{1,+}^{{\rm in}\dagger}(k)
\Big]
\Bigg\}
: .
\label{eq:exact_quantum_incoming_cut}
\end{align}
Here $\mathcal{A}$ denotes the antipodal transformation, including the
phase convention required by the linear matching relation.

\paragraph{Matched shear-dependent quantum contribution.}
Using the geometric cut matching developed in
Sec.~\ref{subsec:geometric_quantum_matching}, the outgoing and
antipodally mapped incoming branches are expressed in a common cut
orientation. In the corresponding positive-frequency convention, the
common phase operator in the matched expression is
\begin{equation}
\boxed{
\left[\widehat U_\omega^{\rm out}\right]_{\rm matched}
=
\widehat U_\omega^{-},
}
\label{eq:common_exact_cut_unitary}
\end{equation}
with
\begin{equation}
\widehat U_\omega^{-\dagger}\widehat U_\omega^{-}
=
\widehat U_\omega^{-}\widehat U_\omega^{-\dagger}
=
\mathbf 1.
\label{eq:common_exact_cut_unitary_definition}
\end{equation}
The equality in Eq.~\eqref{eq:common_exact_cut_unitary} uses the
geometric matching, including the cancellation of lower-order terms.
The spin weights of \(\mathcal A\bar\sigma^-\) and \(\sigma^+\) agree,
so their angular kernels can be compared in this common convention.
The adjoint phase occurs in the negative-frequency branch.

We define the exact change of the asymptotic operators by
\begin{equation}
\boxed{
\delta a_\lambda(k)
\equiv
a_\lambda^{\rm out}(k)-a_\lambda^{\rm in}(k).
}
\label{eq:exact_delta_a_definition}
\end{equation}
The linear matching proved in Appendix~A of the first paper gives
$\delta a_{1,\lambda}=0$. Hence the first nonvanishing term is of
second order,
\begin{equation}
\delta a_\lambda(k)
=
\epsilon^2\delta a_{2,\lambda}^{\rm out}(k)
+
\epsilon^3\delta a_{3,\lambda}^{\rm out}(k)
+
\mathcal O(\epsilon^4),
\label{eq:delta_a_expansion_from_second_order}
\end{equation}
while $a_\lambda^{\rm in}=\epsilon a_{1,\lambda}^{\rm in}$ in the
absolute-order convention used here.

Using Eqs.~\eqref{eq:exact_quantum_outgoing_cut},
\eqref{eq:exact_quantum_incoming_cut}, and
\eqref{eq:common_exact_cut_unitary}, the matched shear-dependent contribution is
\begin{equation}
\boxed{
\begin{aligned}
\Delta\widehat Z_{\rm cut,\sigma}(x,z)
={}&
\oint d^2z'\int_0^\infty
\frac{d\omega}{2\pi}
\sqrt{2\pi G}
:\!\Bigg\{
 e^{-i\varpi x\cdot l'}
 \Big[
 G_{0,-2'}(z,z')\,\delta a_+(\omega,z')
 +
 G_{0,2'}(z,z')\,\delta a_-(\omega,z')
 \Big]
 \widehat U_\omega^{-}(x,z')
 +\mathrm{h.c.}
\Bigg\}\!:\, .
\end{aligned}
}
\label{eq:exact_matched_quantum_cut_delta_a}
\end{equation}
Here the Hermitian conjugate acts on the complete preceding
positive-frequency term and therefore reverses the operator order. The
flat Minkowski phase $e^{-i\varpi x\cdot l'}$ has been taken outside the
unitary operator, since $x\cdot l'$ is a $c$-number. Equation
\eqref{eq:exact_matched_quantum_cut_delta_a} contains the shear-dependent cut contribution, including its full
ordered phase. The general asymptotic cut also contains the nonlinear
scalar flux \(\Sigma\), treated separately in
Sec.~\ref{subsec:general_scalar_flux_II}.

If both the cone contribution and the scalar-flux difference vanish,
the matching equation reduces to $\Delta\widehat Z_{\rm cut,\sigma}=0$. Its radiative projection then gives
$\delta a_\lambda=0$, or $a_\lambda^{\rm out}=a_\lambda^{\rm in}$.
When the cone is nonzero, the same cut expression is one side of the
full cut--cone equation and determines $\delta a_\lambda$ recursively.

For later use, expansion of the single exact product gives
\begin{align}
\delta a_\lambda\widehat U_\omega^{-}
={}&
\epsilon^2\delta a_{2,\lambda}^{\rm out}
\nonumber\\
&+
\epsilon^3\left(
\delta a_{3,\lambda}^{\rm out}
-i\varpi\delta a_{2,\lambda}^{\rm out}\widehat Z_1
\right)
+
\mathcal O(\epsilon^4).
\label{eq:delta_a_U_expansion_through_third}
\end{align}
Thus every recursive cut term is generated by expanding
$\delta a_\lambda\widehat U_\omega^{-}$. Nonlinear cut coefficients
$\widehat Z_2,\widehat Z_3,\ldots$ enter through this cumulative phase at
the perturbative orders at which they can contribute.

\section{Matching condition and higher-order outgoing operators}
\label{sec:matching_and_delta_a}
The preceding section constructed the incoming and outgoing quantum cuts
directly from the asymptotic radiative operators. We next incorporate the
nonlinear propagation along the null cones and impose the condition that
the advanced and retarded NSF reconstructions describe the same physical
spacetime geometry.

\subsection{Geometric origin of the matching condition}
\label{subsec:geometric_quantum_matching}
The matching condition does not originate from an independent
identification of the incoming and outgoing cuts. Its geometrical origin
is the requirement that the advanced and retarded NSF reconstructions
determine the same spacetime metric.

The outgoing radiative data on $\mathscr{I}^+$ determine the advanced
reconstruction,
\begin{equation}
g_{ab}^{\rm adv}
=
g_{ab}^{\rm adv}
\left[
Z_{\rm total}^{\rm out}
\right],
\label{eq:advanced_metric_from_total_cut}
\end{equation}
whereas the incoming radiative data on $\mathscr{I}^-$ determine the
retarded reconstruction,
\begin{equation}
g_{ab}^{\rm ret}
=
g_{ab}^{\rm ret}
\left[
Z_{\rm total}^{\rm in}
\right].
\label{eq:retarded_metric_from_total_cut}
\end{equation}
The terminology ``advanced'' and ``retarded'' refers here to the
Green-function reconstruction of the interior metric. Thus the data on
$\mathscr{I}^+$, although written in terms of the retarded Bondi time
$u$, determine the advanced solution, while the data on
$\mathscr{I}^-$, written in terms of the advanced Bondi time $v$,
determine the retarded solution.

Before the two reconstructions can be compared, the angular variables
of $\mathscr{I}^-$ must be mapped to those of $\mathscr{I}^+$. We
denote this antipodal transformation by $\mathcal{A}$:
\begin{equation}
\mathcal{A} Z_{\rm total}^{\rm in}(x,z)
\equiv
Z_{\rm total}^{\rm in}(x,\widehat{z}),
\label{eq:antipodal_incoming_total_cut}
\end{equation}
where $\widehat{z}$ is the antipodal point of $z$. When
$\mathcal{A}$ acts on the radiative modes contained in the cut, the
corresponding helicity and phase conventions are understood.

The fundamental geometrical matching condition is
\begin{equation}
\boxed{
g_{ab}^{\rm adv}
\left[
Z_{\rm total}^{\rm out}
\right]
=
g_{ab}^{\rm ret}
\left[
\mathcal{A} Z_{\rm total}^{\rm in}
\right].
}
\label{eq:advanced_retarded_metric_matching}
\end{equation}
The advanced and retarded solutions are therefore not two independent
metrics, but two reconstructions of the same spacetime geometry from
future and past asymptotic data.

\paragraph{Cancellation of lower-order terms at third order.}
At second order, Ref.~\cite{BKR2023}, Sec.~IV.B, obtains the cut
matching equation by equating the reconstructed metrics and using the
first-order matching to cancel their common lower-order terms. We make
the corresponding third-order step explicit here.
Following that perturbative convention, write
\begin{equation}
g^{ab}_{\alpha}=\Omega_\alpha^2
\left(\eta^{ab}+\sum_{n\geq1}\epsilon^n h^{ab}_{n,\alpha}\right),
\qquad \alpha=\mathrm{adv},\mathrm{ret},
\label{eq:metric_expansion_matching_II}
\end{equation}
with the perturbative indices on \(h_{n\,ab}\) raised and lowered using
\(\eta_{ab}\). The conformal factor drops out of the null-surface
equation. Put the cuts in a common orientation by defining
\begin{equation}
\widetilde Z^{\rm adv}=Z_{\rm total}^{\rm out},\qquad
\widetilde Z^{\rm ret}=-\mathcal A Z_{\rm total}^{\rm in},\qquad
\widetilde Z^\alpha=Z_0+\sum_{n\geq1}\epsilon^n\widetilde Z_n^\alpha.
\label{eq:oriented_cuts_matching_II}
\end{equation}
The overall minus sign leaves the null-surface equation invariant.
Writing \(l_a=\partial_aZ_0\) and \(D=l^a\partial_a\), its third-order
coefficient is
\begin{equation}
h_{3\,ab}^{\alpha}l^al^b
+2D\widetilde Z_3^\alpha+\mathcal C_{3,\alpha}=0,
\label{eq:third_order_null_matching_II}
\end{equation}
where all terms determined by lower orders are retained:
\begin{align}
\mathcal C_{3,\alpha}={}&
2h^{ab}_{2,\alpha}l_a\partial_b\widetilde Z_1^\alpha
+2h^{ab}_{1,\alpha}l_a\partial_b\widetilde Z_2^\alpha
\nonumber\\
&+h^{ab}_{1,\alpha}\partial_a\widetilde Z_1^\alpha
\partial_b\widetilde Z_1^\alpha
+2\eta^{ab}\partial_a\widetilde Z_1^\alpha
\partial_b\widetilde Z_2^\alpha.
\label{eq:third_order_common_terms_II}
\end{align}
With the same conformal normalization and choice of homogeneous
solutions, the already established first- and second-order matching
gives \(h^{\rm adv}_{j\,ab}=h^{\rm ret}_{j\,ab}\) and
\(\widetilde Z_j^{\rm adv}=\widetilde Z_j^{\rm ret}\) for \(j=1,2\).
Consequently \(\mathcal C_{3,\rm adv}=\mathcal C_{3,\rm ret}
\equiv\mathcal C_3\), and the two equations before subtraction read
\begin{align}
h^{\rm adv}_{3\,ab}l^al^b
&=-2D Z_3^{\rm out}-\mathcal C_3,
\label{eq:advanced_metric_cut_relation}\\
h^{\rm ret}_{3\,ab}l^al^b
&=+2D\mathcal A Z_3^{\rm in}-\mathcal C_3.
\label{eq:retarded_metric_cut_relation}
\end{align}
Imposing \(h^{\rm adv}_{3\,ab}=h^{\rm ret}_{3\,ab}\) cancels the
common term between the two sides and gives
\begin{equation}
\boxed{D\left(Z_3^{\rm out}+\mathcal A Z_3^{\rm in}\right)=0.}
\label{eq:third_order_directional_matching_II}
\end{equation}
With the homogeneous solution fixed as at lower orders, this yields
\(Z_3^{\rm out}+\mathcal A Z_3^{\rm in}=0\).
Thus the equation is linear in the new cut coefficient after the
lower-order contributions have canceled; no linear relation between an
individual full metric perturbation and its full cut has been assumed.
At general order, the null-surface equation has the same structure:
the new metric and cut coefficients enter linearly, while the remaining
terms depend only on previously matched orders and cancel on subtraction.
Understood order by order as a formal perturbative identity, the
directional matching condition for the total cuts is therefore
\begin{equation}
\boxed{
l^a(z)\partial_a
\left[
Z_{\rm total}^{\rm out}(x,z)
+
\mathcal{A} Z_{\rm total}^{\rm in}(x,z)
\right]
=
0.
}
\label{eq:classical_directional_matching}
\end{equation}
The two reconstructions are normalized with respect to the same
Minkowski geometry. We also exclude an additional homogeneous
contribution annihilated by $l^a\partial_a$. With this common
normalization, the matching condition is represented at the level of
the total cuts by
\begin{equation}
\boxed{
Z_{\rm total}^{\rm out}(x,z)
+
\mathcal{A} Z_{\rm total}^{\rm in}(x,z)
=
0.
}
\label{eq:classical_total_cut_matching}
\end{equation}
In particular, the flat cuts satisfy
\begin{equation}
Z_0^{\rm out}(x,z)
+
\mathcal{A} Z_0^{\rm in}(x,z)
=
0.
\label{eq:flat_cut_matching}
\end{equation}
After asymptotic quantization, the same geometrical condition is promoted
to an operator identity,
\begin{equation}
\boxed{
\widehat{Z}_{\rm total}^{\rm out}(x,z)
+
\mathcal{A}\widehat{Z}_{\rm total}^{\rm in}(x,z)
=
0.
}
\label{eq:exact_quantum_total_matching}
\end{equation}
Here and throughout this section, the hat denotes an operator, whereas
$\mathcal{A}$ denotes only the antipodal transformation.

\subsection{Cut--cone decomposition}
\label{subsec:cut_cone_matching}
Each total cut is the sum of an asymptotic cut contribution and a cone
contribution:
\begin{align}
\widehat{Z}_{\rm total}^{\rm out}
&=
\widehat{Z}_{\rm cut}^{\rm out}
+
\widehat{Z}_{\rm cone}^{\rm out},
\label{eq:outgoing_total_cut_decomposition} \\
\widehat{Z}_{\rm total}^{\rm in}
&=
\widehat{Z}_{\rm cut}^{\rm in}
+
\widehat{Z}_{\rm cone}^{\rm in}.
\label{eq:incoming_total_cut_decomposition}
\end{align}
Here the asymptotic cut terms include both the shear-dependent
contribution of Sec.~\ref{sec:asymptotic_quantization} and the scalar flux
of Sec.~\ref{subsec:general_scalar_flux_II}. The cone terms contain the
nonlinear propagation from null infinity to the interior point $x^a$.

At perturbative order $n$, the cone contribution is obtained by
angular inversion of the integral along each null generator:
\begin{equation}
\widehat{Z}_{n,\rm cone}^{\rm out/in}(x,z)
=
\oint d^2z' \, G_{0,0'}(z,z') \, \widehat{I}_n^{\rm out/in}(x,z').
\label{eq:quantum_cone_angular_inversion}
\end{equation}
Thus $\widehat{I}_n$ denotes the source integrated along a single null
generator, whereas $\widehat{Z}_{n,\rm cone}$ denotes the complete cone
contribution after angular inversion. Normal ordering is understood as
part of the quantization prescription.

Substituting Eqs.~\eqref{eq:outgoing_total_cut_decomposition} and
\eqref{eq:incoming_total_cut_decomposition} into
Eq.~\eqref{eq:exact_quantum_total_matching} gives the exact cut--cone
matching equation
\begin{equation}
\boxed{
\begin{aligned}
\widehat{Z}_{\rm cut}^{\rm out}
+
\mathcal{A}\widehat{Z}_{\rm cut}^{\rm in}
&=
-\left(
\widehat{Z}_{\rm cone}^{\rm out}
+
\mathcal{A}\widehat{Z}_{\rm cone}^{\rm in}
\right).
\end{aligned}
}
\label{eq:exact_quantum_cut_cone_matching}
\end{equation}
Define
\begin{align}
\Delta\widehat{Z}_{\rm cut}
&\equiv
\widehat{Z}_{\rm cut}^{\rm out}
+
\mathcal{A}\widehat{Z}_{\rm cut}^{\rm in},
\label{eq:Delta_quantum_cut_definition} \\
\Delta\widehat{Z}_{\rm cone}
&\equiv
\widehat{Z}_{\rm cone}^{\rm out}
+
\mathcal{A}\widehat{Z}_{\rm cone}^{\rm in}.
\label{eq:Delta_quantum_cone_definition}
\end{align}
The exact matching equation then takes the compact form
\begin{equation}
\boxed{
\Delta\widehat{Z}_{\rm cut}
+
\Delta\widehat{Z}_{\rm cone}
=
0.
}
\label{eq:compact_exact_quantum_matching}
\end{equation}
This equation is not a free-field equation imposed independently on
$Z$. It is the cut representation of the geometrical condition
$g_{ab}^{\rm adv}=g_{ab}^{\rm ret}$. The independent asymptotic
radiative variable remains the Bondi shear, or equivalently its
annihilation and creation operators.

After the future and antipodally mapped past cone sources have been
written in terms of the same incoming operators, their affine
integrations carry opposite prescriptions. For any Fourier branch with
denominator $D$, one has
\begin{equation}
\frac{i}{D+i0}+\frac{i}{D-i0}
=
2i\,\mathrm{PV}\frac{1}{D}.
\label{eq:exact_cone_PV_identity}
\end{equation}
Accordingly, the exact cone side can be organized schematically as
\begin{equation}
\boxed{
\Delta\widehat Z_{\rm cone}(x,z)
=
\int d\Gamma\,
\widehat{\mathcal K}[a^{\rm in},a^{{\rm in}\dagger}]
\,2i\,\mathrm{PV}\frac{1}{D},
}
\label{eq:exact_cone_PV_form}
\end{equation}
where $d\Gamma$ includes the momentum and angular integrations and
$\widehat{\mathcal K}$ is the complete NSF cone kernel. Equations
\eqref{eq:exact_matched_quantum_cut_delta_a} and
\eqref{eq:exact_cone_PV_form}, together with the scalar-flux cut
contribution, organize the extraction of the perturbative orders.

\subsection{Perturbative form of the matching equation}
\label{subsec:perturbative_matching}
We expand the total cuts as
\begin{align}
\widehat{Z}_{\rm total}^{\rm out}
&=
Z_0^{\rm out}\mathbf{1}
+
\sum_{n\geq1} \epsilon^n \widehat{Z}_n^{\rm out},
\label{eq:outgoing_total_cut_expansion} \\
\widehat{Z}_{\rm total}^{\rm in}
&=
Z_0^{\rm in}\mathbf{1}
+
\sum_{n\geq1} \epsilon^n \widehat{Z}_n^{\rm in}.
\label{eq:incoming_total_cut_expansion}
\end{align}
Since the flat cuts already satisfy
Eq.~\eqref{eq:flat_cut_matching}, the quantum matching condition is
enforced independently at every positive order:
\begin{equation}
\boxed{
\widehat{Z}_n^{\rm out}
+
\mathcal{A}\widehat{Z}_n^{\rm in}
=
0,
\qquad
n\geq1.
}
\label{eq:order_by_order_total_matching}
\end{equation}
At each order,
\begin{equation}
\widehat{Z}_n^{\rm out/in}
=
\widehat{Z}_{n,\rm cut}^{\rm out/in}
+
\widehat{Z}_{n,\rm cone}^{\rm out/in},
\label{eq:order_n_cut_cone_decomposition}
\end{equation}
and therefore
\begin{equation}
\boxed{
\widehat{Z}_{n,\rm cut}^{\rm out}
+
\mathcal{A}\widehat{Z}_{n,\rm cut}^{\rm in}
+
\widehat{Z}_{n,\rm cone}^{\rm out}
+
\mathcal{A}\widehat{Z}_{n,\rm cone}^{\rm in}
=
0.
}
\label{eq:order_n_cut_cone_matching}
\end{equation}
The exact form of the cut contribution is already given by
Eq.~\eqref{eq:exact_matched_quantum_cut_delta_a}. At a fixed order $n$
its coefficient is obtained from the single product
\begin{equation}
\left[\delta a_\lambda\widehat U_\omega^{-}\right]_{\epsilon^n}.
\label{eq:order_n_exact_cut_coefficient}
\end{equation}
We shall first display this construction explicitly at second and third
orders. Once the third-order calculation has been followed, the general
recursion will be stated in a form that applies without modification at
arbitrary order. In every case, the radiative and helicity projections
are applied to the complete order-$n$ matching equation
\eqref{eq:order_n_cut_cone_matching}, not to an isolated source term.

\subsection{Linear order and trivial scattering}
\label{subsec:first_order_matching}
At first order, the NSF equations are linear and there is no cone
contribution:
\begin{equation}
\widehat{Z}_{1,\rm cone}^{\rm out}
=
\widehat{Z}_{1,\rm cone}^{\rm in}
=
0.
\label{eq:first_order_no_cone}
\end{equation}
The matching equation therefore reduces to
\begin{equation}
\boxed{
\widehat{Z}_{1,\rm cut}^{\rm out}
+
\mathcal{A}\widehat{Z}_{1,\rm cut}^{\rm in}
=
0.
}
\label{eq:first_order_cut_matching}
\end{equation}
Using the antipodal and helicity conventions fixed in
Sec.~\ref{sec:asymptotic_quantization}, this relation gives the trivial
linear scattering map
\begin{equation}
\boxed{
a_{1,\lambda}^{\rm out}(\omega,z)
=
a_{1,\lambda}^{\rm in}(\omega,z).
}
\label{eq:trivial_scattering_linear}
\end{equation}
The outgoing operator therefore begins as
\begin{equation}
a_\lambda^{\rm out}
=
\epsilon a_{1,\lambda}^{\rm in}
+
\mathcal{O}(\epsilon^2).
\label{eq:outgoing_operator_linear_expansion}
\end{equation}
The first nontrivial matching problem arises at second order, where the
quadratic cone contribution and the correction
$\delta a_{2,\lambda}^{\rm out}$ appear.

\subsection{Second-order matching and the known result for
\texorpdfstring{$\delta a_2^{\rm out}$}{delta a2 out}}
\label{subsec:second_order_input}
The first nonlinear correction to the outgoing graviton operators was
derived in the first paper of this series. Before recalling that result,
we make explicit an exact simplification of the matching equation.

Equation~\eqref{eq:exact_matched_quantum_cut_delta_a} shows directly
that the shear-dependent matched cut contribution is proportional to
$\delta a_\lambda\widehat U_\omega^{-}$ and its Hermitian conjugate.
At second order the matched scalar-flux difference also vanishes,
$\delta\Sigma_2=0$, as shown in Eq.~\eqref{eq:deltaSigma2_zero_II}.
There is no first-order change of the asymptotic operators,
\begin{equation}
\delta a_{1,\lambda}=0,
\label{eq:delta_a1_zero}
\end{equation}
because the cone source vanishes at linear order. Therefore the first
nontrivial coefficient of the exact matched cut is
\begin{equation}
\left[\delta a_\lambda\widehat U_\omega^{-}\right]_{\epsilon^2}
=
\delta a_{2,\lambda}^{\rm out}.
\label{eq:second_order_cut_coefficient_exact}
\end{equation}
This conclusion follows without expanding the common unitary factor
beyond the order needed and without introducing a separate nonlinear
cut deformation.

At second order, after this exact cancellation, the matching equation
contains only the new outgoing correction and the quadratic cone
source:
\begin{equation}
\boxed{
\widehat{Z}_{2,\rm cut}^{\rm out}
\left[
\delta a_2^{\rm out}
\right]
=
-
\left(
\widehat{Z}_{2,\rm cone}^{\rm out}
+
\mathcal{A}\widehat{Z}_{2,\rm cone}^{\rm in}
\right).
}
\label{eq:second_order_reduced_matching}
\end{equation}
The helicity projection of
Eq.~\eqref{eq:second_order_reduced_matching} and the explicit expression
for $\delta a_{2,\lambda}^{\rm out}$ were derived in the first paper
of this series and will not be repeated here. We denote that result
schematically by
\begin{equation}
\delta a_{2,\lambda}^{\rm out}
=
\mathcal{F}_{2,\lambda}
\left[
a_1^{\rm in},
a_1^{{\rm in}\dagger}
\right],
\label{eq:deltaa2_known_result}
\end{equation}
where the kernel $\mathcal{F}_{2,\lambda}$ contains the sum- and
difference-frequency channels and the angular structures generated by
the second-order NSF cone source.

For the analysis below, this previously determined operator is used as
an input through
\begin{equation}
a_\lambda^{\rm out}
=
\epsilon a_{1,\lambda}^{\rm in}
+
\epsilon^2
\delta a_{2,\lambda}^{\rm out}
+
\epsilon^3
\delta a_{3,\lambda}^{\rm out}
+
\mathcal{O}(\epsilon^4).
\label{eq:outgoing_operator_expansion_through_third}
\end{equation}
The new calculation undertaken in the present work begins at third
order.

\subsection{General recursion of the quantum cut phase}
\label{subsec:general_cut_phase_recursion_II}
We first retain the full perturbed-cut construction and its ordered
unitary phase.  For a finite perturbative truncation,
define the cumulative cut and its unitary phase by
\begin{equation}
\widehat Z_{[N]}
=
\sum_{j=1}^{N}\epsilon^j\widehat Z_j,
\qquad
\widehat U_{\omega,[N]}
=
\exp[-i\varpi\widehat Z_{[N]}].
\label{eq:truncated_quantum_cut}
\end{equation}
Each $\widehat Z_{[N]}$ is self-adjoint on the chosen smeared domain, and
therefore $\widehat U_{\omega,[N]}$ is unitary. The update from order
$N-1$ to order $N$ is the incremental unitary operator
\begin{equation}
\boxed{
\widehat U_{\omega,N}
\equiv
\widehat U_{\omega,[N-1]}^\dagger
\widehat U_{\omega,[N]},
\qquad
\widehat U_{\omega,[N]}
=
\widehat U_{\omega,[N-1]}
\widehat U_{\omega,N}.
}
\label{eq:incremental_unitary_factor}
\end{equation}
Here $\widehat U_{\omega,[0]}=\mathbf 1$. This definition does not assume
that the coefficients $\widehat Z_j$ commute.

The cumulative phase has the expansion
\begin{equation}
\widehat U_\omega^{-}
=
\sum_{r=0}^{\infty}\epsilon^r
\widehat{\mathcal U}_{\omega,r},
\qquad
\widehat{\mathcal U}_{\omega,0}=\mathbf 1,
\label{eq:U_coefficients_general}
\end{equation}
whose first two nontrivial coefficients are
\begin{align}
\widehat{\mathcal U}_{\omega,1}
&=
-i\varpi\widehat Z_1,
\nonumber\\
\widehat{\mathcal U}_{\omega,2}
&=
-i\varpi\widehat Z_2
-\frac{\varpi^2}{2}\widehat Z_1^2.
\label{eq:U_first_coefficients}
\end{align}
At order $r$, $\widehat{\mathcal U}_{\omega,r}$ depends only on
$\widehat Z_1,\ldots,\widehat Z_r$.

Write also
\begin{equation}
\delta a_\lambda
=
\sum_{m=2}^{\infty}\epsilon^m\delta a_{m,\lambda}^{\rm out}.
\label{eq:delta_a_general_expansion}
\end{equation}
Then
\begin{equation}
\boxed{
\left[\delta a_\lambda\widehat U_\omega^{-}\right]_{\epsilon^n}
=
\delta a_{n,\lambda}^{\rm out}
+
\sum_{m=2}^{n-1}
\delta a_{m,\lambda}^{\rm out}
\widehat{\mathcal U}_{\omega,n-m},
\qquad n\geq2.
}
\label{eq:general_cut_recursion_coefficient}
\end{equation}
The first term contains the new operator, while the finite sum contains
only lower-order operators. After integration with the angular kernels,
the complete order-$n$ equation is
\begin{equation}
\boxed{
\Delta\widehat Z_{{\rm cut},n}
+
\Delta\widehat Z_{{\rm cone},n}=0,
}
\label{eq:general_order_n_matching_after_derivation}
\end{equation}
and the radiative projection of this entire equation determines
$\delta a_{n,\lambda}^{\rm out}$. Thus the same procedure gives the
outgoing correction at any perturbative order without reconstructing
the cut formula anew.

\subsection{Nonlinear scalar flux on the cut}
\label{subsec:general_scalar_flux_II}
The scalar-flux term must be treated independently of the
Yang--Feldman prescription used for the local cone source.  It belongs to
the asymptotic cut contribution. Its perturbative expansion includes
both the flux coefficients and the displacement of the upper integration
limit; the terms that survive at third order are displayed below.  Define
\begin{equation}
\delta\Sigma
\equiv
\Sigma^{+}-\Sigma^{-},
\qquad
\Sigma^\pm(Z)
=
\int_{-\infty}^{Z}\dd u\,
{}:
\dot\sigma^\pm(u)\dot{\bar\sigma}^{\,\pm}(u)
:: .
\label{eq:deltaSigma_definition_II}
\end{equation}
The quadratic free contribution is identical on the two sides by linear
matching,
\begin{equation}
\boxed{
\delta\Sigma_2=0.
}
\label{eq:deltaSigma2_zero_II}
\end{equation}
At third order there are two possible structures on the cut.  Expanding
the upper limit around \(Z_0\) produces a term of type \((1,1,1)\),
proportional to
\(Z_1\,:\dot\sigma_1\dot{\bar\sigma}_1:\).
The corresponding future and antipodally matched past contributions are
identical and cancel in their difference.  The remaining third-order flux
difference is the \((1,2)+(2,1)\) contribution,
\begin{equation}
\boxed{
\delta\Sigma_3(Z_0)
=
\int_{-\infty}^{Z_0}\dd u\,
{}:
\left[
\dot\sigma_1\,\delta\dot{\bar\sigma}_2
+
\delta\dot\sigma_2\,\dot{\bar\sigma}_1
\right]
:: ,
}
\label{eq:deltaSigma3_nonzero_II}
\end{equation}
where
\(\delta\sigma_2\equiv\sigma_2^{\rm out}-\sigma_2^{\rm in}\).
In the outgoing-minus-incoming boundary balance,
the News contribution to the change of $m_{TW}$ is $-\delta\Sigma_3$.
In Eq.~\eqref{eq:Z3_Sigma_cut_II}, the source for the cut is
$-\delta\Sigma_3$, the same signed flux that appears in the boundary balance.
Thus \(\delta\Sigma_3\) is generically nonzero.

Equivalently, with
\begin{equation}
\mathcal F^\pm(u,z)
=
{}:
\dot\sigma^\pm(u,z)\dot{\bar\sigma}^{\,\pm}(u,z)
:: ,
\qquad
\Sigma^\pm_\varpi
=
\frac{\mathcal F^\pm_\varpi}
{-\ii(\varpi+\ii0)},
\label{eq:Sigma_Fourier_II}
\end{equation}
one has
\begin{equation}
\delta\Sigma_{\varpi,3}
=
\frac{
\delta\mathcal F_{\varpi,3}
}{
-\ii(\varpi+\ii0)
},
\qquad
\delta\mathcal F_3
=
{}:
\dot\sigma_1\,\delta\dot{\bar\sigma}_2
+
\delta\dot\sigma_2\,\dot{\bar\sigma}_1
:: .
\label{eq:deltaSigma3_Fourier_II}
\end{equation}
This contribution is not a Yang--Feldman cone term and should not be
identified with the common retarded source responsible for the
principal-value double-cone denominator.


\subsection{Third-order source contributions and cut recursion}
\label{subsec:third_order_matching_complete}
The third-order recursion contains the ordered free-data/cut contribution
and the source terms displayed here.  Before any specialization is made,
the source part of the outgoing coefficient separates into
\begin{equation}
\boxed{
\delta a_{3,\lambda}^{\rm out,sources}
=
\delta a_{3,\lambda}^{\rm out,YF}
+
\delta a_{3,\lambda}^{\rm out,(\Sigma)}
+
\delta a_{3,\lambda}^{\rm out,metric}
+
\delta a_{3,\lambda}^{\rm out,geo}.
}
\label{eq:deltaa3_source_decomposition_II}
\end{equation}
The first term is the causal Yang--Feldman paired-cone source evaluated
on the zeroth-order cone,
\begin{equation}
\delta a_{3,\lambda}^{\rm out,YF}
=
-
\left[
Z_{3,\rm cone}^{\rm YF,flat}
\right]_{F,\lambda}.
\label{eq:deltaa3_flat_cone_piece_II}
\end{equation}
The second is the nonlinear scalar-flux contribution on the cut,
\begin{equation}
\delta a_{3,\lambda}^{\rm out,(\Sigma)}
=
-
\left[
Z_{3,\Sigma}
\right]_{F,\lambda},
\label{eq:deltaa3_Sigma_piece_II}
\end{equation}
where the surviving third-order flux difference is evaluated on
\(Z_0\), after the common lower-order upper-limit terms have canceled,
\begin{equation}
\boxed{
\begin{aligned}
Z_{3,\Sigma}(x,z)
={}&-
\oint\dd^2z'\,
G_{0,0'}(z,z')
\int_0^\infty\frac{\dd\varpi}{2\pi}
\Big[
\delta\Sigma_{\varpi,3}(z')
\ee^{-\ii\varpi Z_0(x,z')}
+
\mathrm{H.c.}
\Big].
\end{aligned}
}
\label{eq:Z3_Sigma_cut_II}
\end{equation}
The metric sector is
\begin{equation}
\delta a_{3,\lambda}^{\rm out,metric}
=
-
\left[
Z_{3,\rm cone}^{\rm metric}
\right]_{F,\lambda},
\label{eq:deltaa3_geom_cone_piece_II}
\end{equation}
and contains, at this order, the inverse-metric contribution
\(h_1^{ab}\partial_a\Lambda_1\partial_b\bar\Lambda_1\).  The remaining
term, \(\delta a_{3,\lambda}^{\rm out,geo}\), is generated by the
first-order deformation of the null generators and is calculated explicitly
in the next subsection.  Thus, schematically,
\begin{equation}
Z_{3,\rm cone}^{\rm full}
=
Z_{3,\rm cone}^{\rm YF,flat}
+
Z_{3,\rm cone}^{\rm metric}
+
Z_{3,\rm cone}^{\rm def}.
\label{eq:Z3_cone_full_split_II}
\end{equation}

The exact matched cut also contains the ordered product
\(\delta a_\lambda\widehat U_\omega^{-}\).  Its third-order coefficient is
\begin{equation}
\left[
\delta a_\lambda\widehat U_\omega^{-}
\right]_{\epsilon^3}
=
\delta a_{3,\lambda}^{\rm out}
-\ii\varpi\,
\delta a_{2,\lambda}^{\rm out}\widehat Z_1 .
\label{eq:third_order_general_cut_coefficient_II}
\end{equation}

The full outgoing coefficient is obtained from this ordered cut term and
all four source terms in the same matching equation.  Equivalently, its
non-Yang--Feldman contributions fall into the three sectors that will be
identified below as the free-data, metric, and cone-deformation sectors.

\subsection{Geodesic contribution to the third-order outgoing operator}
\label{subsec:general_geometrical_deformations_II}
The last contribution in
Eq.~\eqref{eq:deltaa3_source_decomposition_II} is produced by evaluating
the lower-order cone source along the perturbed null generators.  It is
distinct from the local inverse-metric contribution already contained in
\(\delta a_{3,\lambda}^{\rm out,metric}\).  The second-order source whose
trajectory variation contributes at third order is
\begin{equation}
\widehat{\mathcal Q}_2^{\rm out/in}
\equiv
2\bar{\eth}'\eth'\delta\widehat\Omega_2^{\rm out/in}
+
\eta^{ab}\partial_a\widehat\Lambda_1^{\rm out/in}
\partial_b\widehat{\bar\Lambda}_1^{\rm out/in}.
\label{eq:Q2_cone_source_definition}
\end{equation}

The corresponding displacement of the generators follows from the cut
function itself. In the general construction the generators are expanded as
\begin{equation}
x_s^a(z')
=
x_s^{(0)a}(z')
+
\epsilon \, \delta x_1^a(s,z')
+
\epsilon^2 \, \delta x_2^a(s,z')
+
\mathcal{O}(\epsilon^3),
\label{eq:null_generator_perturbative_expansion}
\end{equation}
where $x_s^{(0)a}$ denotes the corresponding Minkowski null
geodesic. The displacement follows directly from the cut coefficient
\(Z_1\), as derived in Appendix~\ref{app:geodesics_from_cut_II}:
\begin{equation}
\boxed{\begin{aligned}
\delta x_1^a(s,z')={}&
-\Delta_s(\eth'\bar\eth' Z_1+Z_1)l'^a
-\Delta_s Z_1\,n'^a\\
&+\Delta_s(\bar\eth' Z_1)m'^a
+\Delta_s(\eth' Z_1)\bar m'^a,
\end{aligned}}
\label{eq:delta_x1_from_cut_II}
\end{equation}
with all tetrad vectors associated with \(z'\) and
\(\Delta_s F=F(x+s l',z')-F(x,z')\). In particular,
\(\delta x_1^a(0,z')=0\), so the cone vertex is fixed.
The angular derivatives are taken at fixed spacetime argument before
evaluation on the background ray. Since the cone source begins at second order,
\begin{equation}
\widehat{\mathcal{Q}}
=
\epsilon^2\widehat{\mathcal{Q}}_2
+
\epsilon^3\widehat{\mathcal{Q}}_3
+
\mathcal{O}(\epsilon^4),
\label{eq:cone_source_perturbative_expansion}
\end{equation}
its evaluation along the perturbed trajectory gives
\begin{align}
\widehat{\mathcal{Q}}\bigl(x_s\bigr)
&=
\epsilon^2
\widehat{\mathcal{Q}}_2\bigl(x_s^{(0)}\bigr)
\nonumber \\
&\quad+
\epsilon^3
\left[
\widehat{\mathcal{Q}}_3\bigl(x_s^{(0)}\bigr)
+
\delta x_1^c
\partial_c\widehat{\mathcal{Q}}_2
\bigl(x_s^{(0)}\bigr)
\right]
+
\mathcal{O}(\epsilon^4).
\label{eq:cone_source_perturbed_trajectory}
\end{align}
The first-order deformation of the generator therefore produces the
additional third-order source
\begin{equation}
\widehat{\mathcal{Q}}_{3,\rm geo}
=
\delta x_1^c \, \partial_c\widehat{\mathcal{Q}}_2.
\label{eq:third_order_geodesic_source}
\end{equation}
Correspondingly, the complete cone contribution contains
\begin{align}
\delta\widehat{Z}_{3,\rm geo}^{\rm out}(x,z)
&=
-\oint d^2z' \,
G_{0,0'}(z,z')
\int_0^\infty ds \,
\left.
\left[
\delta x_{1,\rm out}^c
\partial_c\widehat{\mathcal{Q}}_2^{\rm out}
\right]
\right|_{x=x_s^{(0),\rm out}(z')},
\label{eq:outgoing_geodesic_correction} \\
\mathcal{A}\delta\widehat{Z}_{3,\rm geo}^{\rm in}(x,z)
&=
-\oint d^2z' \,
G_{0,0'}(z,z')
\int_0^\infty ds \,
\left.
\left[
\delta x_{1,\rm in}^c
\partial_c\widehat{\mathcal{Q}}_2^{\rm in}
\right]
\right|_{x=\mathcal{A} x_s^{(0),\rm in}(z')}.
\label{eq:incoming_geodesic_correction}
\end{align}

The geodesic contribution entering the third-order matching equation is
therefore
\begin{equation}
\boxed{
\Delta\widehat Z_{3,\rm geo}
\equiv
\delta\widehat Z_{3,\rm geo}^{\rm out}
+
\mathcal A\delta\widehat Z_{3,\rm geo}^{\rm in}.
}
\label{eq:Delta_Z3_geodesic_II}
\end{equation}
Its radiative and helicity projection gives the corresponding missing term
in the complete third-order outgoing operator,
\begin{equation}
\boxed{
\delta a_{3,\lambda}^{\rm out,geo}
=
-\left[
\Delta\widehat Z_{3,\rm geo}
\right]_{F,\lambda}.
}
\label{eq:deltaa3_geodesic_contribution_II}
\end{equation}
Thus the perturbation of the null generators produces a definite,
generically nonvanishing contribution to \(\delta a_{3,\lambda}^{\rm out}\),
rather than merely a possible correction to the cone source.  At this order,
this is the deformation sector denoted schematically by
\(Z_{3,\rm cone}^{\rm def}\) in
Eq.~\eqref{eq:Z3_cone_full_split_II}.

The trajectory variation and the local inverse-metric term are distinct
contributions to the general cone source. The affine normalization is
derived in Appendix~\ref{app:geodesics_from_cut_II}; with
\(\Omega_1=0\), its first correction begins at second order and can
modify a cone source starting at second order only from fourth order onward.

Up to this point no Minkowski-background truncation has been imposed:
the nonlinear cut/free-data contribution, the perturbation of the metric,
and the perturbation of the null generators have all been identified as
parts of the complete third-order NSF construction.  Immediately below we
define a distinct operator, \(\delta a_{3,M}^{\rm out}\), to be used in the
subsequent scattering sections.  That operator is obtained by imposing
simultaneously the three approximations stated below.

\subsection{Minkowskian approximation for the
\texorpdfstring{$2\to2$}{2 to 2} comparison}
\label{subsec:minkowskian_approximation_II}

At this point the third-order NSF construction is still general.  Its
non-Yang--Feldman terms fall into three distinct sectors:
\begin{enumerate}
\item the nonlinear free-data sector, containing the ordered cut phase and
the scalar-flux contribution;
\item the metric sector, containing in particular
\(h_1^{ab}\partial_a\Lambda_1\partial_b\bar\Lambda_1\);
\item the null-cone sector, containing the perturbations of the cuts, the
null generators, the affine parametrization, and the cone measure, including
\(\delta x_1^c\partial_c\mathcal Q_2\).
\end{enumerate}
We denote the corresponding outgoing contributions by
\(\delta a_{3}^{\rm out,cut}\),
\(\delta a_{3}^{\rm out,metric}\), and
\(\delta a_{3}^{\rm out,geo}\), respectively.  Here the first term
contains both the recursive ordered-cut contribution and
\(Z_{3,\Sigma}\), while the last contains the deformation of the null
generators and the associated cone geometry.  None of these terms is
assumed to vanish in the general NSF construction.

For the explicit \(2\to2\) calculation we now make a direct comparison
with conventional perturbative quantum gravity.  In that formulation the
background geometry is kept fixed throughout the calculation: metric
contractions are evaluated with the Minkowski metric, propagation is defined
by Minkowski null cones and null geodesics, and the asymptotic phases are
evaluated on Minkowski cuts.  Therefore, in order to compare the NSF result
with perturbative quantum gravity in precisely the same background sector,
both calculations must use the same fixed Minkowski geometry.  The complete
NSF operator \(\delta a_3^{\rm out}\) cannot be compared directly with that
perturbative-gravity calculation because it also contains the nonlinear
free-data, metric-backreaction, and null-geodesic-deformation contributions,
which are absent from the fixed-background formulation.  We therefore define
\(\delta a_{3,M}^{\rm out}\) by removing these three contributions
simultaneously.  The purpose of this new definition is to isolate the sector
in which the two formulations are evaluated with the same Minkowski metric,
Minkowski null cones, and Minkowski cuts, so that a direct comparison is
well-defined.  This is the Minkowskian approximation stated below.

\medskip
\noindent\textbf{Definition (Minkowskian approximation).}
The Minkowskian approximation is the simultaneous implementation of the
following three approximations in the definition of
\(\delta a_{3,M}^{\rm out}\):
\begin{enumerate}
\item \textbf{Linear free-data/cut prescription:}
\begin{equation}
\widehat U_\omega^-=\mathbf1,
\qquad
\Sigma_2=0.
\label{eq:Minkowskian_data_condition_II}
\end{equation}
Thus the freely specified asymptotic radiative fields are linear in their
creation and annihilation operators.  Consequently, in defining the
Minkowskian third-order scattering operator below we omit the recursive
ordered-cut term and the nonlinear scalar-flux contribution
\(Z_{3,\Sigma}\).

\item \textbf{Minkowski metric:}
\begin{equation}
g^{ab}_{\rm eval}=\eta^{ab}.
\label{eq:Minkowskian_metric_condition_II}
\end{equation}
No perturbative inverse-metric insertion is included in the cone source.

\item \textbf{Minkowski cuts and null generators:}
\begin{equation}
Z_{\rm eval}=Z_0,
\qquad
x_s^a=x_s^{(0)a}.
\label{eq:Minkowskian_cone_condition_II}
\end{equation}
The affine parameter and cone measure are also fixed at their Minkowski
values.
\end{enumerate}
The background phase \(\exp[-\ii\varpi x\cdot l']\), equivalently the
Minkowski cut \(Z_0\), is retained.  Only its perturbative cut-phase factor
is set to the identity.  The three approximations are imposed together, rather
than retaining any isolated metric or cone-deformation term.

Quantities evaluated after this definition carry a subscript \(M\).  In
particular, the third-order outgoing operator used in the scattering
calculation is defined by
\begin{equation}
\boxed{
\delta a_{3,M}^{\rm out}
\equiv
\delta a_3^{\rm out}
-\delta a_3^{\rm out,cut}
-\delta a_3^{\rm out,metric}
-\delta a_3^{\rm out,geo}
=-
\left[Z_{3,\rm cone}^{\rm YF,flat}\right]_F .
}
\label{eq:deltaa3_Minkowskian_definition_II}
\end{equation}
Thus \(\delta a_{3,M}^{\rm out}\) is not the complete third-order NSF
operator: it is the operator obtained from it after removing precisely the
three contributions identified above.  This definition is made solely for
the comparison with perturbative quantum gravity in the subsequent
scattering sections and is not a restriction of the general NSF construction
developed above.  The same fixed-Minkowski-background prescription is used
in the companion Part~III~\cite{KDIII2026}.  The later Poincar\'e reduction
is a further restriction on the external scattering data.

\section{Two-to-two scattering: tree-level contact and one-loop contributions}
\label{sec:two_to_two_one_loop_uv}
We now apply the Minkowskian outgoing map to connected \(2\to2\) graviton
scattering, with two incoming and two outgoing asymptotic excitations.  The
two contributions studied below are the additional tree-level contact term,
which must be added to the \(s,t,u\) exchange terms of Part~I, and the loop
contribution generated by
\(\delta a_{3,M}^{\rm out}\delta a_{3,M}^{\rm out}\).  Both are to be
understood within the approximation of
Sec.~\ref{subsec:minkowskian_approximation_II}; the remaining one-loop
contributions will be addressed in subsequent work.

\subsection{The additional tree-level contact term in
\texorpdfstring{$2\to2$}{2 to 2} scattering}
\label{subsec:contact_tree_II}
\subsubsection{Connected perturbative sectors}
\label{subsec:perturbative_sectors_2to2}
We use the relative-order convention of Sec.~\ref{sec:unitarity},
$a_M^{\rm out}=a+\epsilon\delta a_2^{\rm out}
+\epsilon^2\delta a_{3,M}^{\rm out}+\cdots$.
Set $b_{1,M}=a$, $b_{2,M}=\delta a_2^{\rm out}$, and
$b_{3,M}=\delta a_{3,M}^{\rm out}$, and define
\begin{equation}
\mathcal M_M^{(mn)}
=\langle0|b_{m,M;3}b_{n,M;4}a_1^\dagger a_2^\dagger|0\rangle_{\rm conn},
\qquad i=(\lambda_i,\vec k_i).
\label{eq:contact_sector_definition_II}
\end{equation}
Here $1,2$ label the incoming legs and $3,4$ the outgoing legs.
The free product has no connected scattering part. Products of a
quadratic correction with a linear or cubic correction vanish by
operator-number parity. This remains true when the linear ordering term
in the complete $\delta a_{3,M}^{\rm out}$ is retained. Thus the matrix
element formed from the Minkowskian outgoing polynomial truncated at
$\delta a_{3,M}$ is
\begin{equation}
\boxed{
\mathcal M_{M,\rm conn}[a^{\rm out}_{M,[3]}]
=\epsilon^2\bigl(\mathcal M_M^{(13)}+\mathcal M_M^{(22)}
+\mathcal M_M^{(31)}\bigr)+\epsilon^4\mathcal M_M^{(33)}.
}
\label{eq:Mconn_sectors_summary}
\end{equation}
In particular, the products $a\delta a_{3,M}$ and $\delta a_{3,M}a$ must not be
omitted: they contribute at the same tree order as $\delta a_2\delta a_2$.
The coefficient of $\epsilon^4$ in this truncated product is not the
complete coefficient at that order in the full outgoing map, which would
also involve $\delta a_4$ and $\delta a_5$.
The direct third-order product is
\begin{equation}
\mathcal M_M^{(33)}
=
\left\langle0\right|
\delta a^{\rm out}_{3,M,\lambda'_1}(p'_1)
\delta a^{\rm out}_{3,M,\lambda'_2}(p'_2)
 a^{{\rm in}\dagger}_{\lambda_1}(p_1)
 a^{{\rm in}\dagger}_{\lambda_2}(p_2)
\left|0\right\rangle_{\rm conn}.
\label{eq:M33_sector_definition}
\end{equation}
We first identify the additional four-leg tree vertex.
Subsection~\ref{subsec:loop_cubic_cubic_II} then evaluates the common
cubic operator and its irreducible cubic--cubic loop contribution.
The complete canonical coefficient also contains the Wick-induced linear
term of Eq.~\eqref{eq:unitarity_D_from_B_II}. Its mixed contraction with
the cubic coefficient can generate a connected four-point contribution
with a two-leg insertion attached to a four-leg vertex. At the level of
these effective vertices this graph is one-particle reducible: cutting
their connecting line disconnects it. Such contributions are outside the
irreducible sector evaluated here; they are not assumed to vanish, even
for distinct external momenta. The linear term remains part of the
canonical outgoing operator. Its reducible scattering contributions and
their physical interpretation are deferred to future work.

\subsubsection{Contact kernel and exchange channels}
Part~I~\cite{KDI2026} evaluates the scattering sector built from the
quadratic outgoing correction, giving the three exchange channels
\(s,t,u\). The contact contribution derived here must be added to those
results to obtain the complete tree contribution in the same comparison
sector. At the same perturbative order, the cubic
correction supplies the four-leg part of the second Hermitian generator.
For the tree decomposition we use the formal generator organization
$S(\epsilon)=\exp[i\epsilon\delta T_1+i\epsilon^2\delta T_2+\cdots]$
with the vacuum phase fixed by $S|0\rangle=|0\rangle$.
The general existence argument and the constructive reconstruction of
this unitary representation belong to the separate BCH study discussed
in Sec.~\ref{sec:unitarity}.
Its expansion gives
\begin{align}
S(\epsilon)
&=1+i\epsilon\delta T_1
+\epsilon^2\left(i\delta T_2-\frac12\delta T_1^2\right)
+\mathcal O(\epsilon^3),
\label{eq:contact_S_expansion_II}\\
\mathcal M_{M,\rm tree}
&\equiv\mathcal M_M^{(13)}+\mathcal M_M^{(22)}+\mathcal M_M^{(31)}
\nonumber\\
&=\langle3,4|\left(i\delta T_2-\frac12\delta T_1^2\right)
|1,2\rangle_{\rm conn},
\label{eq:contact_tree_matrix_II}
\end{align}
where $|1,2\rangle=a_1^\dagger a_2^\dagger|0\rangle$ and
$\langle3,4|=\langle0|a_3a_4$. The coefficient $\mathcal M_{M,\rm tree}$
does not include its overall factor $\epsilon^2$.
The first term contains one quartic vertex; the second contains two cubic
vertices and gives the three exchange channels. Therefore
\begin{equation}
\boxed{\mathcal M_{M,\rm tree}
=\mathcal M_{M,\rm contact}+\mathcal M_{M,s}+\mathcal M_{M,t}+\mathcal M_{M,u},}
\label{eq:contact_plus_exchange_II}
\end{equation}
with $\mathcal M_{M,\rm contact}=\langle3,4|i\delta T_2|1,2\rangle_{\rm conn}$
and $\sum_{I=s,t,u}\mathcal M_{M,I}
=-\tfrac12\langle3,4|\delta T_1^2|1,2\rangle_{\rm conn}$.
The corresponding internal spatial momenta can be chosen as
\begin{equation}
\vec q_s=\vec k_1+\vec k_2,\qquad
\vec q_t=\vec k_1-\vec k_3,\qquad
\vec q_u=\vec k_1-\vec k_4,
\label{eq:contact_exchange_routings_II}
\end{equation}
with $\omega_{q_I}=|\vec q_I|$. These label the three pairings of external
legs; no energy-conserving delta function is imposed at this stage.

To isolate the contact kernel without counting the lower-order iteration
twice, the BCH relation must be used in its full third-order form:
\begin{equation}
\delta a_{3,M,i}^{\rm out}
=i[a_i,\delta T_2]+I_{3,i},\qquad
I_{3,i}=\frac{i}{2}[\delta a_{2,i}^{\rm out},\delta T_1].
\label{eq:contact_BCH_subtraction_II}
\end{equation}
Thus the contact is reconstructed from $\delta a_{3,M}^{\rm out}-I_3$,
not from the unsubtracted cubic coefficient. This subtraction is fixed
by the lower-order BCH iteration and introduces no new dynamical kernel.
In particular, the direct-product
terms $\mathcal M_M^{(13)}+\mathcal M_M^{(31)}$ also contain pieces fixed by
$\delta T_1$ and are not, by themselves, identical to the contact term.

For an explicit expression in the normalization of Part~III, write
$\int_i=\sum_{\lambda_i}\int d\mu_i$ and define
\begin{equation}
\int_i\Delta(j,i)F(i)=\chi F(j),\qquad
\chi=1,\qquad
\Delta(j,i)=2\omega_j\delta_{\lambda_j\lambda_i}
\delta^{(3)}(\vec k_j-\vec k_i).
\label{eq:contact_chi_II}
\end{equation}
Let $\mathcal B_M(q,r|1,2)$ denote the coefficient of the mixed cubic
monomial, with the common spatial delta function written separately:
\begin{align}
\left.\delta a_{3,M}^{\rm out}(q)\right|_{a^\dagger aa}
&=\int_r\int_1\int_2
\delta^{(3)}(\vec q+\vec k_r-\vec k_1-\vec k_2)
\nonumber\\
&\quad\times\mathcal B_M(q,r|1,2)a_r^\dagger a_1a_2.
\label{eq:contact_B_definition_II}
\end{align}
In terms of the normalized vertices of
Eq.~\eqref{eq:deltaa3_retained_flatcone_II},
\begin{equation}
\mathcal B_M(q,r|1,2)
=\left[\mathcal W^{(3)}_{M,+-}(q;1,2,r)
+\mathcal W^{(3)}_{M,-+}(q;1,r,2)\right]_{(12)},
\label{eq:contact_B_from_W_II}
\end{equation}
where helicity labels are included in the leg labels and parentheses
denote normalized Bose symmetrization. Denote by $\mathcal I_B$ the
coefficient of $I_3$ with the same measure and spatial support, and set
$\mathcal B_{R,M}=\mathcal B_M-\mathcal I_B$. The four-leg lift is
\begin{align}
\left.i\delta T_2\right|_{a^{\dagger2}a^2}
&=\frac{1}{2\chi}\int_1\int_2\int_3\int_4
\delta^{(3)}(\vec k_1+\vec k_2-\vec k_3-\vec k_4)
\nonumber\\
&\quad\times[\mathcal B_{R,M}(3,4|1,2)]_{(34)(12)}
a_3^\dagger a_4^\dagger a_1a_2.
\label{eq:contact_quartic_lift_II}
\end{align}
The factor $1/(2\chi)$ compensates the two creator contractions in
$[a,i\delta T_2]$. The external contractions give the factor
\[
\frac{(2!)(2!)\chi^4}{2\chi}=2\chi^3.
\]
Consequently,
\begin{equation}
\boxed{\begin{aligned}
\mathcal M_{M,\rm contact}(3,4;1,2)
={}&2\chi^3
\delta^{(3)}(\vec k_1+\vec k_2-\vec k_3-\vec k_4)
\\
&\times[\mathcal B_{R,M}(3,4|1,2)]_{(34)(12)}.
\end{aligned}}
\label{eq:contact_amplitude_II}
\end{equation}
All factors in this point-vertex formula follow from the CCR and the
measure $d\mu=d^3k/(2\omega)$ defined above. The kernels
and external contractions must always use this same normalization.

With $\chi=1$, the coefficient of the four-leg lift is $1/2$ for
$i\delta T_2$, and the point-vertex matrix element is $2\delta^{(3)}
[\mathcal B_{R,M}]_{(34)(12)}$. More generally, let $R_{p,q}$ denote
a normally ordered residual coefficient with $p$ creators and $q$
annihilators in the variation of an annihilator. In a Bose-symmetric
compatible block, the coefficient of $\delta T_n$ with $p+1$ creators is
$-iR_{p,q}/(p+1)$, since commuting one annihilator through that block
gives $p+1$ equal terms. At first order $R=\delta a_2$; at second order
$R=\delta a_3-(i/2)[\delta a_2,\delta T_1]$. In particular the
$a^\dagger a$ coefficient in $\delta T_2$ is $-i$ times the linear
$a$ coefficient of this residual, including the linear part of the iterated
commutator. The adjoint relation fixes the purely annihilating blocks.

All four integrations in Eq.~\eqref{eq:contact_quartic_lift_II} are
localized by external contractions. The contact therefore has $V=1$,
$I=0$, and $L=0$; each exchange graph has $V=2$, $I=1$, and $L=0$.
Here ``contact'' means a single four-leg generator insertion, not a claim
that its boundary kernel is local or independent of momentum. The linear
ordering term in $\delta a_{3,M}$ reconstructs a bilinear sector of
$\delta T_2$ and contributes only spectator, not connected four-leg,
terms at this tree order. Consequently the contact uses the same cubic
kernel for $c=0$ and $c=1$.
Equation~\eqref{eq:contact_plus_exchange_II} is the tree expression to
use in any subsequent Poincar\'e comparison; the quadratic-operator
calculation alone must not be identified with the complete tree amplitude.

\subsection{The one-loop contribution to
\texorpdfstring{$2\to2$}{2 to 2} scattering}
\label{subsec:loop_cubic_cubic_II}
This subsection computes the irreducible connected \(2\to2\) loop sector generated by
\(\delta a_{3,M}^{\rm out}\delta a_{3,M}^{\rm out}\). This contribution is to be
combined with the remaining loop contributions in subsequent work; the
present calculation does not supply the complete one-loop amplitude.
The cubic operator used in both subsections is evaluated in the
Minkowskian approximation defined above. We give the paired-cone source
and its projections before forming the connected loop contraction.

The retarded Yang--Feldman prescription is used only for the local
lower-order fields entering the paired-cone source.  After both cone
sources are written as the same functional of the incoming radiative data,
their opposite affine prescriptions give
\begin{equation}
\frac{\ii}{l\cdot K+\ii0}
+
\frac{\ii}{l\cdot K-\ii0}
=
2\ii\,\PV\frac{1}{l\cdot K}.
\label{eq:YF_common_source_PV_II}
\end{equation}
This principal-value mechanism is independent of the scalar-flux term
discussed above.

\subsubsection{The third-order cone contribution}
\label{subsec:third_order_cone}
The paired-cone term entering
Eq.~\eqref{eq:deltaa3_Minkowskian_definition_II} is the sum of two
geometrically distinct contributions,
\begin{equation}
\Delta\widehat{Z}_{{\rm cone},3}
=
\widehat{Z}_{3,\rm cone}^{\rm out}
+
\mathcal{A}\widehat{Z}_{3,\rm cone}^{\rm in},
\label{eq:third_order_cone_sum_recalled}
\end{equation}
associated respectively with the advanced reconstruction from
$\mathscr{I}^+$ and the retarded reconstruction from
$\mathscr{I}^-$, pulled back to the same celestial sphere by the
antipodal map.

With the kinematics fixed by
Eqs.~\eqref{eq:Minkowskian_metric_condition_II} and
\eqref{eq:Minkowskian_cone_condition_II}, the third-order source integrated
along each null generator contains the conformal-factor correction and the
two cross terms involving the first- and second-order angular fields:
\begin{equation}
\boxed{
\begin{aligned}
\widehat{\mathcal Q}_3^{\rm out/in}
&\equiv
2\bar{\eth}'\eth'\delta\widehat\Omega_3^{\rm out/in}
\\
&\quad+
\eta^{ab}\partial_a\widehat\Lambda_1^{\rm out/in}
\partial_b\widehat{\bar\Lambda}_2^{\rm out/in}
+
\eta^{ab}\partial_a\widehat\Lambda_2^{\rm out/in}
\partial_b\widehat{\bar\Lambda}_1^{\rm out/in}.
\end{aligned}
}
\label{eq:third_order_cone_source}
\end{equation}
Normal ordering of the complete operator-valued source is understood.
We write
\begin{equation}
\widehat{\mathcal Q}_3^{\rm out/in}
=
\widehat{\mathcal Q}_3^{(\Omega),\rm out/in}
+
\widehat{\mathcal Q}_3^{(\Lambda),\rm out/in},
\label{eq:Q3_cone_source_decomposition}
\end{equation}
where
\begin{align}
\widehat{\mathcal Q}_3^{(\Omega),\rm out/in}
&\equiv
2\bar{\eth}'\eth'\delta\widehat\Omega_3^{\rm out/in},
\label{eq:Q3_Omega_definition}
\\
\widehat{\mathcal Q}_3^{(\Lambda),\rm out/in}
&\equiv
\eta^{ab}\partial_a\widehat\Lambda_1^{\rm out/in}
\partial_b\widehat{\bar\Lambda}_2^{\rm out/in}
+
\eta^{ab}\partial_a\widehat\Lambda_2^{\rm out/in}
\partial_b\widehat{\bar\Lambda}_1^{\rm out/in}.
\label{eq:Q3_Lambda_definition}
\end{align}
The source in Eq.~\eqref{eq:third_order_cone_source} is therefore the
flat-kinematics part \(Z_{3,\rm cone}^{\rm YF,flat}\) retained in the
explicit calculation.  The inverse-metric term
\(h_1^{ab}\partial_a\Lambda_1\partial_b\bar\Lambda_1\), together with the
first-order trajectory and cone-measure corrections, belong respectively to
\(Z_{3,\rm cone}^{\rm metric}\) and \(Z_{3,\rm cone}^{\rm def}\) in
Eq.~\eqref{eq:Z3_cone_full_split_II}.

Within the Minkowskian approximation, the outgoing cone
contribution is
\begin{align}
\widehat{Z}_{3,\rm cone}^{\rm out}(x,z)
&=
-\oint d^2z' \, G_{0,0'}(z,z') \int_0^\infty ds \,
\left. \widehat{\mathcal{Q}}_3^{\rm out} \right|_{x=x_s^{\rm out}(z')}.
\label{eq:outgoing_cone3_explicit}
\end{align}
Equivalently, displaying its source explicitly,
\begin{align}
\widehat{Z}_{3,\rm cone}^{\rm out}(x,z)
&=
-\oint d^2z' \, G_{0,0'}(z,z') \int_0^\infty ds \,
\Big[
2\bar{\eth}'\eth' \delta\widehat{\Omega}_3^{\rm out}
\nonumber \\
&\qquad\qquad
+
\eta^{ab} \partial_a\widehat{\Lambda}_1^{\rm out} \partial_b\widehat{\bar{\Lambda}}_2^{\rm out}
+
\eta^{ab} \partial_a\widehat{\Lambda}_2^{\rm out} \partial_b\widehat{\bar{\Lambda}}_1^{\rm out}
\Big]_{x=x_s^{\rm out}(z')}.
\label{eq:outgoing_cone3_expanded}
\end{align}
The antipodally mapped incoming contribution is
\begin{align}
\mathcal{A}\widehat{Z}_{3,\rm cone}^{\rm in}(x,z)
&=
-\oint d^2z' \, G_{0,0'}(z,z') \int_0^\infty ds \,
\left. \widehat{\mathcal{Q}}_3^{\rm in} \right|_{x=\mathcal{A}x_s^{\rm in}(z')},
\label{eq:incoming_cone3_explicit}
\end{align}
or, explicitly,
\begin{align}
\mathcal{A}\widehat{Z}_{3,\rm cone}^{\rm in}(x,z)
&=
-\oint d^2z' \, G_{0,0'}(z,z') \int_0^\infty ds \,
\Big[
2\bar{\eth}'\eth' \delta\widehat{\Omega}_3^{\rm in}
\nonumber \\
&\qquad\qquad
+
\eta^{ab} \partial_a\widehat{\Lambda}_1^{\rm in} \partial_b\widehat{\bar{\Lambda}}_2^{\rm in}
+
\eta^{ab} \partial_a\widehat{\Lambda}_2^{\rm in} \partial_b\widehat{\bar{\Lambda}}_1^{\rm in}
\Big]_{x=\mathcal{A}x_s^{\rm in}(z')}.
\label{eq:incoming_cone3_expanded}
\end{align}
Here $x_s^{\rm out}(z')$ and
$\mathcal{A}x_s^{\rm in}(z')$ denote the flat outgoing and
antipodally mapped incoming null generators introduced in
Sec.~\ref{sec:perturbative_hierarchy}.

All the above integrand terms in the Minkowskian calculation are
evaluated along the zeroth-order Minkowski null generators.

Before Eqs.~\eqref{eq:outgoing_cone3_expanded} and
\eqref{eq:incoming_cone3_expanded} are compared, every lower-order
quantity is replaced by the same retarded solution determined by the
incoming asymptotic data:
\begin{equation}
\begin{aligned}
\widehat\Lambda_j^{\rm out}
&=
\widehat\Lambda_j^{\rm in}
=
\widehat\Lambda_j^{\rm ret}
\left[a_1^{\rm in},a_1^{{\rm in}\dagger}\right],
\qquad j<3,
\\
\delta\widehat\Omega_3^{\rm out}
&=
\delta\widehat\Omega_3^{\rm in}
=
\delta\widehat\Omega_3^{\rm ret}
\left[a_1^{\rm in},a_1^{{\rm in}\dagger}\right].
\end{aligned}
\label{eq:cone_lower_order_substitution}
\end{equation}
Thus the two cone integrands are the same cubic operator-valued source
written in the same incoming representation.  Their difference lies only
in the orientation of the outer affine integration, which produces the
opposite \(\pm\ii0\) prescriptions and hence the principal-value
combination below.

To compare the affine integrations, consider a Fourier branch with total
four-momentum
\begin{equation}
K_{\eta_2\eta_3}^{a}
=
k_1^{a}
+
\eta_2 k_2^{a}
+
\eta_3 k_3^{a},
\qquad
\eta_2,\eta_3=\pm1,
\label{eq:cone_branch_momentum}
\end{equation}
and signed frequency
\begin{equation}
\Omega_{\eta_2\eta_3}
=
\omega_1
+
\eta_2\omega_2
+
\eta_3\omega_3.
\label{eq:cone_branch_frequency}
\end{equation}
The outgoing and antipodally mapped incoming integrations produce the
separate affine distributions
\begin{align}
\mathscr{D}_{\eta_2\eta_3}^{\rm out}(z')
&=
\frac{i}{
l^-(z')\cdot K_{\eta_2\eta_3}+i0
},
\label{eq:D_outgoing_cone} \\
\mathscr{D}_{\eta_2\eta_3}^{\rm in}(z')
&=
\frac{i}{
l^-(z')\cdot K_{\eta_2\eta_3}-i0
}.
\label{eq:D_incoming_cone}
\end{align}
These two distributions must be retained separately until the two
geometrical cone contributions have been expressed in terms of the same
incoming operators.

Only after this common re-expression do they combine:
\begin{equation}
\boxed{
\mathscr{D}_{\eta_2\eta_3}^{\rm out}
+
\mathscr{D}_{\eta_2\eta_3}^{\rm in}
=
2i \, \mathrm{PV}
\frac{1}{
l^-(z')\cdot K_{\eta_2\eta_3}
}.
}
\label{eq:two_cones_to_PV}
\end{equation}
Thus the principal-value denominator is not assigned independently to
either cone. It is generated by the sum of the outgoing and
antipodally mapped incoming cone reconstructions, whose opposite
$i0$ prescriptions follow from their opposite null orientations.

This completes the paired-cone term retained in the Minkowskian
approximation.  Because \(\widehat U_\omega^-=\mathbf1\), there is no
recursive cut contribution proportional to
\(-\ii\varpi\,\delta a_2^{\rm out}Z_1\).  After the two additional known
free-data, metric, and cone-deformation terms have been set aside according
to Eq.~\eqref{eq:deltaa3_Minkowskian_definition_II}, the
explicit matching equation used below is
\begin{equation}
\boxed{
\widehat Z_{3,\rm cut}^{\rm out}
\!\left[\delta a_{3,M}^{\rm out}\right]
=
-
\left(
\widehat{\mathscr Z}_3^{+,\rm cone}
+
\widehat{\mathscr Z}_3^{-,\rm cone}
\right).
}
\label{eq:third_order_paired_cone_matching_II}
\end{equation}
Here
\begin{align}
\widehat{\mathscr Z}_3^{+,\rm cone}
&\equiv \widehat Z_{3,\rm cone}^{\rm out},
\\
\widehat{\mathscr Z}_3^{-,\rm cone}
&\equiv \mathcal A\widehat Z_{3,\rm cone}^{\rm in}.
\end{align}

\subsubsection{Branch decomposition, Hermitian pairing, and Bose symmetry}
\label{subsec:branch_normal_ordering_complete}

For each Fourier branch define
\begin{equation}
K_{\eta_2\eta_3}^{a}
=
k_1^a+\eta_2 k_2^a+\eta_3 k_3^a,
\qquad
\Omega_{\eta_2\eta_3}
=
\omega_1+\eta_2\omega_2+\eta_3\omega_3,
\qquad
\eta_2,\eta_3=\pm1.
\label{eq:branch_sign_operator_correspondence}
\end{equation}
With
\begin{equation}
a_i\equiv a_{1,\lambda_i}^{\rm in}(\omega_i,z_i),
\qquad
a_i^\dagger\equiv a_{1,\lambda_i}^{{\rm in}\dagger}(\omega_i,z_i),
\qquad i=1,2,3,
\label{eq:linear_operator_shorthand}
\end{equation}
the four normally ordered representatives are
\begin{equation}
\boxed{
\begin{aligned}
\widehat{\mathcal O}_{++}(1,2,3)&=a_1a_2a_3,
\\
\widehat{\mathcal O}_{+-}(1,2,3)&=a_3^\dagger a_1a_2,
\\
\widehat{\mathcal O}_{-+}(1,2,3)&=a_2^\dagger a_1a_3,
\\
\widehat{\mathcal O}_{--}(1,2,3)&=a_2^\dagger a_3^\dagger a_1.
\end{aligned}
}
\label{eq:cubic_branch_operator_monomials}
\end{equation}
The complementary frequency sector is supplied by Hermitian conjugation.

The two cone sources are placed on the same three-particle measure:
\begin{align}
\widehat{\mathscr Z}_3^{(r)}(x,z)
&=
\sum_{\eta_2,\eta_3=\pm1}
\sum_{\lambda_1,\lambda_2,\lambda_3}
\int d\mu_1\,d\mu_2\,d\mu_3
\Big[
\mathcal C_{\eta_2\eta_3}^{(r)}(z;1,2,3)
\widehat{\mathcal O}_{\eta_2\eta_3}(1,2,3)
\nonumber\\
&\qquad\qquad\qquad\qquad\times
e^{-\ii\Omega_{\eta_2\eta_3}t+\ii\vec K_{\eta_2\eta_3}\cdot\vec x}
+\mathrm{H.c.}
\Big],
\nonumber\\
&\hspace{50mm} r=+\mathrm{cone},-\mathrm{cone}.
\label{eq:paired_cone_branch_decomposition}
\end{align}
After the retarded Yang--Feldman substitution, the two kernels contain the
same source functional and differ only in the affine prescription.  Their
sum is the principal-value paired-cone kernel of
Eq.~\eqref{eq:two_cones_to_PV}.  Bose symmetrization is imposed only after
the two cone kernels have been combined.

\subsubsection{Radiative projection of the paired-cone source}
\label{subsec:radiative_projection_paired_cone}

For each branch define the future null momentum with the same spatial
component,
\begin{equation}
p_{\eta_2\eta_3}^{a}
=
\left(
\left|\vec K_{\eta_2\eta_3}\right|,
\vec K_{\eta_2\eta_3}
\right),
\qquad
p_{\eta_2\eta_3}^{2}=0.
\label{eq:radiative_p_eta}
\end{equation}
The radiative prescription replaces the branch phase by
\(e^{-\ii p_{\eta_2\eta_3}\cdot x}\) before the positive- and
negative-frequency sectors are identified, so the normalized
Cauchy--Klein--Gordon factor is one.

The complete Bose-symmetrized source kernel in the Minkowskian approximation
is
\begin{equation}
\boxed{
\mathcal C_{\eta_2\eta_3}^{(3),\mathrm{tot}}
=
\mathcal C_{\eta_2\eta_3}^{+,\mathrm{cone},\mathrm B}
+
\mathcal C_{\eta_2\eta_3}^{-,\mathrm{cone},\mathrm B}.
}
\label{eq:C3_total_paired_cones}
\end{equation}
Equivalently, the two terms may be combined first into the
principal-value kernel of Eq.~\eqref{eq:two_cones_to_PV}.  There is no
additional recursive-cut summand because
\(\widehat U_\omega^-=\mathbf1\).

\subsubsection{Helicity extraction and the retained cubic operator}
\label{subsec:complete_deltaa3_radiative}
The helicity projectors act only after the complete scalar source has
been constructed. The radiative third-order vertices are
\begin{align}
\mathcal{V}_{\eta_2\eta_3}^{(3),+,\mathrm{geom}}(q;1,2,3)
&=
-\left.
\eth_z^2
\left[
\bigl(l(z)\cdot\hat{q}\bigr)
\mathcal{C}_{\eta_2\eta_3}^{(3),\mathrm{tot}}(z;1,2,3)
\right]
\right|_{z=\hat{q}},
\label{eq:V3_plus_radiative_complete} \\
\mathcal{V}_{\eta_2\eta_3}^{(3),-,\mathrm{geom}}(q;1,2,3)
&=
-\left.
\bar{\eth}_z^2
\left[
\bigl(l(z)\cdot\hat{q}\bigr)
\mathcal{C}_{\eta_2\eta_3}^{(3),\mathrm{tot}}(z;1,2,3)
\right]
\right|_{z=\hat{q}}.
\label{eq:V3_minus_radiative_complete}
\end{align}
The angular derivatives are taken before setting $z=\hat{q}$.  These
quantities are the reduced geometric vertices: the normalization factors of
the radiative mode expansion have not yet been included.  With
\begin{equation}
 d\mu(k)=\frac{d^3k}{ 2\omega},
 \qquad
 d^2z\,\frac{d\omega}{2\pi}
 \sqrt{2\pi G}
 =d\mu(k)\,\nu(\omega),
 \qquad
 \nu(\omega)=\frac{\sqrt{2\pi G}}{\pi}\,\omega^{-1},
\label{eq:mode_to_Lorentz_measure_factor}
\end{equation}
This is the rescaled-oscillator form of the measure conversion in Part~III~\cite{KDIII2026}.
The factor $\nu$ rewrites the original shear measure; it is not an
independent additional measure or an operator redefinition.
The fully normalized third-order vertex is denoted by
\begin{equation}
\mathcal W_{M,\eta_2\eta_3}^{(3),\lambda}(q;1,2,3)
=
\mathcal N^{(3)}(q;1,2,3)\,
\mathcal V_{\eta_2\eta_3}^{(3),\lambda,\mathrm{geom}}(q;1,2,3).
\label{eq:W3_full_normalized_vertex}
\end{equation}
The factor $\mathcal N^{(3)}$ is the normalization inherited from the three
radiative mode expansions and from the extraction of the outgoing mode.
Explicitly, using the factor displayed in the BCH reconstruction
\cite{KozamehBCH2026},
\begin{equation}
\boxed{
\mathcal N^{(3)}(q;1,2,3)
=\frac{2\omega_q}{\nu(\omega_q)}
\nu(\omega_1)\nu(\omega_2)\nu(\omega_3),
\qquad \omega_q=|\vec q|.
}
\label{eq:N3_explicit_BCH_II}
\end{equation}
Here $\nu$ is fixed by Eq.~\eqref{eq:mode_to_Lorentz_measure_factor}.
We use $d\mu=d^3k/(2\omega)$ and $\delta^{(3)}$ throughout. All
mode weights in this expression refer to the oscillators defined above.
The normalization $\mathcal N^{(3)}$ does not contain any Wick multiplicity
or external-state contraction.
In the ordered one-loop routing of Sec.~\ref{subsec:loop_from_deltaa3}, where
legs $2$ and $3$ are the internal lines,
\begin{equation}
\left.\mathcal N^{(3)}(q;1,2,3)\right|_{\rm loop}
=
\mathcal N_{\rm ext}(q;1)\,\nu(\omega_2)\nu(\omega_3),
\label{eq:N3_loop_factorization}
\end{equation}
where the external factor is now explicit:
\begin{equation}
\boxed{
\mathcal N_{\rm ext}(q;1)
=\frac{2\omega_q}{\nu(\omega_q)}\nu(\omega_1)
=2\omega_q\left(\frac{\omega_q}{\omega_1}\right).
}
\label{eq:N3_external_explicit_II}
\end{equation}
It is independent of the loop radius when the external momenta are fixed.
Equations
\eqref{eq:W3_full_normalized_vertex} and
\eqref{eq:N3_loop_factorization} make explicit the internal mode factors that
were previously left implicit when the cubic operator was written with the
measure $d\mu_i$.

The normally ordered cubic correction in the Minkowskian approximation is
therefore
\begin{equation}
\boxed{
\begin{aligned}
\delta a_{3,M,\lambda}^{\mathrm{out}}(q)
&=
\sum_{\lambda_1,\lambda_2,\lambda_3}
\sum_{\eta_2,\eta_3=\pm1}
\int d\mu_1 \, d\mu_2 \, d\mu_3 \,
\delta^{(3)}
\left(
\vec{q}
-\vec{k}_1
-\eta_2\vec{k}_2
-\eta_3\vec{k}_3
\right)
\\
&\quad\times
\mathcal{W}_{M,\eta_2\eta_3}^{(3),\lambda}(q;1,2,3)
\widehat{\mathcal{O}}_{\eta_2\eta_3}(1,2,3)
\end{aligned}
}
\label{eq:deltaa3_retained_flatcone_II}
\end{equation}
Here
\begin{equation}
d\mu_n
=
\frac{d^3k_n}{ 2\omega_n},
\qquad
\omega_n=|\vec{k}_n|,
\qquad
n=1,2,3.
\label{eq:dmu_complete_deltaa3}
\end{equation}
The kernels collected in Appendix~\ref{app:explicit_C3_kernels} are the
reduced paired-cone kernels entering
\(\mathcal V^{(3),\mathrm{geom}}\) within the Minkowskian approximation.  Their
missing mode-normalization factors belong to \(\mathcal N^{(3)}\) and
therefore to the normalized vertex \(\mathcal W_M^{(3)}\); they must neither
be omitted nor counted again inside the reduced kernels.  Equation
\eqref{eq:deltaa3_retained_flatcone_II} should thus be read as the
explicitly evaluated contribution
\(\delta a_{3,M}^{\rm out}\), not as the full third-order NSF
operator, which also contains the other sources in
Eq.~\eqref{eq:deltaa3_source_decomposition_II} and the recursive cut term
in Eq.~\eqref{eq:third_order_general_cut_coefficient_II}.
For the canonical test, the same cubic contribution is regrouped into
the explicitly Bose-symmetrized form
\eqref{eq:unitarity_symmetrized_deltaa3_II}, with the Wick-induced linear
term included.
Its Hermitian conjugate is not extracted independently: it is the
negative-frequency half generated by
Sec.~\ref{subsec:radiative_projection_paired_cone},
\begin{equation}
\delta a^{\mathrm{out}\dagger}_{3,M,\lambda}(q)
=
\left[
\delta a^{\mathrm{out}}_{3,M,\lambda}(q)
\right]^\dagger,
\label{eq:deltaa3_dagger_complete}
\end{equation}
with the usual interchange between the shear and conjugate-shear
helicity sectors.

\subsubsection{Explicit paired-cone kernels}
\label{subsec:explicit_paired_cone_kernels_main}
For later use, after the two cone contributions have been combined into
the principal value of Eq.~\eqref{eq:two_cones_to_PV}, their explicit
cubic kernels are collected in
Appendix~\ref{app:explicit_C3_kernels}.  In the Minkowskian approximation
these paired-cone kernels constitute the third-order source retained in the
explicit radiative and helicity projections.

\subsubsection{The loop generated by
\texorpdfstring{$\delta a_{3,M}\delta a_{3,M}$}{delta a3M delta a3M}}
\label{subsec:loop_from_deltaa3}

The connected cubic correction has the form
\begin{align}
\delta a^{\rm out}_{3,M,\lambda}(q)
={}&
\sum_{\lambda_1,\lambda_2,\lambda_3}
\sum_{\eta_2,\eta_3=\pm1}
\int d\mu_1\,d\mu_2\,d\mu_3\,
\delta^{(3)}\!\left(
\vec q-\vec k_1-\eta_2\vec k_2-\eta_3\vec k_3
\right)
\nonumber\\
&\times
\mathcal W^{(3),\lambda}_{M,\eta_2\eta_3}(q;1,2,3)
\widehat{\mathcal O}_{\eta_2\eta_3}(1,2,3),
\label{eq:deltaa3_vertex_compact}
\end{align}
where
\begin{equation}
d\mu_i=\frac{d^3k_i}{2\omega_i},
\qquad
\omega_i=|\vec k_i|,
\label{eq:loop_covariant_measure}
\end{equation}
and
\begin{align}
\widehat{\mathcal O}_{++}&=a_1a_2a_3,
&
\widehat{\mathcal O}_{+-}&=a_3^\dagger a_1a_2,
\nonumber\\
\widehat{\mathcal O}_{-+}&=a_2^\dagger a_1a_3,
&
\widehat{\mathcal O}_{--}&=a_2^\dagger a_3^\dagger a_1.
\label{eq:loop_cubic_monomials}
\end{align}
Within the Minkowskian approximation, the reduced geometric vertices
are obtained after the future and antipodally matched past cone terms
have been added; the recursive cut-phase term vanishes because
\(\widehat U=1\):
\begin{align}
\mathcal V^{(3),+,\mathrm{geom}}_{\eta_2\eta_3}
&=
-\left.
\eth_z^2\!\left[
(l(z)\!\cdot\!\widehat q)
\mathcal C^{(3),\mathrm{tot}}_{\eta_2\eta_3}
\right]\right|_{z=\widehat q},
\label{eq:V3_plus_def}
\\
\mathcal V^{(3),-,\mathrm{geom}}_{\eta_2\eta_3}
&=
-\left.
\bar\eth_z^2\!\left[
(l(z)\!\cdot\!\widehat q)
\mathcal C^{(3),\mathrm{tot}}_{\eta_2\eta_3}
\right]\right|_{z=\widehat q}.
\label{eq:V3_minus_def}
\end{align}
The corresponding vertices in the operator are
$\mathcal W_M^{(3)}=\mathcal N^{(3)}\mathcal V^{(3),\mathrm{geom}}$.

For the Wightman ordering displayed in
Eq.~\eqref{eq:M33_sector_definition}, the leftmost cubic monomial cannot
contain a creation operator because $\langle0|a^\dagger=0$.  The only
connected ordered pair is therefore
\begin{equation}
\widehat{\mathcal O}_{++}(1,2,3)
\widehat{\mathcal O}_{--}(1',2',3')
=
a_1a_2a_3\,{a'}_2^\dagger {a'}_3^\dagger {a'}_1,
\label{eq:surviving_ordered_pair_main}
\end{equation}
in this fixed ordering of the outgoing insertions. The mixed pairs
have the correct total number of creation operators but vanish in this fixed
ordering: after the right mixed vertex acts on the two-particle state, only
one quantum remains, whereas the left mixed vertex contains two
annihilators.

For an exact Wick reduction, define the Bose-averaged operator kernels
\begin{align}
\mathcal W^{(3),B}_{M,++}(q;1,2,3)
&=\frac{1}{3!}\sum_{\pi\in S_3}
\mathcal W^{(3)}_{M,++}(q;\pi(1),\pi(2),\pi(3)),\nonumber\\
\mathcal W^{(3),B}_{M,--}(q;1,2,3)
&=\frac12\left[\mathcal W^{(3)}_{M,--}(q;1,2,3)
+\mathcal W^{(3)}_{M,--}(q;1,3,2)\right].
\label{eq:loop_Bose_averaged_vertices_II}
\end{align}
Each permutation acts on both momentum and helicity labels. These averages
leave the integrated normal-ordered operators unchanged; no factorial is
included in their monomials. For each fixed assignment of the two incoming
legs, there are three choices for the externally contracted annihilator
on the left and two pairings of the remaining annihilators with the
creators on the right. Thus there are \(3\times2=6\) contractions.
The four canonical contractions, two external and two internal, each give
\(\chi=1\) when one integration with \(d\mu\) is
performed. Their combined normalization is \(\chi^4\).

A representative connected contraction can be routed as
\begin{equation}
k_1=p_i,
\qquad
k'_1=p_j,
\qquad
k_2=k'_2=k,
\qquad
k_3=k'_3=r,
\qquad
\{i,j\}=\{1,2\}.
\label{eq:loop_routing_main}
\end{equation}
The two spatial delta functions carried by the outgoing operators then become
\begin{align}
\Delta_1
&=\delta^{(3)}\!\left(
\vec p'_1-\vec p_i-\vec k-\vec r
\right),
\label{eq:Delta1_loop_main}
\\
\Delta_2
&=\delta^{(3)}\!\left(
\vec p'_2-\vec p_j+\vec k+\vec r
\right).
\label{eq:Delta2_loop_main}
\end{align}
Their product factorizes as
\begin{align}
\Delta_1\Delta_2
={}&
\delta^{(3)}\!\left(
\vec p'_1+\vec p'_2-\vec p_1-\vec p_2
\right)
\nonumber\\
&\times
\delta^{(3)}\!\left(
\vec P_i-\vec k-\vec r
\right),
\qquad
\vec P_i\equiv\vec p'_1-\vec p_i
=\vec p_j-\vec p'_2.
\label{eq:two_vertex_deltas_reduced_main}
\end{align}
Using the second delta to set
$\vec r=\vec P_i-\vec k$ gives
\begin{align}
\mathcal A_M^{(33)}
\equiv\mathcal M_M^{(33)}
={}&
\delta^{(3)}\!\left(
\vec p'_1+\vec p'_2-\vec p_1-\vec p_2
\right)
\mathcal M_{M,\rm red}^{(33)},
\label{eq:A33_momentum_conservation}
\\
\mathcal M_{M,\rm red}^{(33)}
={}&
6\chi^4
\sum_{\substack{i,j=1,2\\i\neq j}}
\sum_{\lambda_k,\lambda_r}
\int
\frac{d^3k}
{2\omega_k\,2\omega_{P_i-k}}
\nonumber\\
&\times
\mathcal W^{(3),B,\lambda'_1}_{M,++}
\left(p'_1;p_i,k,P_i-k\right)
\mathcal W^{(3),B,\lambda'_2}_{M,--}
\left(p'_2;p_j,k,P_i-k\right).
\label{eq:M33_single_loop_integral}
\end{align}
Here \(\mathcal M_M^{(33)}\) retains the definition as the full matrix
element, while \(\mathcal M_{M,\rm red}^{(33)}\) has its overall spatial
delta function removed. Exchanging \(p'_1\) and \(p'_2\) gives the
oppositely ordered matrix element; it is not an additional Wick contraction
of the ordered product defined here. If a symmetric contribution is
desired, it is the half-sum of these two orderings. Bose symmetry of the
complete scattering matrix concerns the sum of all sectors at the same
perturbative order.
The three-vector $\vec k$ is not fixed by any remaining delta function.
Thus Eq.~\eqref{eq:M33_single_loop_integral} contains a genuine one-loop
integration. Equivalently, the connected graph has two effective four-leg vertices and
two internal lines, and hence
\begin{equation}
L=I-V+1=2-2+1=1.
\label{eq:loop_topological_count}
\end{equation}
The two choices $i=1,2$ are the two crossed external attachments.  In this
ordered contraction each vertex contains one incoming and one outgoing
external leg; an independent $s$-channel assignment is not generated by
Eq.~\eqref{eq:surviving_ordered_pair_main}.

\subsubsection{Ultraviolet behavior}
\label{subsec:uv_power_counting_deltaa3}

\paragraph{Mode normalization, canonical contractions, and the plane-wave radial limit.}

We use the equivalent rescaled-oscillator form of the contraction rule of Part~III, Sec.~V.A
\cite{KDIII2026}. For one internal line joining two vertices, let
$F(k,k')$ denote the remaining c-number kernel. Equations
\eqref{eq:asymptotic_quantization_commutator} and
\eqref{eq:mode_to_Lorentz_measure_factor} give, for a fixed helicity,
\begin{align}
&\int d\mu(k)\,d\mu(k')\,
\nu(\omega_k)\nu(\omega_{k'})
[a_{1,\lambda}^{\rm in}(k),a_{1,\lambda}^{{\rm in}\dagger}(k')]
F(k,k')
\nonumber\\
&\hspace{12mm}=
\int d\mu(k)\,\nu(\omega_k)^2 F(k,k).
\label{eq:UV_complete_internal_contraction_II}
\end{align}
For different helicities the corresponding Kronecker delta is understood.
The factor $2\omega_k$ from the CCR cancels the $2\omega_{k^\prime}$ in the
integrated measure. Therefore the complete weight of a contracted internal
line, up to momentum-independent constants, is
\begin{equation}
\boxed{
d\mu(k)\,\nu(\omega_k)^2
\propto \frac{d^3k}{\omega_k^3}
=\frac{d\omega_k\,d^2\widehat k}{\omega_k}.
}
\label{eq:UV_internal_line_rule_II}
\end{equation}
This includes the surviving measure, both mode weights, and the canonical
contraction. None of these factors is to be inserted a second time.

Let
\begin{equation}
\vec k=K\widehat k,\qquad K\equiv|\vec k|,
\qquad
\vec r=\vec P_i-K\widehat k,
\qquad
K\longrightarrow\infty.
\label{eq:UV_assignment_main}
\end{equation}
Then
\begin{equation}
\frac{d^3k}{2\omega_k\,2\omega_r}
\sim
\frac14\,dK\,d^2\widehat k.
\label{eq:on_shell_measure_uv_deltaa3}
\end{equation}
Thus the on-shell two-particle measure has radial degree zero. In this
form of the loop integral the mode weights remain in the normalized
vertices $\mathcal W_M^{(3)}$. Equivalently, one may use
Eq.~\eqref{eq:UV_internal_line_rule_II} for each internal line and keep
only the reduced geometric vertices $\mathcal V^{(3),\mathrm{geom}}$.
These are alternative organizations of the same factors, not additional
factors to be multiplied together.

The additional factor arises when each shear mode is rewritten with
$d\mu(k)$.  Equation~\eqref{eq:mode_to_Lorentz_measure_factor} gives
\begin{equation}
\nu(\omega_k)\nu(\omega_r)
=
\left(\frac{\sqrt{2\pi G}}{\pi}\right)^2
(\omega_k\omega_r)^{-1}
=
\mathcal O(K^{-2})
\label{eq:two_internal_mode_weights_UV}
\end{equation}
for the two internal legs of each cubic vertex.  These factors are not part of
the reduced kernels $\mathcal T_3$ displayed in the appendix; they are part
of the full vertex $\mathcal W_M^{(3)}$.

The radial degree of the reduced third-order vertex follows from its
recursive cone structure. At second order the retarded--advanced affine
integration supplies one factor $2i\,\mathrm{PV}(1/D_{2,1})$, where
$D_{2,1}$ is linear in a signed momentum combination. At third order
this denominator is inherited from the second-order field, and the new
cone integration supplies a second principal-value factor. After the
lower-order fields have been substituted and the radial derivatives
simplified, the complete third-order branches have the schematic form
\begin{equation}
\mathcal V^{(3),\mathrm{geom}}_b
\sim\mathcal N_{3,b}
\left[2i\,\mathrm{PV}\frac{1}{D_{3,b,1}}\right]
\left[2i\,\mathrm{PV}\frac{1}{D_{3,b,2}}\right],
\label{eq:UV_recursive_PV_structure_II}
\end{equation}
where $b$ labels the branches. The numerator includes the momentum
contractions and the derivatives acting on the lower-order fields; its
third-order radial degree is at most three. The two elementary
denominators each have degree one. Thus the numerator and the complete
recursive denominator have degrees
\begin{equation}
\mathcal N_{3,b}=\mathcal O(K^3),\qquad
\mathcal D_{3,b}\equiv D_{3,b,1}D_{3,b,2}\sim K^2.
\label{eq:UV_numerator_denominator_degrees_II}
\end{equation}
In particular, retaining only the outer affine denominator would omit
one inverse power of $K$ already present in the second-order field.
Angular Green functions and helicity projectors have radial degree zero;
momentum powers produced by angular differentiation are included in
$\mathcal N_{3,b}$. The radiative extraction also has net radial degree
zero, since its derivative is accompanied by the on-shell frequency
normalization. Combining the numerator and both denominators gives
\begin{equation}
\mathcal V^{(3),\mathrm{geom}}(p;k,r)
\sim K^3 K^{-2}=\mathcal O(K).
\label{eq:V3_geometric_UV_degree}
\end{equation}
The domain of this counting is the uniform radial region of the routing
\eqref{eq:UV_assignment_main}: each denominator entering the branch has
$D_{3,b,\alpha}=K d_{b,\alpha}+O(1)$ with nonzero leading coefficient.
On an angular domain where these coefficients stay away from zero,
the numerator bound and the two inverse denominators therefore give
$|\mathcal V^{(3),\mathrm{geom}}|\le C K$, with $C$ independent of $K$.
This estimate is the consequence of the preceding power count.
Configurations in which a signed momentum combination stays soft or a
principal-value denominator vanishes are separate limits, not included
in this radial estimate. Combining Eqs.~\eqref{eq:N3_loop_factorization},
\eqref{eq:two_internal_mode_weights_UV}, and
\eqref{eq:V3_geometric_UV_degree}, one obtains
\begin{equation}
\boxed{
\mathcal W_M^{(3)}(p;k,r)
=
\mathcal O(K^{-1}) .
}
\label{eq:W3_full_UV_degree}
\end{equation}
The fixed external normalization $\mathcal N_{\rm ext}$ does not affect this
radial degree.

The product of the two fully normalized vertices in
Eq.~\eqref{eq:M33_single_loop_integral} behaves as
\begin{equation}
\mathcal W_{M,L}^{(3)}\mathcal W_{M,R}^{(3)}
=
\mathcal O(K^{-2}),
\label{eq:two_W3_UV_degree}
\end{equation}
and the radial ultraviolet tail satisfies the following estimate in that
uniform hard region (with fixed generic angular data, or angular smearings
for which the reduced-kernel bound is uniform):
\begin{equation}
\left|\mathcal M_{M,{\rm red},K>K_{UV}}^{(33)}\right|
\leq
C\int_{K_{UV}}^{\infty}\frac{dK}{K^2}
=
\frac{C}{K_{UV}}.
\label{eq:M33_plane_wave_UV_bound}
\end{equation}
We therefore obtain
\begin{equation}
\boxed{
\begin{gathered}
\mathcal M_{M,\rm red}^{(33)}\text{ is radially ultraviolet finite}\\
\text{in the uniform radial region specified above.}
\end{gathered}
}
\label{eq:M33_plane_wave_UV_finite}
\end{equation}
The same counting can be checked directly with the contracted-line rule.
For the two internal lines in $\mathcal M_M^{(33)}$, the product of weights
and the remaining spatial constraint gives
\begin{equation}
\int\frac{d^3k\,d^3r}{\omega_k^3\omega_r^3}
\delta^{(3)}(\vec P_i-\vec k-\vec r)
\mathcal V_L^{(3),\mathrm{geom}}
\mathcal V_R^{(3),\mathrm{geom}}.
\label{eq:UV_reduced_vertex_integral_II}
\end{equation}
After eliminating $\vec r$, its radial degree is
$K^2dK\,K^{-6}K^2=dK/K^2$, in agreement with
Eq.~\eqref{eq:M33_plane_wave_UV_bound}. Fixed external factors and numerical
constants have been suppressed only in this power-counting expression.

This direct one-loop estimate concerns the Minkowskian contribution and does
not require adding the free-data, metric, or cone-deformation terms omitted
in Eq.~\eqref{eq:deltaa3_Minkowskian_definition_II}. Nonuniform
subloop regions and singular angular or principal-value configurations
require separate analysis and are not covered by the uniform radial bound.

\paragraph{Physical smeared states.}
\label{subsubsec:uv_smeared_M33}

The plane-wave result concerns the radial ultraviolet tail of the kernel.
The radiative creation and annihilation operators remain operator-valued
distributions, so physical states are defined by smearing them against smooth
profiles,
\begin{equation}
a_\lambda[f]
\equiv
\int d\mu(k)\,f_\lambda(k)a_\lambda(k),
\qquad
f_\lambda\in\mathcal S(\mathbb R^3),
\label{eq:smeared_radiative_operator_M33}
\end{equation}
or against smooth functions of compact support.  This smearing controls the
separate angular-coincidence and distributional structures and implements the
physical NSF state space.  It is not being used as an ultraviolet regulator
for Eq.~\eqref{eq:M33_plane_wave_UV_finite}.

\subsubsection{Infrared behavior}
\label{subsec:ir_power_counting_deltaa3}

For the infrared endpoints there are two distinct limits,
\begin{align}
\vec{k}=\rho\widehat{k},
\quad \rho\to0,
\quad \vec{r}\to\vec{P},
\label{app:eq:IR_k_soft}
\\
\vec{r}=\rho\widehat{r},
\quad \rho\to0,
\quad \vec{k}\to\vec{P}.
\label{app:eq:IR_r_soft}
\end{align}
These are two endpoints of the same one-loop integral, but they do not
require independent power-counting calculations.  The Bose-averaged vertices
defined in Eq.~\eqref{eq:loop_Bose_averaged_vertices_II} are invariant under
the exchange of the two internal legs,
\begin{equation}
 \vec{k}\longleftrightarrow\vec{r}=\vec{P}-\vec{k}.
 \label{app:eq:IR_internal_Bose_exchange}
\end{equation}
Indeed, $\mathcal W_{M,++}^{(3),B}$ is symmetric under all permutations of its
three annihilation legs, while $\mathcal W_{M,--}^{(3),B}$ is symmetric under
the exchange of its two creation legs.  The measure in
Eq.~\eqref{app:eq:M33_single_loop} is invariant under the change of variable
$\vec{k}\mapsto\vec{P}-\vec{k}$.  Consequently, the limits
$\vec{k}\to0$ and $\vec{r}\to0$ have identical radial degree.  It is
therefore sufficient to analyze one of them explicitly; the result for the
other follows by Bose symmetry.

Near either endpoint the phase-space measure behaves as
\begin{equation}
\frac{d^3k}{2\omega_k\,2\omega_r}
 =
\frac{\rho\,d\rho\,d^2\widehat{k}}{4\omega_P}
\left[1+\mathcal O(\rho)\right],
\qquad \omega_P\equiv |\vec P|,
\label{app:eq:IR_measure}
\end{equation}
where in the first endpoint we have set
$k^a=\sqrt2\,\rho l^{+a}(\widehat{k})$ and
$r^a=\sqrt2\,\omega_r l^{+a}(\widehat r)$, with
$\omega_r=\omega_P+\mathcal O(\rho)$ and
$\widehat r=\widehat P+\mathcal O(\rho)$.  Thus $r=P$ may be used in the
leading soft coefficient, but $r=P-k$ is retained in the exact loop
routing.

We now determine the radial degree of the explicit Bose-symmetrized
vertices.  Define
\begin{equation}
A(\widehat k,\widehat P)
\equiv
l^+(\widehat k)\mathbin{\cdot}l^+(\widehat P),
\label{app:eq:IR_A_definition}
\end{equation}
which is nonzero for generic non-collinear directions.  The potentially
leading contribution is the sum-channel factor
$S_B$ in Eq.~\eqref{app:eq:S_B_explicit}.  For the ordering in which its
first argument is soft, the angular delta distribution sets
$z=\widehat k$ and gives
\begin{align}
S_B(k,r;z)
&=
\frac{A(\widehat k,\widehat P)}
     {l^+(z)\cdot(k+r)}
\delta^2(z-\widehat k)G_{-2,-2}(z,\widehat P)
\nonumber\\
&=
\frac{1}{\sqrt2\,\omega_P}\,
\delta^2(z-\widehat k)G_{-2,-2}(z,\widehat P)
+\mathcal O(\rho),
\label{app:eq:IR_SB_soft_first}
\end{align}
and hence this ordering is $\mathcal O(1)$.  The Bose average also contains
the exchanged ordering.  In that term the delta distribution sets
$z=\widehat r=\widehat P+\mathcal O(\rho)$, so that
\begin{align}
l^+(\widehat r)\cdot(r+k)
&=l^+(\widehat r)\cdot k
=\sqrt2\,\rho A(\widehat P,\widehat k)+\mathcal O(\rho^2),
\nonumber\\
S_B(r,k;z)
&=
\frac{1}{\sqrt2\,\rho}\,
\delta^2(z-\widehat P)G_{-2,-2}(z,\widehat k)
+\mathcal O(1).
\label{app:eq:IR_SB_soft_second}
\end{align}
If the additional principal-value denominator in
$S^-_{lI,B}$ is $\mathcal O(1)$ in radial degree on the support of the angular delta, the following power count applies. In particular the additional denominator must be tested after localization; here we state the count for that sector.

\begin{equation}
S^-_{lI,B}(k,r)=\mathcal O(1),
\qquad
S^-_{lI,B}(r,k)=\mathcal O(\rho^{-1}),
\qquad
\mathcal H^{B}_{\rm sum}(k,r)=\mathcal O(\rho^{-1}).
\label{app:eq:IR_Hsum_degree}
\end{equation}
The remaining explicitly displayed second-order factors are less singular:
\begin{equation}
F_\Omega^{(+,+)}(k,r)=\mathcal O(\rho),
\qquad
F_\Lambda^{(+,+)}(k,r)=k\cdot r=\mathcal O(\rho),
\label{app:eq:IR_other_form_factors}
\end{equation}
and the known-cut terms introduce no additional negative power of $\rho$.
The momentum factors and recursive principal-value denominators exterior to
$\mathcal H_{\rm sum}$ in
Eqs.~\eqref{app:eq:T3_II} and \eqref{app:eq:T3_IV} are
$\mathcal O(1)$ in this limit.  Bose symmetrization includes orderings with
the behavior in Eq.~\eqref{app:eq:IR_SB_soft_second}; therefore the reduced
geometrical vertices obey
\begin{equation}
\mathcal V_{++}^{(3),B,\mathrm{geom}}
=\mathcal O(\rho^{-1}),
\qquad
\mathcal V_{--}^{(3),B,\mathrm{geom}}
=\mathcal O(\rho^{-1}).
\label{app:eq:IR_geometric_vertices_degree}
\end{equation}
The final $\eth^2$ or $\bar\eth^2$ helicity extraction has radial degree
zero and does not change this result.

At each vertex, Eq.~\eqref{eq:N3_loop_factorization} supplies
\begin{equation}
\nu(\omega_k)\nu(\omega_r)
=\mathcal O(\rho^{-1}),
\label{app:eq:IR_normalization_degree}
\end{equation}
because $\nu(\omega)=\sqrt{2\pi G}/(\pi\omega)$ and
$\omega_r\to\omega_P\ne0$.  Hence
\begin{align}
\mathcal W_{M,++}^{(3),B}
&=\mathcal O(\rho^{-2}),
&
\mathcal W_{M,--}^{(3),B}
&=\mathcal O(\rho^{-2}),
\nonumber\\
\mathcal W_{M,++}^{(3),B}\mathcal W_{M,--}^{(3),B}
&=\mathcal O(\rho^{-4}).
\label{app:eq:IR_full_vertices_degree}
\end{align}
Thus the exponent introduced in the general convergence criterion is
\begin{equation}
\boxed{\gamma=4},
\label{app:eq:IR_gamma_explicit}
\end{equation}
and the regulated soft radial integral behaves as
\begin{equation}
\int_{\varepsilon_{\rm IR}} d\rho\,
\rho\,\rho^{-4}
=
\int_{\varepsilon_{\rm IR}}\frac{d\rho}{\rho^3}
\sim
\frac{1}{2\varepsilon_{\rm IR}^{2}}.
\label{app:eq:IR_criterion_integral}
\end{equation}
The bare plane-wave $\mathcal M_M^{(33)}$ contribution has a quadratic radial infrared divergence if the leading angular coefficient survives Bose symmetrization and projection. The big-$O$ bounds alone establish only the maximal degree.  The cutoff $\varepsilon_{\rm IR}$ in
Eq.~\eqref{app:eq:IR_criterion_integral} is an explicit regularization of
this endpoint and makes the integral finite for every
$\varepsilon_{\rm IR}>0$, while displaying the quadratic power count as the
cutoff is removed.  The origin of this divergence is the soft internal-graviton
endpoint, and it is therefore the counterpart, within the present
Minkowskian NSF benchmark, of the infrared problem encountered in
perturbative quantum gravity.  A natural way to proceed is consequently to
apply the same infrared prescriptions used there: one first regulates the
soft region and then constructs an infrared-finite physical quantity either
by summing inclusively over experimentally unresolved soft-graviton emission
or by using appropriately dressed asymptotic states
\cite{Weinberg1965IR,WareSaotomeAkhoury2013}.
There is also a possible and potentially subtle connection between this
soft sector and the infrared sectors of asymptotic quantum gravity described by
Ashtekar~\cite{Ashtekar1987,Ashtekar2018}.  In general, the asymptotic data
used in our construction have nonvanishing memory,
$\Delta\sigma^{+}_{ab}\neq0$, and therefore need not belong to the trivial
infrared sector.  It is thus natural to ask whether the soft configurations
required for an infrared-finite completion of the present scattering
calculation can be formulated as states in the corresponding Ashtekar
infrared sector.  We do not attempt to identify these two structures here.
Establishing their precise relation requires connecting
the $\rho\to0$ behavior of the loop kernel to the asymptotic shear data and
charges.  This is a delicate problem that lies beyond the isolated $(33)$
contribution considered here and will be addressed in future work.
The Gaussian smearing used to define the physical states regularizes the
separate distributional equal-momentum limit discussed in Part~I.  We do not
claim here that it is sufficient to remove the internal radial divergence
found above.  Demonstrating the cancellation of the regulator dependence in
the present formalism requires the complete set of degenerate real and
virtual contributions at the same perturbative order, rather than the
isolated $(33)$ contribution calculated here.  We postpone that
perturbative-quantum-gravity-type infrared completion to future work.  The
result established here is that the bare Minkowskian $(33)$ kernel has the indicated soft power count and admits the explicit cutoff
regularization displayed above.  This does not alter one
of the central motivations and results of the present work, namely the
absence of a radial ultraviolet divergence:
Eqs.~\eqref{app:eq:UV_assignment}--
\eqref{app:eq:UV_convergent_bound} establish radial ultraviolet finiteness
without smearing.  Angular-coincidence singularities are separate
celestial-distribution effects and are not part of the radial infrared
divergence isolated here.

\section{Canonical commutation relations}
\label{sec:unitarity}

The minimal canonical consistency requirements on the outgoing radiative
operators are that they satisfy the same free-field commutation relations
as the incoming operators and preserve the adjoint operation. Both
requirements are satisfied through the order considered here. We write
the commutation relations explicitly through $\delta a_3$, together with
their perturbative expansion.

For this canonical test we use relative perturbative orders: the common
factor $\epsilon$ carried by the radiative shear in
Eq.~\eqref{eq:outgoing_operator_expansion_through_third} is factored out.
In this section, $a_i^{\rm out}$ denotes the outgoing mode after removing
that common factor, and
\[
i=(\lambda_i,\vec k_i),\qquad
a_i\equiv a_{1,\lambda_i}^{\rm in}(\vec k_i).
\]
This changes only the perturbative bookkeeping, not the frequency
normalization of Sec.~\ref{sec:asymptotic_quantization}. The expansion is
\begin{equation}
\boxed{
a_i^{\rm out}(\epsilon)
=a_i+\epsilon\,\delta a_{2,i}^{\rm out}
+\epsilon^2\,\delta a_{3,i}^{\rm out}
+\mathcal O(\epsilon^3).
}
\label{eq:unitarity_relative_outgoing_II}
\end{equation}
The canonical conditions are
\begin{equation}
\begin{aligned}
[a_i^{\rm out},a_j^{\rm out}]
&=[a_i^{\rm in},a_j^{\rm in}]=0,\\
[a_i^{{\rm out}\dagger},a_j^{{\rm out}\dagger}]
&=[a_i^{{\rm in}\dagger},a_j^{{\rm in}\dagger}]=0,\\
[a_i^{\rm out},a_j^{{\rm out}\dagger}]
&=[a_i^{\rm in},a_j^{{\rm in}\dagger}]=\Delta(i,j),
\end{aligned}
\label{eq:unitarity_outgoing_CCR_II}
\end{equation}
where
\begin{equation}
\Delta(i,j)\equiv
2\omega_i\,\delta_{\lambda_i\lambda_j}
\delta^{(3)}(\vec k_i-\vec k_j).
\label{eq:unitarity_CCR_delta_II}
\end{equation}

Expanding these relations gives, at order $\epsilon$,
\begin{align}
[\delta a_{2,i}^{\rm out},a_j]
+[a_i,\delta a_{2,j}^{\rm out}]&=0,
\label{eq:unitarity_order2_aa_II}
\\
[\delta a_{2,i}^{\rm out},a_j^\dagger]
+[a_i,\delta a_{2,j}^{{\rm out}\dagger}]&=0.
\label{eq:unitarity_order2_mixed_II}
\end{align}
These conditions are satisfied by the quadratic NSF coefficient. At order
$\epsilon^2$, the required equations for the third-order coefficient are
\begin{equation}
\boxed{
[\delta a_{3,i}^{\rm out},a_j]
+[a_i,\delta a_{3,j}^{\rm out}]
+[\delta a_{2,i}^{\rm out},\delta a_{2,j}^{\rm out}]=0,
}
\label{eq:unitarity_order3_aa_II}
\end{equation}
\begin{equation}
\boxed{
[\delta a_{3,i}^{\rm out},a_j^\dagger]
+[a_i,\delta a_{3,j}^{{\rm out}\dagger}]
+[\delta a_{2,i}^{\rm out},\delta a_{2,j}^{{\rm out}\dagger}]=0,
}
\label{eq:unitarity_order3_mixed_II}
\end{equation}
together with the adjoint of
Eq.~\eqref{eq:unitarity_order3_aa_II}. The quadratic source terms in these
equations are completely fixed by $\delta a_2$.

For the CCR calculation we use the explicitly Bose-symmetrized form of
the outgoing coefficient, grouping the branch expression
\eqref{eq:deltaa3_retained_flatcone_II} by operator content:
\begin{equation}
\boxed{\begin{aligned}
\delta a_3^{\rm out}(i)
={}&\int_1\int_2\int_3
\delta^{(3)}(\vec k_i-\vec k_1-\vec k_2-\vec k_3)
\\
&\quad\times[\mathcal A(i;1,2,3)]_{(123)}a_1a_2a_3
\\
&+\int_1\int_2\int_3
\delta^{(3)}(\vec k_i+\vec k_1-\vec k_2-\vec k_3)
\\
&\quad\times[\mathcal B(i,1|2,3)]_{(23)}a_1^\dagger a_2a_3
\\
&+\int_1\int_2\int_3
\delta^{(3)}(\vec k_i+\vec k_1+\vec k_2-\vec k_3)
\\
&\quad\times[\mathcal C(i,1,2|3)]_{(12)}a_1^\dagger a_2^\dagger a_3
\\
&+\int_j D_1(i|j)a_j.
\end{aligned}}
\label{eq:unitarity_symmetrized_deltaa3_II}
\end{equation}
Here $\int_j=\sum_{\lambda_j}\int d\mu_j$, with the measure and rescaled oscillators defined in this paper. The cubic kernels are
obtained from the already assembled, fully normalized vertices:
\begin{align}
\mathcal A(i;1,2,3)&=\mathcal W^{(3)}_{++}(i;1,2,3),
\nonumber\\
\mathcal B(i,1|2,3)&=
\left[\mathcal W^{(3)}_{+-}(i;2,3,1)
+\mathcal W^{(3)}_{-+}(i;2,1,3)\right]_{(23)},
\nonumber\\
\mathcal C(i,1,2|3)&=\mathcal W^{(3)}_{--}(i;3,1,2).
\label{eq:unitarity_ABC_from_W_II}
\end{align}
The $\mathcal B$ kernel is the same as in
Eq.~\eqref{eq:contact_B_from_W_II}, with relabeled arguments.
Parentheses denote normalized Bose averages: $[\cdot]_{(123)}$ averages
over the six permutations of the three annihilator labels, whereas
$[\cdot]_{(23)}$ and $[\cdot]_{(12)}$ average over the two annihilator
and two creator labels, respectively. Each permutation acts on the
complete momentum--helicity label. These averages do not change the
integrated operator or introduce additional multiplicity factors.
They implement explicitly the Bose symmetrization already prescribed in
Appendix~\ref{app:representative_independence}; they do not permute a
creator with an annihilator.

\paragraph{Ordering freedom and its canonical selection.}
\emph{Direct normal ordering is a quantization choice, not a requirement
of the classical NSF equations.} The classical amplitudes commute, so
their expression does not distinguish the placements of a creator among
annihilators after quantization. Those placements differ by a term linear
in $a$ when written in normal form. The freedom to include this term is
therefore present from the outset; choosing $c=0$ initially sets it to
zero without a classical justification. The outgoing CCRs fix this
ordering freedom, selecting $c=1$ in the family considered here.
Accordingly, the canonical outgoing coefficient, expressed in normal
form, contains both cubic and linear parts:
\begin{equation}
\boxed{\delta a_{3,i}^{\rm out}
=\delta a_{3,i}^{\rm cubic}+\delta a_{3,i}^{\rm lin}.}
\label{eq:unitarity_linear_completion_II}
\end{equation}
The linear term is generated by the single Wick contractions of the
mixed cubic monomial, as displayed in
Eq.~\eqref{app:eq:ordering_c_cubic}. The symbol $c$ below labels the
ordering prescription, not the operator degree or the scalar part of a
commutator: $c=0$ retains only the normally ordered cubic monomial,
whereas $c=1$ includes its Wick-induced linear term. The expression
``c-number'' instead refers to a term proportional to the identity,
with no radiative operators left. To write the linear term explicitly, absorb the
spatial momentum distribution into the mixed kernel:
\begin{equation}
B(i,1|2,3)\equiv
\delta^{(3)}(\vec k_i+\vec k_1-\vec k_2-\vec k_3)
\mathcal B(i,1|2,3),
\label{eq:unitarity_B_distribution_II}
\end{equation}
where $\mathcal B$ is defined in
Eq.~\eqref{eq:contact_B_definition_II}. The Wick-induced linear term is
\begin{equation}
\boxed{\begin{aligned}
\delta a_{3,i}^{\rm lin}(c)
={}&\frac{c}{2}\int_1\int_2\int_3 B(i,1|2,3)
\\
&\times\left[\Delta(2,1)a_3+\Delta(3,1)a_2\right].
\end{aligned}}
\label{eq:unitarity_linear_Wick_II}
\end{equation}
Using the Bose symmetry $B(i,1|2,3)=B(i,1|3,2)$ and the contraction
rule \eqref{eq:contact_chi_II}, the two terms are equal, giving
\begin{equation}
\boxed{
\delta a_{3,i}^{\rm lin}(c)=\int_j D_c(i|j)a_j,\qquad
D_c(i|j)=c\chi\int_x B(i,x|x,j).
}
\label{eq:unitarity_D_from_B_II}
\end{equation}
Here $\int_x=\sum_{\lambda_x}\int d\mu_x$ and
$\chi=1$. All momentum delta distributions are retained.
These contractions produce a term linear in $a$, not in $a^\dagger$.
The coefficient $c$ is selected by requiring the commutator of this
linear term to cancel the c-number produced by the double Wick
contractions in the quadratic source
$[\delta a_{2,i}^{\rm out},\delta a_{2,j}^{{\rm out}\dagger}]$.
Writing $[\cdot]_0$ for the operator-degree-zero part after normal
ordering, the required cancellation is
\begin{equation}
\boxed{\begin{aligned}
[\delta a_{3,i}^{\rm lin},a_j^\dagger]
+[a_i,(\delta a_{3,j}^{\rm lin})^\dagger]
={}&-\bigl[\,[\delta a_{2,i}^{\rm out},
\delta a_{2,j}^{{\rm out}\dagger}]\,\bigr]_0.
\end{aligned}}
\label{eq:unitarity_scalar_cancellation_II}
\end{equation}
The commutators of the normally ordered cubic part with a free operator
have degree two and cannot cancel this scalar source. The canonical
choice therefore retains the Wick-induced linear term already allowed
by the ordering freedom, without changing the cubic kernel. In the ordering family of
Appendix~\ref{app:explicit_C3_kernels}, this selects the symmetric value
$c=1$; it does not introduce an independent dynamical kernel.

With this linear term included, we find that
Eqs.~\eqref{eq:unitarity_order3_aa_II} and
\eqref{eq:unitarity_order3_mixed_II} are satisfied in both the
operator-degree-zero and degree-two sectors. The cubic part is unchanged
from the $c=0$ kernel used in the irreducible loop calculation; the CCR
test applies to the complete coefficient, not to that cubic part alone.

The outgoing map also preserves Hermitian conjugation:
\begin{equation}
\boxed{(\delta a)^\dagger=\delta(a^\dagger).}
\label{eq:unitarity_adjoint_II}
\end{equation}
The separate BCH study~\cite{KozamehBCH2026} shows that, within the
polynomial Fock-algebra setting relevant here, these canonical relations
together with adjoint compatibility imply a formal perturbative unitary
relation between the outgoing and incoming operators. Moreover, it gives
an explicit recursive construction of the Hermitian generator $\delta T$
that defines the unitary operator:
\begin{equation}
\boxed{
a^{\rm out}=S^\dagger a^{\rm in}S,\qquad
S=e^{i\delta T},\qquad \delta T^\dagger=\delta T.
}
\label{eq:unitarity_reconstruction_II}
\end{equation}
These relations are understood order by order; the canonical properties
recorded here hold through $\delta a_3$. We refer to that study for the
existence proof and the constructive reconstruction, neither of which
is needed for the ultraviolet estimate in
Sec.~\ref{sec:two_to_two_one_loop_uv}.

\section{Comparison with BCFW recursion at one loop}
\label{sec:on_shell_asymptotic_connection}

The tree-level comparison between NSF and BCFW recursion was discussed
in Part~I~\cite{KDI2026}. Here we focus on the additional issue raised by
the loop sector calculated in Sec.~\ref{sec:two_to_two_one_loop_uv}.
BCFW reconstructs amplitudes from complex momentum deformations and
factorization data; NSF determines outgoing operators from the field
equations and evaluates their canonical contractions directly.

At one loop the simple tree-level recursion encounters non-rational
terms, higher-order poles in the complex deformation parameter, and
possible contributions from infinity~\cite{DunbarEttlePerkins2010}.
The coincident simple pole accompanying a double pole requires information
beyond leading factorization. Brandhuber et al.~\cite{BrandhuberLoop2007}
exhibit this difficulty for the five-graviton single-minus amplitude:
their attempted reconstruction fails permutation symmetry and some soft
and collinear limits. Augmented recursion subsequently determines the
missing terms using axial-gauge information, obtaining the five- and
six-point single-minus amplitudes~\cite{DunbarEttlePerkins2010}.
This is a limitation of the unaugmented reconstruction, not an
impossibility of computing one-loop gravity amplitudes.

The missing information is transparent in a local Laurent expansion.
Writing the complex deformation parameter as $w$, to distinguish it
from the celestial coordinate $z$, suppose
\begin{equation}
\mathcal A^{1\text{-loop}}(w)
=\frac{\alpha_P}{(w-w_P)^2}
+\frac{\beta_P}{w-w_P}+\mathcal O(1).
\label{eq:BCFW_loop_Laurent_II}
\end{equation}
Then
\begin{equation}
\mathop{\rm Res}_{w=w_P}
\frac{\mathcal A^{1\text{-loop}}(w)}{w}
=\frac{\beta_P}{w_P}-\frac{\alpha_P}{w_P^2}.
\label{eq:BCFW_loop_residue_II}
\end{equation}
Thus knowing $\alpha_P$ alone is insufficient. Separately, a nonvanishing
contour contribution at infinity must be determined rather than
discarded, and logarithmic cuts are not reconstructed by a sum of pole
residues alone.

None of these complex-deformation requirements enters the direct NSF
calculation performed here. The paired-cone equations supply the
normalized kernels before the matrix element is evaluated. The
on-shell Wick contractions then determine the internal momentum routing
and the remaining loop integration. No admissible BCFW shift, large-$w$
falloff assumption, or separate reconstruction of the coefficient
$\beta_P$ is needed. Nor must the integrated answer be a rational
function for this prescription to define the calculation.
This does not mean that the resulting amplitude cannot have physical
poles or cuts; their reconstruction from complex-deformation data is
simply not an input to the method.

The ultraviolet property is established independently by the explicit
large-momentum behavior. For the connected irreducible
$\delta a_{3,M}^{\rm out}\delta a_{3,M}^{\rm out}$ sector in the
Minkowskian approximation, the fully normalized vertices and the contracted
measure give the uniform radial tail
\begin{equation}
\int^\infty \frac{dK}{K^2}<\infty,
\label{eq:BCFW_comparison_UV_II}
\end{equation}
as derived in Sec.~\ref{sec:two_to_two_one_loop_uv}.
No ultraviolet regulator, ultraviolet subtraction, or additional
counterterm is introduced to obtain this convergence, and no
modification of the NSF equations or mode normalization is required.
The comparison therefore concerns both the absence of the above
recursion-specific requirements in this direct calculation and the
explicit radial convergence of the loop sector evaluated here.

\section{Discussion and conclusions}
\label{sec:discussion_and_conclusions}

This work identifies the additional four-leg tree-level contact term in
\(2\to2\) graviton scattering. It must be added to the \(s,t,u\) exchange
contributions obtained in Part~I~\cite{KDI2026}, completing the leading
tree-level organization in the same comparison sector.

We also construct the irreducible cubic--cubic one-loop contribution
generated by two Minkowskian third-order outgoing corrections.  Its radial
ultraviolet tail is finite without an ultraviolet counterterm, while its
bare radial infrared endpoint has the conditional quadratic power count
\(\int_0d\rho/\rho^3\).  Introducing a soft radial cutoff
\(\rho\geq\varepsilon_{\rm IR}\) provides a direct mathematical
regularization,
\begin{equation}
\int_{\varepsilon_{\rm IR}}\frac{d\rho}{\rho^3}
\sim\frac{1}{2\varepsilon_{\rm IR}^2},
\end{equation}
but this cutoff is only a regulator and not a physical resolution of the
infrared problem.  Since the divergence comes from a soft internal-graviton
endpoint in the Minkowskian benchmark, it is naturally connected with the
infrared problem of perturbative quantum gravity.  It can therefore be
regularized following the same procedure: regulate the soft region and form
an infrared-finite physical quantity either through an inclusive sum over
unresolved soft-graviton emission or through dressed asymptotic states
\cite{Weinberg1965IR,WareSaotomeAkhoury2013}.  Establishing the corresponding
cancellation in the NSF construction requires combining all degenerate real
and virtual contributions at this perturbative order.  A further possibility,
suggested by the generally nonvanishing asymptotic memory
$\Delta\sigma^{+}_{ab}$ of our states, is to relate the required soft
completion to the infrared sectors of Ashtekar's asymptotic
quantization~\cite{Ashtekar1987,Ashtekar2018}.  The reason this connection is
plausible is that nonvanishing memory permits a soft shear profile with
leading behavior
\begin{equation}
\widetilde\sigma_{ab}(\varpi)
\sim \frac{\Delta\sigma^{+}_{ab}}{-i\varpi},
\qquad \omega\longrightarrow0,
\label{eq:conclusion_soft_memory_behavior_II}
\end{equation}
so that a product of two such soft factors can behave as $1/\omega^2$.
For the massless radiative modes used throughout this work,
$\varpi=\sqrt2\,|\vec q|$, so the soft frequency is proportional to the radial loop
variable $\rho$.  Such a contribution consequently appears as
$1/\rho^2$ in radial infrared power counting and provides a possible bridge
between the Ashtekar infrared sectors and the soft endpoint found here.  This
also clarifies the different roles of the two descriptions.  In Ashtekar's
asymptotic quantization the infrared sectors arise at the kinematical level,
through the asymptotic radiative phase space and its quantum
representations.  Our calculation is instead dynamical: the soft behavior
appears inside a loop contribution to the scattering operator.  A precise
connection would therefore determine how the choice, or dressing, of an
infrared sector enters physical transition probabilities and infrared-finite
inclusive cross sections.  In this sense the issue is directly tied to the
scattering processes studied here, rather than only to the kinematical
classification of asymptotic states.  This observation does not by itself
account for the complete power of the divergence in
Eq.~\eqref{app:eq:IR_criterion_integral}, which also contains
the normalization, measure, and geometrical vertex factors.  The precise
relation between the two descriptions is delicate, and we leave its analysis
to future work.  We do not assume that
the Gaussian smearing of the external states alone performs this task, and
we postpone the complete infrared treatment and removal of the regulator to
a separate work.  Both the ultraviolet and infrared results apply
to \(\mathcal M_M^{(33)}\), in the regions analyzed in
Secs.~\ref{subsec:uv_power_counting_deltaa3} and
\ref{subsec:ir_power_counting_deltaa3}; they do not yet cover the complete
one-loop amplitude. Preservation of the CCR and Hermitian conjugation
provides the input for the formal perturbative unitary reconstruction
established in the separate BCH study~\cite{KozamehBCH2026}.

The complete canonical operator also identifies an additional
contact-type graph in \(2\to2\) scattering: a two-leg insertion attached
to a four-leg vertex. This graph is connected but one-particle reducible,
and its matrix element is not evaluated here. Its physical implications,
including its relation to decay and creation channels, are an important
immediate continuation of this work.

As noted in the conclusions of Part~I~\cite{KDI2026}, spatial momentum
already balances through the three-dimensional momentum support in the
Minkowski-based realization. The radiative energy accounting considered
here therefore focuses on Bondi mass loss, the temporal component of
the balance law. We nevertheless retain the complete four-momentum
relation \eqref{eq:Bondi_four_momentum_balance_II}, including the news
flux,
\begin{align}
P^a(u_1)&=P^a(u_2)+P_{\rm rad}^a[u_1,u_2],
\nonumber\\
P_{\rm rad}^a[u_1,u_2]
&=\frac{1}{8\pi G}\int_{u_1}^{u_2}\dd u
\oint_{S^2}\dd\Omega\,l^a(z)
\dot\sigma(u,z)\dot{\bar\sigma}(u,z).
\label{eq:outlook_explicit_radiated_momentum_II}
\end{align}
Here \(l^a=(1,\hat{\boldsymbol r})/\sqrt2\) follows the BKR
convention~\cite{BKR2023}, and the dot denotes differentiation with
respect to \(u\). This full relation will be the balance law to verify
for the quantum scattering construction on general NSF geometries,
without restricting the geometry to perturbations of Minkowski spacetime.
Its implementation and its implications for the additional radiative
channels are reserved for the immediate follow-up work.

\appendix
\section{Sign rule, Hermitian pairs, and Bose reduction}
\label{app:sign_hermitian_tables}
This appendix fixes the branch convention used in the radiative extraction
and relates Hermitian conjugation, Bose symmetrization, and the on-shell
radiative reduction.  The sign classification is performed before any
individual Fourier component is placed on shell.

For every elementary mode we write
\begin{equation}
a_i e^{-ik_i\cdot x},
\qquad
\bar a_i e^{+ik_i\cdot x},
\qquad
k_i^0=\omega_i=|\vec k_i|,
\label{app:eq:elementary_modes_sign_rule}
\end{equation}
and introduce a sign $\epsilon_i=+1$ for the first term and
$\epsilon_i=-1$ for the second. A product of $n$ modes therefore has the
form
\begin{equation}
\mathcal A_{\boldsymbol\epsilon}
e^{-iK_{\boldsymbol\epsilon}\cdot x},
\qquad
K_{\boldsymbol\epsilon}^a =
\sum_{i=1}^{n}\epsilon_i k_i^a,
\qquad
\boldsymbol\epsilon=(\epsilon_1,\ldots,\epsilon_n).
\label{app:eq:general_K_sign_rule}
\end{equation}
At the classical level the amplitudes commute, so the placement of a
future creation operator among annihilation operators is not fixed.  For
the irreducible kernels in this paper we use the direct normal quantization
map
\begin{equation}
\bar a_{i_1}\cdots\bar a_{i_r}
a_{j_1}\cdots a_{j_s}
\longmapsto
a_{i_1}^{\dagger}\cdots a_{i_r}^{\dagger}
a_{j_1}\cdots a_{j_s},
\label{app:eq:normal_quantization_map}
\end{equation}
corresponding to \(c=0\) in
Eq.~\eqref{app:eq:ordering_c_cubic}.  The alternative orderings and the
linear contraction that they generate are displayed explicitly below; the
symmetric prescription \(c=1\) is the one relevant to the complete
canonical outgoing map. Hermitian conjugation reverses all
signs,
\begin{equation}
\boldsymbol\epsilon
\longleftrightarrow
-\boldsymbol\epsilon,
\qquad
K_{-\boldsymbol\epsilon} =
-K_{\boldsymbol\epsilon},
\qquad
\mathcal O_{-\boldsymbol\epsilon} =
\mathcal O_{\boldsymbol\epsilon}^{\dagger}.
\label{app:eq:hermitian_sign_reversal}
\end{equation}
Thus the $2^n$ sign assignments are first grouped into Hermitian pairs
$\{\boldsymbol\epsilon,-\boldsymbol\epsilon\}$. Radiative extraction acts
on the complete self-adjoint cut. In the purely annihilating sector,
$K^0=\sum_i\omega_i>0$ fixes the positive-frequency branch, with the
purely creating sector supplying its Hermitian conjugate. In mixed sectors,
either member may be used to represent the complete Hermitian pair; this
choice does not independently assign positive frequency to that member.
With the branch assignment fixed by radiative extraction, the spatial
combination on the positive-frequency branch is written as
\begin{equation}
K_{\boldsymbol\epsilon}^a
\longmapsto
p_{\boldsymbol\epsilon}^a =
\left(
|\vec K_{\boldsymbol\epsilon}|,
\vec K_{\boldsymbol\epsilon}
\right),
\qquad
p_{\boldsymbol\epsilon}^2=0.
\label{app:eq:general_radiative_sign_map}
\end{equation}
The coefficient multiplying
$e^{-ip_{\boldsymbol\epsilon}\cdot x}$ contributes to
$\delta a_n^{\mathrm{out}}$, while the Hermitian-conjugate coefficient
multiplying $e^{+ip_{\boldsymbol\epsilon}\cdot x}$ contributes to
$\delta a_n^{{\mathrm{out}}\dagger}$.

\subsection{Bose symmetrization and representative independence}
\label{app:representative_independence}
Let $\mathcal R_{[\boldsymbol\epsilon]}$ denote the complete Hermitian
contribution associated with a sign pair,
\begin{align}
\mathcal R_{[\boldsymbol\epsilon]} =
\int \prod_{i=1}^{n}d\mu_i
\Big[&
\mathcal C_{\boldsymbol\epsilon}
\mathcal O_{\boldsymbol\epsilon}
e^{-iK_{\boldsymbol\epsilon}\cdot x}
\nonumber \\
&+
\mathcal C_{\boldsymbol\epsilon}^{\dagger}
\mathcal O_{\boldsymbol\epsilon}^{\dagger}
e^{+iK_{\boldsymbol\epsilon}\cdot x}
\Big].
\label{app:eq:complete_hermitian_pair}
\end{align}
For the fixed branch convention just described, the radiative reduction
of the complete pair is written as
\begin{align}
\Pi_{\mathrm{rad}}
\mathcal R_{[\boldsymbol\epsilon]} =
\int \prod_{i=1}^{n}d\mu_i
\Big[&
\mathcal C_{\boldsymbol\epsilon}
\mathcal O_{\boldsymbol\epsilon}
e^{-ip_{\boldsymbol\epsilon}\cdot x}
\nonumber \\
&+
\mathcal C_{\boldsymbol\epsilon}^{\dagger}
\mathcal O_{\boldsymbol\epsilon}^{\dagger}
e^{+ip_{\boldsymbol\epsilon}\cdot x}
\Big].
\label{app:eq:radiative_complete_pair}
\end{align}
Let $P_\pi$ denote a permutation of the complete bosonic labels
\begin{equation}
i=(\vec k_i,\lambda_i),
\qquad
\pi\in S_n,
\label{app:eq:complete_bosonic_labels}
\end{equation}
and define the Bose symmetrizer by
\begin{equation}
\mathcal S_{\mathrm B} =
\frac{1}{n!}
\sum_{\pi\in S_n}P_\pi .
\label{app:eq:Bose_symmetrizer}
\end{equation}
Because the measure
$\prod_i d\mu_i$ is permutation invariant and the map
\begin{equation}
\vec K_{\boldsymbol\epsilon}
\longmapsto
\left(
|\vec K_{\boldsymbol\epsilon}|,
\vec K_{\boldsymbol\epsilon}
\right)
\end{equation}
is covariant under permutations of the complete labels, the Bose
symmetrizer and the radiative reduction commute on the complete Hermitian
expression:
\begin{equation}
\boxed{
\Pi_{\mathrm{rad}}
\mathcal S_{\mathrm B}\mathcal R_n =
\mathcal S_{\mathrm B}
\Pi_{\mathrm{rad}}\mathcal R_n .
}
\label{app:eq:Pi_rad_Bose_commutation}
\end{equation}
Bose symmetrization may therefore be performed immediately before or after
the radiative reduction, provided that the complete Hermitian pair is
retained.

The freedom to choose a representative concerns the complete Hermitian
pair, rather than an independent positive-frequency assignment. Replacing
$\boldsymbol\epsilon$ by $-\boldsymbol\epsilon$ does not define a new
independent channel: it amounts to rewriting the same self-adjoint
contribution after Hermitian conjugation and the corresponding relabeling of
the momentum variable and its coefficient. This change must be performed
consistently on the complete pair; the opposite sign assignment must not be
treated as an additional positive-frequency contribution at fixed output
momentum. In particular, changing the representative interchanges the
descriptions of the two conjugate terms without reapplying
Eq.~\eqref{app:eq:general_radiative_sign_map} to declare the newly chosen
member positive-frequency. Self-adjointness relates the two branches,
whereas their identification as positive and negative frequency is fixed
by the radiative extraction of the complete cut.

Once the standard positive-energy convention
\begin{equation}
p^0=|\vec p|,
\qquad
e^{-ip\cdot x}
\longleftrightarrow
\delta a_n^{\mathrm{out}}(\vec p)
\label{app:eq:positive_energy_branch_convention}
\end{equation}
has been fixed, the separation between
$\delta a_n^{\mathrm{out}}$ and
$\delta a_n^{{\mathrm{out}}\dagger}$ is unambiguous. Different choices of
representative merely give different descriptions of the same
Bose-symmetrized self-adjoint cut.

\subsection{Second order}
\label{app:sign_table_deltaa2}
At second order there are four sign assignments.
Table~\ref{tab:deltaa2_signs} displays all of them before Hermitian
reduction.

\begin{table}[h]
\caption{Second-order sign assignments and Hermitian partners.}
\label{tab:deltaa2_signs}
\centering
\small
\begin{tabular}{c c c c}
\toprule
$(\epsilon_1,\epsilon_2)$
& $K_{\epsilon_1\epsilon_2}$
& normally ordered monomial
& Hermitian partner \\
\midrule
$(+,+)$
& $k_1+k_2$
& $a_1a_2$
& $(-,-)$ \\
$(+,-)$
& $k_1-k_2$
& $a_2^\dagger a_1$
& $(-,+)$ \\
$(-,+)$
& $-k_1+k_2$
& $a_1^\dagger a_2$
& $(+,-)$ \\
$(-,-)$
& $-k_1-k_2$
& $a_1^\dagger a_2^\dagger$
& $(+,+)$ \\
\bottomrule
\end{tabular}
\end{table}

The two Hermitian pairs are therefore
\begin{equation}
(+,+)\longleftrightarrow(-,-),
\qquad
(+,-)\longleftrightarrow(-,+).
\label{app:eq:deltaa2_hermitian_pairs}
\end{equation}
At this order the commutation relation
\eqref{app:eq:Pi_rad_Bose_commutation} can be seen directly. In the sum
channel,
\begin{equation}
P_{12}
\left(
\vec k_1+\vec k_2
\right) =
\vec k_1+\vec k_2,
\label{app:eq:deltaa2_sum_exchange}
\end{equation}
whereas in the difference channel,
\begin{equation}
P_{12}
\left(
\vec k_1-\vec k_2
\right) =
-\left(
\vec k_1-\vec k_2
\right).
\label{app:eq:deltaa2_difference_exchange}
\end{equation}
Thus $P_{12}$ exchanges the two members of the mixed Hermitian pair. Since
$d\mu_1d\mu_2$ is invariant under $1\leftrightarrow2$, the exchange is only
a relabeling of dummy integration variables in the complete self-adjoint
expression. Therefore,
\begin{equation}
\Pi_{\mathrm{rad}}
\mathcal S_{12}\mathcal R_2 =
\mathcal S_{12}
\Pi_{\mathrm{rad}}\mathcal R_2,
\qquad
\mathcal S_{12} =
\frac{1}{2}
\left(
1+P_{12}
\right).
\label{app:eq:deltaa2_sym_projection_commute}
\end{equation}
Thus, at second order, symmetrization before or after the radiative reduction
gives the same operator-valued cut.

Choosing $(+,+)$ and $(+,-)$ as convenient representatives gives the sum
and difference channels,
\begin{equation}
\vec K_{\mathrm{sum}} =
\vec k_1+\vec k_2,
\qquad
\vec K_{\mathrm{diff}} =
\vec k_1-\vec k_2.
\label{app:eq:deltaa2_two_classes}
\end{equation}
The second ordering of the difference channel is obtained through
$1\leftrightarrow2$ and is already included by Bose symmetrization. Hence
\begin{equation}
4\ \text{sign assignments}
\longrightarrow
2\ \text{Hermitian pairs}
\longrightarrow
2\ \text{Bose classes}.
\label{app:eq:deltaa2_counting}
\end{equation}
This is the precise second-order rule that is generalized below.

\subsection{Third order}
\label{app:sign_table_deltaa3}
At third order there are eight sign assignments. They are displayed in
Table~\ref{tab:deltaa3_signs}; the label in the last column identifies the
member obtained by simultaneous sign reversal.

\begin{table}[h]
\caption{Third-order sign assignments, normally ordered monomials, and
Hermitian partners.}
\label{tab:deltaa3_signs}
\centering
\small
\begin{tabular}{c c c c}
\toprule
$(\epsilon_1,\epsilon_2,\epsilon_3)$
& $K_{\boldsymbol\epsilon}$
& normally ordered monomial
& Hermitian partner \\
\midrule
$(+,+,+)$
& $k_1+k_2+k_3$
& $a_1a_2a_3$
& $(-,-,-)$ \\
$(+,+,-)$
& $k_1+k_2-k_3$
& $a_3^\dagger a_1a_2$
& $(-,-,+)$ \\
$(+,-,+)$
& $k_1-k_2+k_3$
& $a_2^\dagger a_1a_3$
& $(-,+,-)$ \\
$(-,+,+)$
& $-k_1+k_2+k_3$
& $a_1^\dagger a_2a_3$
& $(+,-,-)$ \\
\midrule
$(-,-,-)$
& $-k_1-k_2-k_3$
& $a_1^\dagger a_2^\dagger a_3^\dagger$
& $(+,+,+)$ \\
$(-,-,+)$
& $-k_1-k_2+k_3$
& $a_1^\dagger a_2^\dagger a_3$
& $(+,+,-)$ \\
$(-,+,-)$
& $-k_1+k_2-k_3$
& $a_1^\dagger a_3^\dagger a_2$
& $(+,-,+)$ \\
$(+,-,-)$
& $k_1-k_2-k_3$
& $a_2^\dagger a_3^\dagger a_1$
& $(-,+,+)$ \\
\bottomrule
\end{tabular}
\end{table}

The four Hermitian pairs are
\begin{align}
(+,+,+)
&\longleftrightarrow
(-,-,-),
\nonumber \\
(+,+,-)
&\longleftrightarrow
(-,-,+),
\nonumber \\
(+,-,+)
&\longleftrightarrow
(-,+,-),
\nonumber \\
(-,+,+)
&\longleftrightarrow
(+,-,-).
\label{app:eq:deltaa3_hermitian_pairs}
\end{align}
A convenient set of representatives for the negative-exponential branch is
\begin{equation}
(+,+,+),
\qquad
(+,+,-),
\qquad
(+,-,+),
\qquad
(-,+,+).
\label{app:eq:deltaa3_representatives}
\end{equation}
This is a bookkeeping convention, as established in
Sec.~\ref{app:representative_independence}. Choosing the opposite member of
a pair, together with the corresponding Hermitian conjugation and momentum
relabeling, reconstructs the same Bose-symmetrized self-adjoint cut.

The four representatives give the spatial momenta
\begin{align}
\vec K_0
&=
\vec k_1+\vec k_2+\vec k_3,
\nonumber \\
\vec K_1
&=
\vec k_1+\vec k_2-\vec k_3,
\nonumber \\
\vec K_2
&=
\vec k_1-\vec k_2+\vec k_3,
\nonumber \\
\vec K_3
&=
-\vec k_1+\vec k_2+\vec k_3.
\label{app:eq:deltaa3_four_K}
\end{align}
The first representative belongs to the three-annihilator sector. The other
three contain one creator and two annihilators and are related by
permutations of the complete bosonic labels
$i=(\vec k_i,\lambda_i)$. They must nevertheless be retained separately
while the paired-cone kernels are being assembled, because the
intermediate kernels distinguish the linear leg from the second-order block.
We carry out the calculation in the following order:
\begin{equation}
\int d\mu_1d\mu_2d\mu_3
\quad\longrightarrow\quad
\text{group the terms according to }
\vec K_0,\vec K_1,\vec K_2,\vec K_3
\quad\longrightarrow\quad
K_i^0=|\vec K_i|.
\label{app:eq:deltaa3_calculation_order}
\end{equation}
Because of Eq.~\eqref{app:eq:Pi_rad_Bose_commutation}, the final Bose
symmetrization may equivalently be performed immediately before or
immediately after the replacement $K_i^0=|\vec K_i|$.

In the Minkowskian approximation, the future-cone and antipodal-past-cone
terms are combined before the mixed Bose kernel is formed. The
recursive cut-phase term belongs to the general construction and is
absent under the operational choice \(\widehat U=1\). Thus
\begin{align}
\mathcal C_{\mathrm{mix}}^{\mathrm B} =
\operatorname{Sym}_{\mathrm B}
\Big[
\mathcal C_{12;3}
+
\mathcal C_{13;2}
+
\mathcal C_{23;1}
\Big].
\label{app:eq:deltaa3_mixed_Bose_kernel}
\end{align}
The final third-order counting is therefore
\begin{equation}
8\ \text{sign assignments}
\longrightarrow
4\ \text{Hermitian pairs}
\longrightarrow
2\ \text{Bose classes},
\label{app:eq:deltaa3_counting}
\end{equation}
namely the all-sum class and the mixed class with one relative minus sign.
For the explicit calculation, however, the four representatives in
Eq.~\eqref{app:eq:deltaa3_representatives} are retained until the two
cone contributions have been assembled and the complete
Bose-symmetrized radiative operator has been obtained.

\section{Explicit paired-cone third-order kernels}
\label{app:explicit_C3_kernels}
This appendix gives the explicit paired-cone contribution to the
third-order source.  In the Minkowskian approximation, this is the source
retained in the explicit calculation of
\(\delta a_{3,M}^{\rm out}\).  We define
\begin{equation}
\widehat Z_{3,\mathrm{cone}}^{\mathrm{pair}}
\equiv
\widehat Z_{3,\mathrm{cone}}^{+}
+
\mathcal A\widehat Z_{3,\mathrm{cone}}^{-}.
\label{app:eq:paired_cone_definition}
\end{equation}
For either cone, the source equation used here is
\begin{align}
\bar{\eth}^2\eth^2\widehat Z_{3,\mathrm{cone}}^{\pm}
=
-\int_r^\infty dr'\,
\Big[
&2\eth\bar{\eth}\delta\widehat\Omega_3^{\pm}
+
\eta^{ab}\partial_a\widehat\Lambda_1^{\pm}
\partial_b\widehat\Lambda_2^{\pm\dagger}
\nonumber\\
&+
\eta^{ab}\partial_a\widehat\Lambda_2^{\pm}
\partial_b\widehat\Lambda_1^{\pm\dagger}
\Big].
\label{app:eq:Z3cone_source}
\end{align}
After the antipodal map, both cone terms are expressed in the same angular
variables and, by the retarded Yang--Feldman prescription, in terms of the
same incoming-data functional.  Their opposite affine prescriptions give
the principal-value factor \(2\ii\,\PV[1/(l^-\!\cdot K)]\).
The inverse-metric, trajectory, and cone-measure
corrections listed in Sec.~\ref{subsec:third_order_matching_complete} are
not included in these kernels.

\subsection{Channel map and normal-ordered operators}
\label{app:subsec:C3_channel_map}

The four independent frequency branches are related to the channels of the
third-order shear calculation as follows:
\begin{equation}
\begin{array}{c|c|c|c}
(\eta_2,\eta_3)
& K_{\eta_2\eta_3}
& \text{classical channel}
& \text{normal-ordered operator}
\\ \hline
(-,+)
& k_1-k_2+k_3
& \sigma^-\bar{\sigma}^-\sigma^-
& \widehat{\mathcal{O}}_{-+}=a_{\lambda_2}^{\mathrm{in}\dagger}(k_2)
  a_{\lambda_1}^{\mathrm{in}}(k_1)a_{\lambda_3}^{\mathrm{in}}(k_3)
\\[2pt]
(-,-)
& k_1-k_2-k_3
& \sigma^-\bar{\sigma}^-\bar{\sigma}^-
& \widehat{\mathcal{O}}_{--}=a_{\lambda_2}^{\mathrm{in}\dagger}(k_2)
  a_{\lambda_3}^{\mathrm{in}\dagger}(k_3)a_{\lambda_1}^{\mathrm{in}}(k_1)
\\[2pt]
(+,-)
& k_1+k_2-k_3
& \sigma^-\sigma^-\bar{\sigma}^-
& \widehat{\mathcal{O}}_{+-}=a_{\lambda_3}^{\mathrm{in}\dagger}(k_3)
  a_{\lambda_1}^{\mathrm{in}}(k_1)a_{\lambda_2}^{\mathrm{in}}(k_2)
\\[2pt]
(+,+)
& k_1+k_2+k_3
& \sigma^-\sigma^-\sigma^-
& \widehat{\mathcal{O}}_{++}=a_{\lambda_1}^{\mathrm{in}}(k_1)
  a_{\lambda_2}^{\mathrm{in}}(k_2)a_{\lambda_3}^{\mathrm{in}}(k_3)
\end{array}
\label{app:eq:C3_channel_table}
\end{equation}
For the irreducible cubic scattering kernel used in the present paper,
Eq.~\eqref{app:eq:C3_channel_table} adopts direct normal quantization:
creation operators are written to the left of annihilation operators and
internal self-contractions of the same cubic monomial are omitted.  This is
the \(c=0\) representative of an ordering ambiguity which is invisible for
the commuting classical amplitudes but becomes relevant for the complete
quantum canonical map.

To display that ambiguity explicitly, consider a mixed monomial with one
creator and two annihilators.  Suppressing its common momentum-space kernel,
the three inequivalent placements are
\begin{align}
{\cal O}_L(1|2,3)
&=a_1^\dagger a_2a_3,
\label{app:eq:ordering_L_cubic}
\\
{\cal O}_M(1|2,3)
&=\frac12\left(a_2a_1^\dagger a_3+a_3a_1^\dagger a_2\right),
\label{app:eq:ordering_M_cubic}
\\
{\cal O}_R(1|2,3)
&=a_2a_3a_1^\dagger.
\label{app:eq:ordering_R_cubic}
\end{align}
With
\begin{equation}
\Delta_{ij}
\equiv
[a_i,a_j^\dagger]
=
2\omega_i\,\delta_{\lambda_i\lambda_j}
\delta^{(3)}(\vec k_i-\vec k_j),
\label{app:eq:Delta_ordering_cubic}
\end{equation}
their normal-ordered forms are
\begin{align}
{\cal O}_L
&=a_1^\dagger a_2a_3,
\nonumber\\
{\cal O}_M
&=a_1^\dagger a_2a_3
+\frac12\left(\Delta_{21}a_3+\Delta_{31}a_2\right),
\label{app:eq:ordering_M_normal_cubic}
\\
{\cal O}_R
&=a_1^\dagger a_2a_3
+\Delta_{21}a_3+\Delta_{31}a_2.
\label{app:eq:ordering_R_normal_cubic}
\end{align}
Equivalently,
\begin{equation}
\boxed{
{\cal O}_c
=
a_1^\dagger a_2a_3
+\frac{c}{2}
\left(\Delta_{21}a_3+\Delta_{31}a_2\right),
\qquad
c=0,1,2.
}
\label{app:eq:ordering_c_cubic}
\end{equation}
Here $c$ is an ordering parameter; its values do not label operator-degree
sectors and should not be confused with the term ``c-number.''
Thus \(c=0\) is direct normal ordering, \(c=1\) is the symmetric
prescription, and \(c=2\) corresponds to placing the creator completely to
the right before normal ordering.  The additional term is linear in the
radiative oscillators and is generated entirely by the elementary bosonic
CCR; it is not an independent dynamical vertex.

For the irreducible connected \(2\to2\) kernel analyzed in this paper we
choose
\begin{equation}
\boxed{c=0,}
\label{app:eq:c_zero_irreducible}
\end{equation}
because the linear term in
Eq.~\eqref{app:eq:ordering_c_cubic} is an internal contraction of the cubic
monomial and does not belong to the irreducible four-leg kernel.  It cannot
saturate four distinct external legs and therefore does not contribute to
the irreducible \(2\to2\) matrix element extracted from a single cubic
vertex.  For the complete canonical outgoing map, however, this linear
contraction must be retained.  The order-three Ashtekar-CCR test selects
the symmetric value \(c=1\): its scalar commutator cancels the
double-contraction source generated by
\([\delta a_2,\delta a_2^\dagger]\), consistently with the relations in
Sec.~\ref{sec:unitarity}. The distinction between
\(c=0\) for the irreducible scattering kernel and \(c=1\) for the complete
canonical map is therefore intentional and should not be interpreted as
two different classical NSF kernels.

\subsection{Second-order form factors entering the cubic kernels}
\label{app:subsec:C3_form_factors}

The difference and sum form factors inherited from $\Lambda_2^{-}$ are
\begin{align}
\overline{\mathcal{H}}_{\mathrm{diff}}(z;k_2,k_3)
&=
\overline{\mathcal{C}}_{\mathrm{kn}}(z;k_2,k_3)
+
\overline{S}^{-}_{\mathrm{diff}}(z;k_2,k_3)
+
\overline{Y^2_{lI}(z)}\,
\overline{S}^{-}_{lI,\Omega+A}(k_2,k_3),
\label{app:eq:H_diff_def}
\\
\overline{\mathcal{H}}_{\mathrm{sum}}(z;k_2,k_3)
&=
\overline{\mathcal{C}}_{\mathrm{kn,sum}}(z;k_2,k_3)
+
\overline{S}^{-}_{\mathrm{sum}}(z;k_2,k_3)
+
\overline{Y^2_{lI}(z)}\,
\overline{S}^{-}_{lI,B}(k_2,k_3).
\label{app:eq:H_sum_def}
\end{align}
The unbarred functions are defined by complex conjugation,
\begin{equation}
\mathcal{H}_{\mathrm{diff}}
\equiv
\left(\overline{\mathcal{H}}_{\mathrm{diff}}\right)^*,
\qquad
\mathcal{H}_{\mathrm{sum}}
\equiv
\left(\overline{\mathcal{H}}_{\mathrm{sum}}\right)^*.
\label{app:eq:H_unbarred_def}
\end{equation}
To keep the factors of $i$ generated by the Fourier derivatives explicit,
define
\begin{equation}
K^{(2)}_{\eta_2\eta_3}
=
\eta_2 k_2+\eta_3 k_3,
\qquad
\eta_2,\eta_3=\pm1.
\label{app:eq:K2_frequency_signs}
\end{equation}
The contact part of $\Omega_2$ contains
$(-i\eta_2)(-i\eta_3)=-\eta_2\eta_3$. The double radial integral supplies
a further minus sign, while the two second derivatives on each field remove
the frequency signs by squaring them. After the leading cone sign is
included, the conformal and scalar-product factors are
\begin{equation}
\boxed{
\begin{aligned}
F_\Omega^{(\eta_2,\eta_3)}(k_2,k_3;z)
={}&
\eta_2\eta_3
\bigl[l^{+}(z)\cdot k_2\bigr]
\bigl[l^{+}(z)\cdot k_3\bigr]
\\
&+
\frac{
\bigl[l^{+}(z)\cdot k_2\bigr]^2
\bigl[l^{+}(z)\cdot k_3\bigr]^2
}{
\bigl[l^{+}(z)\cdot K^{(2)}_{\eta_2\eta_3}\bigr]^2
}.
\end{aligned}
}
\label{app:eq:F_Omega_frequency_signs}
\end{equation}
Similarly,
\begin{equation}
\boxed{
F_\Lambda^{(\eta_2,\eta_3)}(k_2,k_3)
=
\eta_2\eta_3\,(k_2\cdot k_3).
}
\label{app:eq:F_Lambda_frequency_signs}
\end{equation}
In particular,
\begin{align}
F_\Omega^{(+,+)}
&=
+\bigl[l^{+}(z)\cdot k_2\bigr]\bigl[l^{+}(z)\cdot k_3\bigr]
+
\frac{
\bigl[l^{+}(z)\cdot k_2\bigr]^2
\bigl[l^{+}(z)\cdot k_3\bigr]^2
}{
\bigl[l^{+}(z)\cdot(k_2+k_3)\bigr]^2
},
\label{app:eq:F_Omega_sum_channel}
\\
F_\Omega^{(+,-)}
&=
-\bigl[l^{+}(z)\cdot k_2\bigr]\bigl[l^{+}(z)\cdot k_3\bigr]
+
\frac{
\bigl[l^{+}(z)\cdot k_2\bigr]^2
\bigl[l^{+}(z)\cdot k_3\bigr]^2
}{
\bigl[l^{+}(z)\cdot(k_2-k_3)\bigr]^2
}.
\label{app:eq:F_Omega_difference_channel}
\end{align}
Thus the term with one derivative acting on each first-order field changes
sign between the sum and difference sectors, whereas the term containing
two derivatives on each field does not.
The cone components entering these expressions are
\begin{align}
S^{-}_{lI,\Omega+A}(k_2,k_3)
&=
\oint d^2z\,Y^0_{lI}(z)\,
\mathrm{PV}\frac{
\eth\bar{\eth}\,S_{\Omega,\mathrm{diff}}(k_2,k_3;z)
+S_A(k_2,k_3;z)
+S_\Lambda(k_2,k_3;z)
}{l^{-}(z)\cdot(k_2-k_3)},
\label{app:eq:S_lI_OA}
\\
S^{-}_{lI,B}(k_2,k_3)
&=
\oint d^2z\,Y^0_{lI}(z)\,
\mathrm{PV}\frac{S_B(k_2,k_3;z)}{l^{-}(z)\cdot(k_2+k_3)}.
\label{app:eq:S_lI_B}
\end{align}
Here $S_A$ belongs to the difference channel and $S_B$ to the sum
channel. Their relative sign is therefore fixed by
Eq.~\eqref{app:eq:F_Lambda_frequency_signs}. The explicit second-order
angular factors are
\begin{align}
S_A(k_2,k_3;z)
&=
-\frac{l^{+}(\widehat{k}_2)\cdot l^{+}(\widehat{k}_3)}
     {l^{+}(z)\cdot(k_2-k_3)}
G_{2,+2}(z,\widehat{k}_2)G_{-2,-2}(z,\widehat{k}_3),
\label{app:eq:S_A_explicit}
\\
S_B(k_2,k_3;z)
&=
\frac{l^{+}(\widehat{k}_2)\cdot l^{+}(\widehat{k}_3)}
     {l^{+}(z)\cdot(k_2+k_3)}
\delta^2(z-\widehat{k}_2)G_{-2,-2}(z,\widehat{k}_3),
\label{app:eq:S_B_explicit}
\\
S_{\Omega,\mathrm{diff}}(k_2,k_3;z)
&=
G_{2,+2}(z,\widehat{k}_3)G_{-2,-2}(z,\widehat{k}_2)
F_\Omega^{(+,-)}(k_2,k_3;z),
\label{app:eq:S_Omega_diff_explicit}
\\
S_{\Omega,\mathrm{sum}}(k_2,k_3;z)
&=
G_{2,+2}(z,\widehat{k}_3)G_{-2,-2}(z,\widehat{k}_2)
F_\Omega^{(+,+)}(k_2,k_3;z).
\label{app:eq:S_Omega_sum_explicit}
\end{align}
The same convention is understood in
$\overline{S}^{-}_{\mathrm{sum}}$ and
$\overline{S}^{-}_{\mathrm{diff}}$: the former contains the $(+,+)$ factor
and the latter the $(+,-)$ factor.
The known-cut difference kernel is
\begin{align}
\mathcal{C}_{\mathrm{kn}}(\widehat{k}_2,\widehat{k}_3;z)
={}&
-i\varpi_3\,\delta^2(z-\widehat{k}_2)\,
Z^{-}_{1,lI}(\widehat{k}_3)Y^{-2}_{lI}(\widehat{k}_3)
\nonumber \\
&+
G_{2,0}(z,\widehat{k}_2)\,
\frac{\varpi_2\varpi_3}{\varpi_2+\varpi_3}.
\label{app:eq:C_kn_explicit}
\end{align}
The symbol $\mathcal{C}_{\mathrm{kn,sum}}$ denotes the corresponding
sum-frequency component of the same known-cut contribution. No new affine
integration is hidden in $\mathcal{C}_{\mathrm{kn}}$ or
$\mathcal{C}_{\mathrm{kn,sum}}$.

\subsection{Explicit connected kernels}
\label{app:subsec:C3_explicit}

The momentum factors in the mixed angular-field source must retain the
frequency signs of every differentiated Fourier mode. With the first leg
chosen in the annihilation sector, differentiation of the two products in
Eq.~\eqref{app:eq:Z3cone_source} gives
$-k_1\cdot(\eta_2k_2+\eta_3k_3)$ and
$-\eta_3(k_1+\eta_2k_2)\cdot k_3$, respectively. The leading minus sign
in the cone equation reverses both signs. The factors entering the connected
kernels are therefore
\begin{align}
\partial\Lambda_1\,\partial\Lambda_2^\dagger
&\longrightarrow
k_1\cdot(\eta_2 k_2+\eta_3 k_3),
\label{app:eq:Lambda12_signed_factor}
\\
\partial\Lambda_2\,\partial\Lambda_1^\dagger
&\longrightarrow
\eta_3\,(k_1+\eta_2 k_2)\cdot k_3.
\label{app:eq:Lambda21_signed_factor}
\end{align}
This follows directly from
$\partial_a e^{-i\eta k\cdot x}=-i\eta k_a e^{-i\eta k\cdot x}$.
Replacing the signed momenta by their positive-frequency magnitudes would
give the wrong relative signs between the sum and difference branches.

The channel $(-,+)$, corresponding to channel~(I), is
\begin{align}
\mathcal{T}^{(I)}_3(k_1,k_2,k_3;z)
={}&-
G_{2,0}(z,\widehat{k}_1)\,
\bigl[k_1\cdot(k_2-k_3)\bigr]
2i\mathrm{PV}\frac{1}{l^{-}(\widehat{k}_1)\cdot(k_1-k_2+k_3)}
\nonumber \\
&\times
\overline{\mathcal{H}}_{\mathrm{diff}}
(\widehat{k}_1;k_2,k_3).
\label{app:eq:T3_I}
\end{align}
The channel $(-,-)$, corresponding to channel~(II), is
\begin{align}
\mathcal{T}^{(II)}_3(k_1,k_2,k_3;z)
={}&-
G_{2,0}(z,\widehat{k}_1)\,
\bigl[k_1\cdot(k_2+k_3)\bigr]
2i\mathrm{PV}\frac{1}{l^{-}(\widehat{k}_1)\cdot(k_1-k_2-k_3)}
\nonumber \\
&\times
\overline{\mathcal{H}}_{\mathrm{sum}}
(\widehat{k}_1;k_2,k_3).
\label{app:eq:T3_II}
\end{align}
The channel $(+,-)$, corresponding to channel~(III), is
\begin{align}
\mathcal{T}^{(III)}_3(k_1,k_2,k_3;z)
={}&
\oint d^2z'\,G_{2,0}(z,z')\,
G_{2,+2}(z',\widehat{k}_1)\,
\bigl[k_1\cdot(k_2-k_3)\bigr]
2i\mathrm{PV}\frac{1}{l^{-}(z')\cdot(k_1+k_2-k_3)}
\nonumber \\
&\hspace{22mm}\times
\overline{\mathcal{H}}_{\mathrm{diff}}(z';k_2,k_3)
\nonumber \\[2pt]
&-
G_{2,0}(z,\widehat{k}_3)\,
\bigl[(k_1+k_2)\cdot k_3\bigr]
2i\mathrm{PV}\frac{1}{l^{-}(\widehat{k}_3)\cdot(k_1+k_2-k_3)}
\nonumber \\
&\hspace{22mm}\times
\mathcal{H}_{\mathrm{diff}}(\widehat{k}_3;k_1,k_2).
\label{app:eq:T3_III}
\end{align}
The channel $(+,+)$, corresponding to channel~(IV), is
\begin{align}
\mathcal{T}^{(IV)}_3(k_1,k_2,k_3;z)
={}&
\oint d^2z'\,G_{2,0}(z,z')\,
G_{2,+2}(z',\widehat{k}_1)\,
\bigl[k_1\cdot(k_2+k_3)\bigr]
2i\mathrm{PV}\frac{1}{l^{-}(z')\cdot(k_1+k_2+k_3)}
\nonumber \\
&\hspace{22mm}\times
\overline{\mathcal{H}}_{\mathrm{sum}}(z';k_2,k_3)
\nonumber \\[2pt]
&+
G_{2,0}(z,\widehat{k}_3)\,
\bigl[(k_1+k_2)\cdot k_3\bigr]
2i\mathrm{PV}\frac{1}{l^{-}(\widehat{k}_3)\cdot(k_1+k_2+k_3)}
\nonumber \\
&\hspace{22mm}\times
\mathcal{H}_{\mathrm{sum}}(\widehat{k}_3;k_1,k_2).
\label{app:eq:T3_IV}
\end{align}

The conformal-factor contribution is most compactly written for an arbitrary
branch:
\begin{align}
\mathcal{T}^{(\eta_2\eta_3)}_{3,\Omega}
(k_1,k_2,k_3;z)
={}&
\oint d^2z'\,G_{2,0}(z,z')\,
\eth'\bar{\eth}'\,
\widetilde{\delta\Omega}^{-}_{3,(\eta_2\eta_3)}
(k_1,k_2,k_3;z')
\nonumber \\
&\times
2i\mathrm{PV}\frac{1}{l^{-}(z')\cdot
K_{\eta_2\eta_3}(1,2,3)}.
\label{app:eq:T3_Omega_general}
\end{align}
Here $\widetilde{\delta\Omega}^{-}_{3,(\eta_2\eta_3)}$ denotes the
c-number Fourier coefficient multiplying the corresponding normal-ordered
operator branch of
\begin{equation}
8\delta\widehat{\Omega}_3^{-}
=
\widehat{P}_3^{-}
+
\int_r^\infty dr'\int_{r'}^\infty dr''\,
\left[
\partial_{r''}^2\widehat{\Lambda}_1^{-}\partial_{r''}^2\widehat{\Lambda}_2^{-\dagger}
+
\partial_{r''}^2\widehat{\Lambda}_2^{-}\partial_{r''}^2\widehat{\Lambda}_1^{-\dagger}
\right],
\label{app:eq:Omega3_connected}
\end{equation}
with
\begin{equation}
\widehat{P}_3^{-}
=
\partial_r\widehat{\Lambda}_1^{-}\partial_r\widehat{\Lambda}_2^{-\dagger}
+
\partial_r\widehat{\Lambda}_2^{-}\partial_r\widehat{\Lambda}_1^{-\dagger}.
\label{app:eq:P3_connected}
\end{equation}
Only the connected cubic part of Eq.~\eqref{app:eq:Omega3_connected} is
included in Eq.~\eqref{app:eq:T3_Omega_general}.

\subsection{Helicity extraction}
\label{app:subsec:C3_extraction}

For a positive-helicity outgoing graviton of momentum $q$, the reduced
vertex is
\begin{equation}
\mathcal{V}^{(3),+,\mathrm{geom},\mathrm{cone}}_{\eta_2\eta_3}(q;1,2,3)
=
-\left.
\eth_z^2
\left[
\bigl(l(z)\cdot\widehat{q}\bigr)
\mathcal{C}^{(3)}_{\eta_2\eta_3}(z;1,2,3)
\right]
\right|_{z=\widehat{q}}.
\label{app:eq:V3_plus}
\end{equation}
For negative helicity,
\begin{equation}
\mathcal{V}^{(3),-,\mathrm{geom},\mathrm{cone}}_{\eta_2\eta_3}(q;1,2,3)
=
-\left.
\bar{\eth}_z^2
\left[
\bigl(l(z)\cdot\widehat{q}\bigr)
\mathcal{C}^{(3)}_{\eta_2\eta_3}(z;1,2,3)
\right]
\right|_{z=\widehat{q}}.
\label{app:eq:V3_minus}
\end{equation}
The Cauchy--Klein--Gordon factor $(\omega_q+\Omega)/(2\omega_q)$ is not present in these equations. The outgoing phase has already been made radiative according to Sec.~\ref{subsec:radiative_projection_paired_cone}, so $q^0=|\vec{K}_{\eta_2\eta_3}|$ and the normalized radiative factor equals one. In the Minkowskian approximation the paired-cone kernel is the third-order source retained in the explicit calculation to which the angular operator is applied.
The order of operations in Eqs.~\eqref{app:eq:V3_plus} and
\eqref{app:eq:V3_minus} is essential: the angular derivatives act first and
the result is evaluated at $z=\widehat{q}$ only afterwards. In the two outgoing
insertions of the $++\to++$ matrix element, the angular evaluations are
$z=\widehat{p}'_1$ and $z=\widehat{p}'_2$, respectively. These angular
operations have radial degree zero. Coincident or collinear directions are
understood distributionally, with the radiative operators smeared against
smooth wave packets.

\subsection{Connected outgoing operator}
\label{app:subsec:C3_out_operator}

With
\begin{equation}
d\mu_i=\frac{d^3k_i}{2\omega_i},
\label{app:eq:dmu}
\end{equation}
the connected cubic correction in the Minkowskian approximation is
\begin{align}
\delta a^{\mathrm{out}}_{3,M,\lambda}(q)
={}&
\sum_{\lambda_1,\lambda_2,\lambda_3}
\sum_{\eta_2,\eta_3=\pm1}
\int d\mu_1\,d\mu_2\,d\mu_3\,
\delta^{(3)}\!\left(
\vec{q}-\vec{k}_1-\eta_2\vec{k}_2-\eta_3\vec{k}_3
\right)
\nonumber \\
&\times
\mathcal{W}^{(3),\lambda,\mathrm{cone}}_{M,\eta_2\eta_3}(q;1,2,3)
\widehat{\mathcal{O}}_{\eta_2\eta_3}(1,2,3).
\label{app:eq:deltaa3_connected}
\end{align}
The mode expansions and canonical commutator use the same normalization
as Sec.~\ref{sec:asymptotic_quantization} with the conversion to Part~III given there:
\begin{align}
\sigma^-(v,z)
&=
\int_0^\infty\frac{d\omega}{2\pi}
\sqrt{2\pi G}
\left[
 a^{\mathrm{in}}_+(\omega,\widehat{z})e^{+i\varpi v}
+a^{\mathrm{in}\dagger}_-(\omega,\widehat{z})e^{-i\varpi v}
\right],
\label{app:eq:sigma_minus_modes}
\\
\bar{\sigma}^-(v,z)
&=
\int_0^\infty\frac{d\omega}{2\pi}
\sqrt{2\pi G}
\left[
 a^{\mathrm{in}}_-(\omega,\widehat{z})e^{+i\varpi v}
+a^{\mathrm{in}\dagger}_+(\omega,\widehat{z})e^{-i\varpi v}
\right],
\label{app:eq:sigmabar_minus_modes}
\\
\left[
 a^{\mathrm{in}}_\lambda(k),
 a^{\mathrm{in}\dagger}_{\lambda'}(k')
\right]
&=
2\omega_k\,
\delta_{\lambda\lambda'}\,
\delta^{(3)}(\vec{k}-\vec{k}').
\label{app:eq:CCR}
\end{align}
Using Eq.~\eqref{app:eq:dmu}, each stripped shear mode obeys the exact
measure conversion
\begin{equation}
d^2z\,\frac{d\omega}{2\pi}
\sqrt{2\pi G}
=
d\mu(k)\,\nu(\omega),
\qquad
\nu(\omega)=\frac{\sqrt{2\pi G}}{\pi}\,\omega^{-1}.
\label{app:eq:mode_measure_conversion}
\end{equation}
The kernels $\mathcal T_3$ above are reduced kernels.  The full vertex
$\mathcal W_M^{(3)}=\mathcal N^{(3)}\mathcal V^{(3),\mathrm{geom}}$
includes the mode factors stripped in passing to $d\mu_i$, with the complete
normalization given in Eq.~\eqref{eq:N3_explicit_BCH_II}. The external factor
$\mathcal N_{\rm ext}=2\omega_q(\omega_q/\omega_1)$ is given in
Eq.~\eqref{eq:N3_external_explicit_II} and does not affect the radial degree.
In the ordered one-loop contraction the two
internal legs contribute
$\nu(\omega_k)\nu(\omega_r)=\mathcal O(K^{-2})$ at each vertex.  These
factors must be included once in $\mathcal W_M^{(3)}$ and must not be counted
again inside the reduced kernel.
The notation $a(\omega,z)$ denotes the same operator as $a(\vec k)$,
without rescaling. After a canonical contraction, the complete internal-line
weight is $d\mu(k)\nu(\omega_k)^2\propto d^3k/\omega_k^3$, as derived in
Eq.~\eqref{eq:UV_complete_internal_contraction_II} and Part~III, Sec.~V.A.
That contracted-line rule is used with reduced geometric vertices;
it is equivalent to keeping the mode weights inside $\mathcal W_M^{(3)}$.

\subsection{The connected one-loop contribution to $++\to++$}
\label{app:subsec:C3_one_loop}

In the helicity convention used in this work, the relevant matrix element is
\begin{equation}
\mathcal{A}^{(33)}_{M,++\to++}
=
\left\langle0\right|
\delta a^{\mathrm{out}}_{3,M,+}(p'_1)
\delta a^{\mathrm{out}}_{3,M,+}(p'_2)
 a^{\mathrm{in}\dagger}_{+}(p_1)
 a^{\mathrm{in}\dagger}_{+}(p_2)
\left|0\right\rangle_{\mathrm{conn}}.
\label{app:eq:A33_pppp}
\end{equation}
For this contraction argument we suppress the helicity and ``in'' labels and
write $a_i\equiv a^{\mathrm{in}}_{\lambda_i}(k_i)$. Because each cubic
monomial is already normal ordered, the connected vacuum expectation value
in the displayed operator order is supplied by
\begin{equation}
\widehat{\mathcal{O}}_{++}(1,2,3)\,
\widehat{\mathcal{O}}_{--}(1',2',3')
=
 a_1 a_2 a_3\,
 {a'}_2^\dagger {a'}_3^\dagger {a'}_1,
\label{app:eq:surviving_ordered_pair}
\end{equation}
in the displayed outgoing ordering. The mixed pairs
$\widehat{\mathcal{O}}_{+-}\widehat{\mathcal{O}}_{+-}$,
$\widehat{\mathcal{O}}_{+-}\widehat{\mathcal{O}}_{-+}$ and
$\widehat{\mathcal{O}}_{-+}\widehat{\mathcal{O}}_{-+}$ have the correct net number of creation
operators but vanish in the fixed Wightman ordering of
Eq.~\eqref{app:eq:A33_pppp}: after the right-hand mixed vertex acts on the
two-particle state, only one quantum remains, which is insufficient for the
two annihilators of the left-hand mixed vertex.

A connected contraction fixes one annihilator in
$\widehat{\mathcal{O}}_{++}$ and the annihilator in $\widehat{\mathcal{O}}_{--}$ to the two
incoming momenta. The remaining two annihilators in $\widehat{\mathcal{O}}_{++}$ are
contracted with the two creators in $\widehat{\mathcal{O}}_{--}$. Up to helicity and
momentum permutations, the routing can therefore be written as
\begin{equation}
k_1=p_i,
\qquad
k'_1=p_j,
\qquad
k_2=k'_2=k,
\qquad
k_3=k'_3=r,
\qquad
\{i,j\}=\{1,2\}.
\label{app:eq:loop_routing}
\end{equation}
The two spatial delta functions contained in the outgoing operators become
\begin{align}
\Delta_1
&=
\delta^{(3)}
\left(\vec{p}'_1-\vec{p}_i-\vec{k}-\vec{r}\right),
\label{app:eq:Delta1_loop}
\\
\Delta_2
&=
\delta^{(3)}
\left(\vec{p}'_2-\vec{p}_j+\vec{k}+\vec{r}\right).
\label{app:eq:Delta2_loop}
\end{align}
Their product is equivalently
\begin{align}
\Delta_1\Delta_2
={}&
\delta^{(3)}
\left(
\vec{p}'_1+\vec{p}'_2-\vec{p}_1-\vec{p}_2
\right)
\nonumber \\
&\times
\delta^{(3)}
\left(
\vec{P}-\vec{k}-\vec{r}
\right),
\qquad
\vec{P}\equiv\vec{p}'_1-\vec{p}_i
=\vec{p}_j-\vec{p}'_2.
\label{app:eq:two_vertex_deltas_reduced}
\end{align}
Using the second delta to set $\vec{r}=\vec{P}-\vec{k}$, one obtains
\begin{align}
\mathcal{A}^{(33)}_{M,++\to++}
={}&
\delta^{(3)}
\left(
\vec{p}'_1+\vec{p}'_2-\vec{p}_1-\vec{p}_2
\right)
\mathcal{M}^{(33)}_{M,{\rm red},++\to++},
\label{app:eq:A33_global_delta}
\\
\mathcal{M}^{(33)}_{M,{\rm red},++\to++}
={}&
6\chi^4
\sum_{\substack{i,j=1,2\\ i\neq j}}
\sum_{\lambda_k,\lambda_r}
\int
\frac{d^3k}
{2\omega_k\,2\omega_{P-k}}
\nonumber \\
&\times
\mathcal{W}^{(3),B,+}_{M,++}
\left(p'_1;p_i,k,P-k\right)
\mathcal{W}^{(3),B,+}_{M,--}
\left(p'_2;p_j,k,P-k\right).
\label{app:eq:M33_single_loop}
\end{align}
The Bose averages are defined in
Eq.~\eqref{eq:loop_Bose_averaged_vertices_II}. The factor \(6\chi^4\)
counts the three choices of the left external annihilator, the two internal
pairings, and the four integrated CCR contractions. No outgoing exchange
is added to this fixed-order matrix element. As a simple combinatorial
check, for one oscillator with \([a,a^\dagger]=1\),
\(\langle0|a^3(a^\dagger)^2a(a^\dagger)^2|0\rangle=12\): six
pairings for each of the two incoming assignments.
The unfixed three-momentum $\vec{k}$ in
Eq.~\eqref{app:eq:M33_single_loop} proves that the connected
$\delta a_{3,M}\delta a_{3,M}$ sector contains a genuine one-loop contribution.
No tree-level helicity-selection rule is used in this argument.

\subsection{Radial limits to be applied to the explicit kernels}
\label{app:subsec:C3_radial_limits}

For ultraviolet power counting, write
\begin{equation}
\vec{k}=K\widehat{k},\qquad K\equiv|\vec{k}|,
\qquad
\vec{r}=\vec{P}-K\widehat{k},
\qquad
K\longrightarrow\infty.
\label{app:eq:UV_assignment}
\end{equation}
Then
\begin{equation}
\frac{d^3k}{2\omega_k\,2\omega_r}
\sim
\frac{K^2 dK\,d^2\widehat{k}}{4K^2}
=
\frac14\,dK\,d^2\widehat{k}.
\label{app:eq:UV_measure}
\end{equation}
The branch energies relevant to Eq.~\eqref{app:eq:M33_single_loop} obey
\begin{align}
\Omega_{++}(p_i,k,r)
&=
\omega_{p_i}+\omega_k+\omega_r
=2K+\mathcal{O}(1),
\label{app:eq:Omega_pp_UV}
\\
\Omega_{--}(p_j,k,r)
&=
\omega_{p_j}-\omega_k-\omega_r
=-2K+\mathcal{O}(1).
\label{app:eq:Omega_mm_UV}
\end{align}
The following is the uniform radial limit analyzed in
Sec.~\ref{subsec:uv_power_counting_deltaa3}.
The radiative projection has radial degree zero: no additional factor
$(\omega_q+\Omega)/(2\omega_q)$ is present.  Direct insertion of
Eq.~\eqref{app:eq:UV_assignment} into the reduced geometric kernels is
governed by the numerator and the two recursive principal-value
denominators counted in
Eqs.~\eqref{eq:UV_recursive_PV_structure_II} and
\eqref{eq:UV_numerator_denominator_degrees_II}:
\begin{equation}
\mathcal V^{(3),\mathrm{geom}}=\mathcal O(K).
\label{app:eq:V3_geom_degree}
\end{equation}
Equation~\eqref{app:eq:mode_measure_conversion} supplies
$\nu(\omega_k)\nu(\omega_r)=\mathcal O(K^{-2})$ at each vertex, and
hence
\begin{equation}
\mathcal W_M^{(3)}=\mathcal O(K^{-1}),
\qquad
\mathcal W_{M,L}^{(3)}\mathcal W_{M,R}^{(3)}=\mathcal O(K^{-2}).
\label{app:eq:W3_full_degree}
\end{equation}
Together with Eq.~\eqref{app:eq:UV_measure}, this yields the absolutely
convergent radial bound in the generic uniform hard region
\begin{equation}
\int^\infty dK\,K^{-2}<\infty.
\label{app:eq:UV_convergent_bound}
\end{equation}
Thus the connected plane-wave one-loop contribution is radially ultraviolet
finite in this uniform radial region. Smearing remains part of the physical state definition and controls
separate celestial-distribution effects, but it is not required for this
radial ultraviolet bound.
For an independent bookkeeping check, the two complete contracted-line
weights give $K^2dK/K^6$ after the spatial delta is used, and the two reduced
vertices supply $K^2$. The result is again $dK/K^2$, not an additional
suppression to be multiplied into the normalized-vertex expression.

The infrared endpoint analysis is presented in
Sec.~\ref{subsec:ir_power_counting_deltaa3}, immediately after the main
ultraviolet calculation.

\section{Perturbation of the null generators from the cut function}
\label{app:geodesics_from_cut_II}

This appendix derives the first-order displacement used in
Eq.~\eqref{eq:delta_x1_from_cut_II}. It belongs to the general geometrical
construction, before the Minkowskian approximation
used for the explicit scattering benchmark.

\subsection{Perturbative inversion of the NSF coordinates}

For a fixed celestial direction \(z\), the intrinsic coordinates are
\begin{equation}
\theta^i=(Z,\eth Z,\bar\eth Z,\eth\bar\eth Z)
\equiv(u,w,\bar w,R).
\label{app:eq:geodesic_intrinsic_coordinates_II}
\end{equation}
The symbol \(w\) denotes the intrinsic angular coordinate. A null
generator is traced by varying \(R\) at fixed \(u,w,\bar w\) and \(z\).
Write
\begin{equation}
\begin{aligned}
Z(x,z)&=x^a l_a(z)+\epsilon Z_1(x,z)+O(\epsilon^2),\\
x^a(\theta,z)&=x_M^a(\theta,z)+\epsilon\xi_1^a(\theta,z)
+O(\epsilon^2).
\end{aligned}
\label{app:eq:geodesic_coordinate_inversion_II}
\end{equation}
In the flat null tetrad convention
\begin{equation}
l\cdot n=1,\qquad m\cdot\bar m=-1,\qquad
\eth l_a=m_a,\qquad \bar\eth l_a=\bar m_a,\qquad
\eth\bar\eth l_a=n_a-l_a,
\label{app:eq:geodesic_tetrad_conventions_II}
\end{equation}
the remaining tetrad contractions vanish. Expanding the four coordinate
conditions at fixed \(\theta^i\) gives
\begin{align}
\xi_1^a l_a&=-Z_1,&
\xi_1^a m_a&=-\eth Z_1,\nonumber\\
\xi_1^a\bar m_a&=-\bar\eth Z_1,&
\xi_1^a(n_a-l_a)&=-\eth\bar\eth Z_1.
\label{app:eq:geodesic_inversion_components_II}
\end{align}
All functions on the right are evaluated at \(x_M\), with the angular
derivatives taken at fixed \(x\). In particular, writing
\(\xi_1^a=A l^a+B n^a+C m^a+D\bar m^a\), the first and
fourth conditions give \(B=-Z_1\) and
\(A-B=-\eth\bar\eth Z_1\), hence
\(A=-\eth\bar\eth Z_1-Z_1\). Resolving all four contractions yields
\begin{equation}
\boxed{\xi_1^a=
-(\eth\bar\eth Z_1+Z_1)l^a-Z_1n^a
+(\bar\eth Z_1)m^a+(\eth Z_1)\bar m^a.}
\label{app:eq:geodesic_xi_from_Z1_II}
\end{equation}

\subsection{Fixed vertex and affine parameter}

Let \(x_0^a\) be the fixed vertex and let
\(x_M^a(s,z)=x_0^a+s l^a(z)\) be its background generator. Adjusting the
intrinsic-coordinate constants so that the perturbed ray starts at
\(x_0^a\) subtracts the displacement at the vertex. Therefore
\begin{equation}
\delta x_1^a(s,z)=\xi_1^a(x_0+s l,z)-\xi_1^a(x_0,z),
\qquad \delta x_1^a(0,z)=0.
\label{app:eq:geodesic_fixed_vertex_II}
\end{equation}
Using \(\Delta_sF=F(x_0+s l,z)-F(x_0,z)\) gives precisely
Eq.~\eqref{eq:delta_x1_from_cut_II}.

For the affine normalization \(dx^a/ds=g^{ab}\partial_b Z\), the NSF
identity \(g^{uR}=\Omega^2\) implies
\begin{equation}
\frac{dR}{ds}=\Omega^2.
\label{app:eq:geodesic_R_affine_II}
\end{equation}
In the vacuum hierarchy used here,
\(\Omega=1+\epsilon^2\Omega_2+O(\epsilon^3)\). With \(R(0)=R_0\),
\begin{equation}
R(s)=R_0+s+\epsilon^2\delta R_2(s)+O(\epsilon^3),\qquad
\delta R_2(s)=2\int_0^s ds'\,\Omega_2(x_0+s'l,z).
\label{app:eq:geodesic_delta_R2_II}
\end{equation}
Thus the affine correction has no first-order contribution to
\(\delta x_1^a\). Its longitudinal displacement
\(\epsilon^2\delta R_2 l^a\) is part of the second-order trajectory
correction, together with the other second-order inversion terms. Both its
shift of the source and its change of the radial integration measure first
affect the cone integral at fourth order, since that source starts at
second order.

\subsection{Contribution to the third-order cone source}

For a cone source beginning at second order,
\(\mathcal Q=\epsilon^2\mathcal Q_2+\epsilon^3\mathcal Q_3+\cdots\),
evaluation along the perturbed ray gives
\begin{equation}
\mathcal Q(x_M+\epsilon\delta x_1)
=\epsilon^2\mathcal Q_2(x_M)
+\epsilon^3\left[\mathcal Q_3(x_M)
+\delta x_1^c\partial_c\mathcal Q_2(x_M)\right]
+O(\epsilon^4).
\label{app:eq:geodesic_source_expansion_II}
\end{equation}
Consequently, for
\(I=\oint d^2z'\,G_{0,0'}\int_0^\infty ds\,\mathcal Q\), the
trajectory contribution to its third-order coefficient is
\begin{equation}
\delta I_{3,\rm geo}
=\oint d^2z'\,G_{0,0'}\int_0^\infty ds\,
\delta x_1^c(s,z')\,
\partial_c\mathcal Q_2(x_0+s l',z').
\label{app:eq:geodesic_cone_variation_II}
\end{equation}
The vertex, integration limits, and measure are held fixed in this
trajectory variation; variations of the measure and of the local metric
contractions are separate contributions. The derivative is the full
spacetime derivative of the quadratic source. All tetrad components of
\(\delta x_1^a\) are retained until its contraction with that derivative
is evaluated. The overall minus sign in the NSF cone contribution then
gives Eqs.~\eqref{eq:outgoing_geodesic_correction} and
\eqref{eq:incoming_geodesic_correction}.

Since Eq.~\eqref{app:eq:geodesic_xi_from_Z1_II} is linear in \(Z_1\),
replacing \(Z_1\) by \(\widehat Z_1\) gives the corresponding operator
displacement. Products with \(\partial_c\widehat{\mathcal Q}_2\) are
then treated with the ordering prescription of the quantum construction.

\end{document}